\documentclass[11pt]{amsart} 
\usepackage[margin= 1 in]{geometry}
\usepackage[usenames,dvipsnames]{xcolor}
\usepackage[hidelinks]{hyperref}
\hypersetup{colorlinks = true,
linkcolor = red, anchorcolor =red,
citecolor = ForestGreen,
filecolor = red,urlcolor = red,
            pdfauthor=author}
\usepackage[capitalize]{cleveref}

\usepackage{mathtools}

\usepackage{amsmath,amssymb,bm}
\usepackage[alphabetic]{amsrefs}
\usepackage{indentfirst}
\usepackage{mathrsfs}
\usepackage{graphicx}
\usepackage{amsthm}

\usepackage[T1]{fontenc}
\usepackage{newpxtext,newpxmath}
\usepackage{anyfontsize} 

\makeatletter 
\@mparswitchfalse%
\makeatother
\normalmarginpar 

\usepackage{todonotes}

\newtheorem{theorem}{Theorem}[section]
\newtheorem{lemma}[theorem]{Lemma}
\newtheorem{proposition}[theorem]{Proposition}
\newtheorem{corollary}[theorem]{Corollary}

\numberwithin{equation}{section}
\newtheorem{example}{Example}[section]

\newtheorem{assumption}{Assumption}

\theoremstyle{remark}
\newtheorem{remark}{Remark}

\theoremstyle{definition}
\newtheorem{definition}[theorem]{Definition}

\DeclareMathOperator{\tr}{tr}

\DeclareMathOperator{\ran}{Ran}

\newcommand{\ud}{\,\mathrm{d}}
\newcommand{\rd}{\mathrm{d}}

\newcommand{\R}{\mathbb{R}}
\newcommand{\C}{\mathbb{C}}
\newcommand{\E}{\mathbb{E}}
\newcommand{\hh}{\mathbb{H}}

\newcommand{\Or}{\mathcal{O}}

\newcommand{\dd}{\cdot}

\newcommand{\na}{\nabla}

\newcommand{\maf}[1]{\mathfrak{#1}}

\IfFileExists{mathabx.sty}%
  {\DeclareFontFamily{U}{mathx}{\hyphenchar\font45}%
   \DeclareFontShape{U}{mathx}{m}{n}{<->mathx10}{}%
   \DeclareSymbolFont{mathx}{U}{mathx}{m}{n}%
   \DeclareFontSubstitution{U}{mathx}{m}{n}%
   \DeclareMathAccent{\widebar}{0}{mathx}{"73}%
}{%
  \PackageWarning{mathabx}{%
    Package mathabx not available, therefore\MessageBreak substituting
    widebar with overline\MessageBreak }%
  \newcommand{\widebar}[1]{\overline{#1}}%
}
\newcommand{\wb}[1]{\widebar{#1}}

\newcommand{\mc}[1]{\mathcal{#1}}

\newcommand{\mf}[1]{\mathsf{#1}}

\newcommand{\eps}{\epsilon}
\newcommand{\lad}{\lambda}
\newcommand{\si}{\sigma}

\newcommand{\p}{\partial}
\def \h {\widehat} 
\def \ep {\varepsilon}
\def \ww {\omega}
\def \l {\langle} 
\def \r {\rangle}
\def \d {\delta}

\def \bh {\mc{B}(\mc{H})}

\def \dhh {\mc{D}_+(\mc{H})}

\def \ka {\kappa}

\def \mi {{\bf 1}}

\newcommand{\norm}[1]{\lVert#1\rVert}
\newcommand{\Norm}[1]{\left\lVert#1\right\rVert}

\newcommand{\bra}[1]{\langle#1\rvert}

\newcommand{\ket}[1]{\lvert#1\rangle}

\renewcommand{\Re}{\mathfrak{Re}}

\def \q {\quad}
\def \w {\widetilde}
\def \mm {\left[\begin{matrix}}
\def \nn {\end{matrix}\right]}

\title{Quantum space-time Poincar\'e inequality for Lindblad dynamics}

\author{Bowen Li} %
\address[B. Li]{Department of Mathematics,  City University of Hong Kong, Kowloon Tong, Hong Kong SAR}
\email{bowen.li@cityu.edu.hk}

\author{Jianfeng Lu}
\address[J. Lu]{Departments of Mathematics, Physics, and Chemistry, Duke University, Durham, NC 27708.}
\email{jianfeng@math.duke.edu}

\begin{document}

\begin{abstract}
    We investigate the mixing properties of primitive Markovian Lindblad dynamics (i.e., quantum Markov semigroups), where the detailed balance is disrupted by a coherent drift term. It is known that the sharp $L^2$-exponential convergence rate of Lindblad dynamics is determined by the spectral gap of the generator. 
    We show that incorporating a Hamiltonian component into a detailed balanced Lindbladian can generically enhance its spectral gap, thereby accelerating the mixing.  In addition, we analyze the asymptotic behavior of the spectral gap for Lindblad dynamics with a large coherent contribution. 
    However, estimating the spectral gap, particularly for a non-detailed
balanced Lindbladian, presents a significant challenge. In the case of hypocoercive Lindblad dynamics, we extend the variational framework originally developed for underdamped Langevin dynamics to derive fully explicit and constructive exponential decay estimates for convergence in the noncommutative $L^2$-norm. This analysis relies on establishing a quantum analog of space-time Poincar\'{e} inequality. Furthermore, we provide several examples with connections to quantum noise and quantum Gibbs samplers as applications of our theoretical results.
\end{abstract}

\maketitle

\section{Introduction}

Quantum systems inherently interact with their environments, leading to dissipative processes and decoherence. In the regime of weak system-environment coupling, these open quantum dynamics become Markovian and can be modeled through the Lindblad master equation, also known as quantum Markov semigroups (QMS) \cite{breuer2002theory}. The Lindblad dynamics share a structural similarity with classical underdamped Langevin equations, involving Hamiltonian evolution alongside the fluctuation-dissipation terms that account for the environmental interactions (see \cref{rem:class1}).

The Markovian assumption simplifies the mathematical treatment of open quantum dynamics and facilitates efficient numerical simulations. Consequently, the Lindblad equation is a widely applicable tool in various fields, including condensed matter physics \cites{zhang2024driven,barthel2022superoperator}, quantum information theory \cites{kastoryano2013quantum,temme2010chi}, and quantum computation \cites{verstraete2009quantum, kastoryano2011dissipative}. With the rapid advancement of quantum computers and quantum algorithms in recent years, there has been a growing interest in applying Lindblad dynamics for quantum Gibbs sampling \cites{chen2023quantum,chen2023efficient,ding2024efficient}, ground state preparation \cite{ding2023single}, and solving nonconvex optimization problems \cite{li2023quantum}. To analyze the efficiency of these Lindblad dynamics-based quantum algorithms, it is crucial to understand the mixing properties of the underlying dynamics, which is also of fundamental mathematical importance. 

The motivation of this work is two-fold. On the one hand, for large quantum unitary dynamics driven by a many-body Hamiltonian or quantum gates, the quantum decoherence (noise) arising from interaction with the environment is generally local, leading to a degenerate dissipative part of the open quantum dynamics, that is, dissipation acting only on a subspace over infinitesimal time. 
This motivates the study of the mixing of hypocoercive Lindblad dynamics (see \cref{sec:relaxation}). On the other hand, it was observed in \cite{ding2023single} that adding a coherent term to the Lindblad dynamics, which breaks the quantum detailed balance, can often accelerate the convergence towards equilibrium (a well-observed phenomenon in class sampling; see \cref{sec:related}). 

In this work, we shall provide an explicit $L^2$-decay rate estimate for the primitive hypocoercive Lindblad dynamics, motivated by the variational methods \cites{albritton2019variational, cao2023explicit} recently developed for kinetic Fokker-Planck equations. Additionally, we will provide a qualitative justification for the acceleration effect of the coherent component by analyzing the spectral gap of the dynamics.

\subsection{Main results and discussions} As reviewed in \cref{sec:related} below, significant progress has been made in understanding classical irreversible dynamics and hypocoercivity theory. However, most results concerning the convergence of quantum Markov processes assume detailed balance conditions (DBC). Results for quantum irreversible dynamics are quite limited.
Temme et al. \cite{temme2010chi} studied the convergence of the quantum channel with or without DBC by singular value analysis. Laracuente \cite{laracuente2022self} considered general open quantum dynamics involving both coherent and dissipative parts and proved the short-time polynomial decay and that the overall long-time exponential decay rate of the quantum relative entropy can be upper bounded by the inverse square root of the MLSI constant of the dissipative part. More recently, Fang et al. \cite{fang2024mixing} extended the classical DMS method \cite{dolbeault2015hypocoercivity} to analyze hypocoercive Lindblad dynamics with broken detailed balance (see the end of \cref{subsec:scaling} for a detailed review of \cite{fang2024mixing}).

In this work, we shall continue the investigation of the mixing property of hypocoercive Lindblad dynamics. We consider a general Lindbladian of the form:
\begin{align*}
    \mc{L}_\alpha = \alpha \mc{L}^H + \mc{L}^D\,,
\end{align*}
for a coupling parameter $\alpha > 0$, where $\mc{L}^H = i [H, \dd]$ is the coherent part and $\mc{L}^D$ is a Lindbladian satisfying $\si$-detailed balance. We assume that $\mc{L}_\alpha$ is primitive with $\si$ being a unique invariant state, which is of particular interest for the state preparation task. The related definitions and preliminaries for Lindblad dynamics are given in \cref{sec:strucQMS}. We connect the spectral gap and the singular value gap of $\mc{L}_\alpha$ with its $L^2$-relaxation time in \cref{sec:relaxation}. In particular, we distinguish between the coercive and hypocoercive  Lindblad dynamics by the primitivity of $\mc{L}^D$: ${\rm dim}\ker(\mc{L}^D) = 1$ (coercivity); ${\rm dim}\ker(\mc{L}^D) > 1$ (hypocoercivity); see \cref{def:coercive} and \cref{rem:coercive}.

It is well known that the spectral gap provides the sharp $L^2$-convergence rate (see \cref{sec:relaxation}). Inspired by \cites{hwang2005accelerating,franke2010behavior},
we characterize the limiting behavior of the spectral gap $\lad(\mc{L}_\alpha)$ as $\alpha \to \infty$ (\cref{thm:limitingspectral}). We also show that under some mild assumption, adding a coherent term can strictly increase the spectral gap and hence accelerate the mixing of the Lindblad dynamics (\cref{rem:accellind}). We emphasize that the results in \cref{sec:limitbeha} do not assume hypocoercive dynamics.
A very interesting but challenging question here is to construct the \emph{optimal Hamiltonian} in the sense that the spectral gap is maximized, i.e., find the optimal value and the optimizer for $\sup \{\lad(\mc{L}_\alpha)\,; [H, \si] = 0\}$ for a fixed $\alpha$.   

From \cref{sec:relaxation}, the convergence of coercive reversible or irreversible Lindblad dynamics is well understood through Poincaré or modified log-Sobolev inequalities (see \cref{rem:coercive}). Our primary focus in \cref{sec:spacetime} is the hypocoercive case, which is relatively independent of \cref{sec:limitbeha}. 
In \cref{sec:spacetime}, we extend the variational framework \cites{albritton2019variational,cao2023explicit,brigati2023construct} to explicitly quantify the exponential decay of the primitive hypocoercive (irreversible) Lindblad dynamics (\cref{thm:ratest}). This framework relies on the time-averaged $L^2$-norm, which is a natural Lyapunov functional with purely exponential decay \eqref{eq:averagedecay}. 
The key step for the dissipation of this time-averaged functional is establishing a quantum analog of space-time Poincar\'{e} inequality associated with $\mc{L}_\alpha$ for time-dependent operators (\cref{thm:tspoincare}).
It reveals that for hypocoercive dynamics, one should consider a time interval rather than a single time slice to allow dissipation to propagate from the subspace to the whole space (\cref{rem:time-avg}). 
The proof of such a space-time Poincar\'{e} inequality relies on a technical abstract divergence lemma (\cref{lem:diverest}) that constructs necessary test operators with a prior estimate. We emphasize that the convergence rate estimate in \cref{thm:ratest} is explicit so that one can carefully choose the model parameters to optimize the rate, as explored in  \cref{subsec:scaling}. As applications of results in \cref{sec:spacetime}, some examples of Davies generators with detailed balance broken by a coherent term are provided in \cref{sec:example}.  

\subsection{Related works} \label{sec:related}
The mixing time of a Markovian system is the duration required for the system to reach its stationary distribution from an arbitrary initial state. One of the most common and important approaches to bounding the mixing time is estimating the spectral gap of the Markov semigroup using $\chi^2$-divergence \cites{diaconis1991geometric,temme2010chi}. It is well known that a system-size-independent spectral gap estimate only yields a $L^1$-mixing time $t_{\rm mix} = \mc{O}(n)$, where $N = 2^n$ is the dimension of the quantum state space. In order to achieve the rapid mixing $t_{\rm mix} = \mc{O}(\log(n))$, one must utilize the modified log-Sobolev inequality (MLSI), which has been explored in \cites{majewski1996quantum,olkiewicz1999hypercontractivity,guionnet2003lectures,kastoryano2013quantum}. Since then, substantial efforts have been made to estimate the spectral gap and MLSI constants for Lindblad dynamics, particularly those related to quantum spin systems. Kastoryano and Brand\~{a}o \cite{kastoryano2016quantum} considered the Gibbs state preparation of spin Hamiltonians with commuting local terms and demonstrated that the spectral gap of Davies generators is independent of system size if and only if the Gibbs state satisfies a strong clustering of correlations. Capel et al. \cites{bardet2021modified,cuevas2019quantum,capel2020modified} investigated the MLSI for quantum spin lattice systems through quasi-factorization or approximate tensorization of the quantum relative entropy. In the recent work \cite{kochanowski2024rapid}, the rapid mixing was proved for a large class of geometrically-2-local Davies generators with commuting Hamiltonians. For the generic positivity of quantum MLSI, the notion of Ricci curvature lower bounds (geodesic convexity) has played an important role. Such a concept was pioneered by Carlen and Maas \cites{carlen2014analog,carlen2017gradient}, where they showed that a primitive QMS satisfying $\si$-GNS detailed balance condition is a gradient flow of the relative entropy with respect to some quantum Wasserstein distance. This allows us to consider the Ricci curvature of a QMS and its connection with MLSI, explored by Datta and  Rouz{\'e} \cite{datta2020relating}. 
The extension of the gradient flow structure and Ricci curvature notion 
to a general family of quantum convex relative entropies were investigated in \cite{li2023interpolation}, which gives the quantum analog of $p$-Beckner inequality \cite{adamczak2022modified}. 

\smallskip 

\noindent \emph{Quantum non-primitive models}. The results mentioned above on the long-time behaviors of Lindblad dynamics are restricted to primitive QMS with detailed balance conditions. Extending these results to the non-primitive case is closely related to the concept of environment-induced decoherence \cite{blanchard2003decoherence} (noting that for classical Markov processes, non-primitivity implies a degenerate diffusive generator). Bardet \cite{bardet2017estimating} initiated the study of mixing (decoherence) properties for non-primitive Lindblad dynamics from the spectral gap (equivalently, the non-primitive Poincar\'{e} inequality), which has emerged as an active research area in the past few years. In this setup, Gao et al. \cite{gao2020fisher} introduced a new tensor-stable MLSI: the complete modified log-Sobolev inequality (CMLSI), proving its validity through the monotonicity of Fisher information. More recently, 
Gao and Rouz\'{e} \cite{gao2021complete} showed that CMLSI holds for any finite-dimensional non-primitive GNS-detailed balanced QMS through a two-sided estimate for the relative entropy. Connections between CMLSI, entropic or geometric Ricci curvature, and gradient estimates have been explored in \cites{li2020graph,wirth2021complete,brannan2022complete,wirth2021curvature}. We also mention that the recent work \cite{gao2022complete} provided a generic asymptotically tight lower bound for the CMLSI constant using the inverse of completely bounded return time and an improved data processing inequality. 

\smallskip 

\noindent \emph{Acceleration by irreversible Markov processes}. As basic models in 
non-equilibrium statistical physics, the irreversible Markov dynamics have been studied extensively in the past few decades \cites{hairer2009hot,eckmann2000non}. Recent findings indicate that breaking reversibility (detailed balance) can often accelerate the mixing of Markov processes, which is a favorable property for the Monte Carlo Markov chain method for probability sampling. Various criteria were proposed to quantify the ergodic properties of the Markov process and compare the efficiency of the corresponding sampling method \cites{fill1991eigenvalue,chatterjee2023spectral,hwang2005accelerating,peskun1973optimum,neal2004improving,rey2015variance,rey2015irreversible,ferre2020large}. Hwang et al. \cite{hwang2005accelerating} proved that under some mild conditions, the spectral gap of a reversible diffusion can be strictly decreased by adding a divergence-free linear drift, thereby accelerating convergence to equilibrium. Further analysis by Franke et al. \cite{franke2010behavior} explored the limiting behavior of the spectral gap of the Laplace-Beltrami operator with a growing linear drift vector field on a compact Riemannian manifold. See also  \cites{fill1991eigenvalue,chatterjee2023spectral} for the spectral gap analysis for finite-state Markov chains. While the spectral property of a Markov generator provides a useful measure of ergodicity, from the Monte Carlo sampling perspective, it is more informative to consider the statistics of the ergodic average $t^{-1} \int_0^t f(X_s)\ud s$ or the empirical measure $t^{-1} \int_0^t \d_{X_s} \ud s$. For example, Dupuis et al. \cite{dupuis2012infinite} utilized the large deviation theory to compare the performance of Markov samplers and the convergence of ergodic averages for parallel tempering-type algorithms. Later, Rey-Bellet and Spiliopoulos \cites{rey2015irreversible,rey2015variance} considered irreversible Langevin samplers and demonstrated that adding a growing drift decreases the asymptotic variance of the estimator and increases the large deviation rate function of the empirical measure. We refer the readers to \cite{ferre2020large} and references therein for more recent results. 

\smallskip 

\noindent \emph{Hypocoercivity theory}.
Another framework for analyzing the convergence of irreversible Markov processes is hypocoercivity theory, which addresses models with a degenerate diffusion (elliptic) component, such as kinetic Fokker-Planck equations (\cref{rem:class1}). The foundational idea, initiated by H\'{e}rau \cites{desvillettes2001trend,herau2004isotropic,herau2006hypocoercivity} and then systematically developed by Villani \cite{villani2009hypocoercivity}, is to split the generator into reversible and irreversible parts and then quantify they interact with each other. Due to the degenerate ellipticity, the system is not coercive in the standard $L^2$-norm. To carry the dissipation from the momentum direction to the whole space, it is often necessary to twist the reference norm. The $H^1$-hypocoercivity has been considered in \cites{villani2009hypocoercivity,mouhot2006quantitative}, directly implying convergence in the $L^2$-norm. However, the resulting rate is implicit and tends to be suboptimal. A direct approach for the $L^2$-convergence was investigated by Dolbeault, Mouhot, and Schmeiser (DMS) \cites{dolbeault2015hypocoercivity,dolbeault2009hypocoercivity}, based on a modified $L^2$-norm. Its underlying idea was inspired from \cite{herau2006hypocoercivity} and revisited in \cite{grothaus2016hilbert}. A series of works \cites{arnold2014sharp,achleitner2015large,arnold2020propagator} by Arnold et al. constructed another modified norm and derived the optimal convergence rates, but it relies on algebraic properties and the spectral information of the involved diffusion operator.

A new variational hypocoercivity framework was recently proposed by Armstrong and Mourrat \cite{albritton2019variational} by using the space-time adapted Poincar\'e inequalities, without changing the $L^2$-scalar product.  It was later extended by Cao et al. \cite{cao2023explicit} for underdamped Langevin dynamics with a confining potential, and by Lu and Wang \cite{lu2022explicit} for piecewise deterministic Markov processes. Compared to the DMS entropy functional approach \cites{dolbeault2015hypocoercivity,dolbeault2009hypocoercivity}, the variational framework directly analyzes the dynamics and can potentially provide optimal estimates on the convergence rate. For instance, in the case of underdamped Langevin dynamics, assuming the Hessian of the potential is bounded from below and choosing the friction coefficient $\gamma = \mc{O}(\sqrt{m})$, 
\cite{cao2023explicit} obtained the optimal rate $\mc{O}(\sqrt{m})$ as $m \to 0$, whereas the DMS approach only yields a rate $\Or(m^{5/2})$ \cite{eberle2024non}*{Remark 29}, where $m$ is the Poincar\'e constant for the target distribution. For a more comprehensive literature review on hypocoercivity, see \cites{brigati2023construct,bernard2022hypocoercivity,cao2023explicit}. It is also worth mentioning that this variational approach with space-time Poincar\'{e} inequalities was recently rephrased and simplified by Eberle and L\"{o}rler using a second-order lift framework \cite{eberle2024non}.

\subsection*{Notation} For ease of exposition, we only consider a finite-dimensional Hilbert space, denoted by $\mc{H}$ with dimension $N$ (all the results and analysis in this work can be adapted to the case of infinite-dimensional $\mc{H}$ without essential difficulty, though some additional assumptions may be needed). Let $\mc{B}(\mc{H})$ be the space of bounded operators. 
The identity element in $\mc{B}(\mc{H})$ is denoted by $\mi_N$ (we often write $\mi$ for simplicity). 
We write $A \ge 0$ (resp., $A > 0$) for a positive semidefinite (resp., definite) operator. We denote by $\mc{D}(\mc{H}) := \{\rho \in \mc{B}(\mc{H})\,;\  \rho \ge 0\,,\ \tr \rho =1 \}$ the convex set of quantum states, and by $\mc{D}_+(\mc{H})$ the subset of full-rank states. Moreover, let $X^\dag$ be the adjoint of $X \in \mc{B}(\mc{H})$ and $\l X, Y\r = \tr (X^\dag Y)$ be the Hilbert-Schmidt (HS) inner product on $\bh$.
The adjoint of a superoperator $\Phi: \mc{B}(\mc{H}) \to \mc{B}(\mc{H})$ for HS inner product is denoted by $\Phi^\dag$ by abuse of notation. For a subset $\mc{M}$ of $\mc{B}(\mc{H})$, we write $\mc{M}^J$ as the set of ${\bf A} = (A_1, \cdots, A_J) \in \mc{M}^J$ with $A_j \in \mc{M}$, $1 \le j \le J$. The HS inner product on $\mc{M}^J$ can be naturally defined as $\l {\bf A} , {\bf B} \r = \sum_{j = 1}^J \l A_j, B_j\r$. In addition, for a given probability measure $\mu$ on some space $\mc{X}$, we define the inner product induced by $\mu$ by $\l f, g \r_{2,\mu}: = \int \bar{f} g \ud \mu$ for complex-valued functions $f,g$. The associated Hilbert space is denoted by $L^2_\mu(\mc{X})$, or simply, $\mc{L}^2_\mu$. We shall adopt the standard asymptotic notations $\Or$ and $\Theta$, where $f = \Theta(g)$ means $f = \Or (g)$ and $g = \Or(f)$.

\section{Quantum Markov semigroup and its mixing}

In \cref{sec:strucQMS}, we introduce the detailed balanced quantum Markov semigroup (QMS) with a coherent term as a special class of irreversible QMS, which would be the primary focus of this work. Then, in \cref{sec:relaxation}, we discuss several concepts for analyzing the mixing property of QMS without detailed balance conditions, allowing us to mathematically formalize the coercivity and hypocoercivity of QMS.

\subsection{Structures of quantum Markov semigroup} \label{sec:strucQMS}

We first recall preliminaries for Markovian open quantum dynamics. A quantum Markov semigroup (QMS) 
$(\mc{P}_t)_{t \ge 0}: \mc{B}(\mc{H}) \to \mc{B}(\mc{H})$ is a semigroup of completely positive and unital maps. Its generator $\mc{L}$, also called Lindbladian, is defined by 
\begin{equation*}
    \mc{L}(X): = \lim_{t \to 0} t^{-1}(\mc{P}_t (X) - X)\,,
\end{equation*}
which has the following GKSL canonical 
form \cites{Lindblad1976,GoriniKossakowskiSudarshan1976}: 
\begin{align} \label{eq:lindbladform}
     \mc{L} (X) = i [H, X] + \sum_{j \in \mc{J}} \left( L_j^\dag X L_j - \frac{1}{2}\left\{L_j^\dag L_j, X  \right\} \right),
\end{align}
with $i [H, X]$ and $\mc{L} (X) -  i [H, X]$ being its coherent and dissipative parts, respectively. Here, the operator $H$ is the Hamiltonian while $\{L_j\}_{j \in \mc{J}}$ are called jump operators. 

Let $\si \in \mc{D}_+(\mc{H})$ be a given full-rank quantum state. We define the modular operator: 
\begin{equation*}
\Delta_{\si}(X) = \si X \si^{-1}\q  \text{for $X \in \bh$}\,,
\end{equation*}
and the following family of inner products on $\bh$ with $s \in \R$: for any $X, Y \in \bh$, 
\begin{align} \label{def:s-inner}
    \l X, Y\r_{\si,s} := \tr(\si^s X^\dag  \si^{1-s} Y) = \l X, \Delta_\si^{1-s} (Y) \si \r\,,
\end{align}
where $\l \dd, \dd\r_{\si,1}$ and $\l \dd, \dd \r_{\si,1/2}$ are the GNS and KMS inner products, respectively. Then, a QMS $\mc{P}_t = e^{t \mc{L}}$ is said to satisfy the
$\si$-{\rm GNS} (resp., $\si$-{\rm KMS}) detailed balance condition (DBC) if the Lindbladian $\mc{L}$ is self-adjoint with respect to the inner product $\l\dd, \dd\r_{\si,1}$  (resp., $\l\dd, \dd\r_{\si,1/2}$). 
In analogy with the classical case, when $\mc{L}$ is detailed balanced, we also say that the QMS $\mc{P}_t$ is reversible if there is no confusion from the context. 

We recall the $\si$-weighted $p$-norm with $p \ge 1$ \cite{olkiewicz1999hypercontractivity}: for $X  \in \mc{B}(\mc{H})$, 
\begin{align*}
    \norm{X}_{p,\si} := \tr\Big(|\Gamma_\si^{1/p}(X)|^p \Big)^{1/p}\,,
\end{align*}
where $|X|: = \sqrt{X^\dag X}$ is the modulus of $X \in \bh$, and $\Gamma_\si$ is the weighting operator for $\si$:
\begin{align*}
    \Gamma_\si X := \si^{\frac{1}{2}}X\si^{\frac{1}{2}}\,.
\end{align*} 
We also introduce the space of operators with expectation zero for $\si$: 
\begin{equation} \label{def:spacezeroavg}
    \maf{H} = \left\{ X \in \mc{B}(\mc{H})\,;\ \tr(\si X) = 0 \right\},
\end{equation}
and there holds $\mc{B}(\mc{H}) = \ran({\bf 1}) \oplus \maf{H}$. Then, the operator norm of a superoperator $\Phi$ on $\bh$ is defined by
\begin{align} \label{def:operatornorm}
 \norm{\Phi}_{2 \to 2} = \sup_{X \in \bh \backslash\{0\}} \frac{\norm{\Phi(X)}_{2,\si}} {\norm{X}_{2,\si}}\,.
\end{align}
If $\maf{S} \subset \bh$ is an invariant subspace of $\Phi$, we can view $\Phi \in \mc{B}(\maf{S})$ and define the norm:
\begin{equation} \label{def:operatornorm2}
 \norm{\Phi}_{\maf{S} \to \maf{S}} = \sup_{X \in \maf{S}\backslash\{0\}} \frac{\norm{\Phi(X)}_{2,\si}} {\norm{X}_{2,\si}}\,.    
\end{equation}
To avoid confusion with the adjoint $\Phi^\dag$ for HS inner product, the adjoint of $\Phi$ for the KMS inner product $\l \dd, \dd \r_{\si,1/2}$ is denoted by $\Phi^\star$, which can be represented as 
\begin{equation*}
    \Phi^\star = \Gamma_\si^{-1} \Phi^\dag \Gamma_\si\,.
\end{equation*}
Moreover, we introduce the Dirichlet form corresponding to a Lindbladian $\mc{L}$ by 
\begin{align*}
    \mc{E}_{\mc{L}}(X,Y): = - \l X, \mc{L} Y \r_{\si,1/2}\,,\q \forall\, X, Y \in \bh\,.
\end{align*}
The next proposition shows that the KMS detailed balanced Lindbladians can be written as squares of derivations in some sense \cite{vernooij2023derivations}*{Theorem 2.5}. 

\begin{proposition} \label{prop:gradientform}
Let $\mc{L}$ be a Lindbladian satisfying $\si$-KMS DBC. Then there exist $\{V_j\}_{j = 1}^{J_D} \subset \bh$ such that $\{V_j\}_{j = 1}^{J_D} = \{V_j^\dag\}_{j = 1}^{J_D}$ and the associated Dirichlet form can be written as 
\begin{align*}
  \mc{E}_{\mc{L}}(X, Y) = \sum_{j = 1}^{J_D} \l [V_j, X], [V_j, Y] \r_{\si,1/2}\,,
\end{align*}
where $[A,B] = AB - BA$ denotes the commutator. 
\end{proposition}

In this work, we focus on the following family of irreversible QMS: for $\alpha > 0$,
\begin{align} \label{model:lalpha}
    \mc{L}_\alpha = \alpha \mc{L}^H + \mc{L}^D\,,
\end{align}
where $\alpha > 0$ is the coupling parameter, representing the relative strength between $\mc{L}^H$ and $\mc{L}^D$. 
Here, the dissipative part 
$\mc{L}^D$ is a Lindbladian satisfying $\si$-KMS DBC for some $\si \in \mc{D}_+(\mc{H})$, and the coherent term 
$\mc{L}^H(\dd) := i [H,\dd]$ is defined by a Hamiltonian $H$ satisfying $[H, \si] = 0$. It follows that $\si$ is a fixed point of the Lindblad dynamic generated by $\mc{L}^\dag_\alpha$, i.e., 
\begin{align*}
\mc{P}_t^\dag(\si) = e^{t \mc{L}_\alpha^\dag}(\si) = \si\,,\q \text{equivalently},\q \mc{L}_\alpha^\dag(\si) = 0\,.
\end{align*}
We remark that by the canonical form of KMS detailed balanced Lindbladian \cite{fagnola2007generators} (see also \cite{ding2024efficient}*{Theorem 10}), there is an additional coherent term $i [G, \dd]$ involved in $\mc{L}^D$, which should not be confused with the one $\mc{L}^H$.

 We shall limit our discussion to the case where the QMS $\mc{P}_t$ is primitive, that is, $\si$ is the unique invariant state. In this case, there holds \cite{frigerio1982long}
\begin{align} \label{eq:conver_qms}
    \lim_{t \to \infty} \mc{P}_t(X) = \tr(\si X)\mi\,,\q \forall\, X \in \bh\,.
\end{align}
It is worth noting, by $[H, \si] = 0$, 
\begin{equation*}
    \l Y, [H,X] \r_{\si,1/2} = \l Y,  H \si^{1/2} X \si^{1/2} -  \si^{1/2} X \si^{1/2} H  \r = \l [H, Y], X \r_{\si,1/2}\,,
\end{equation*}
which implies that $\mc{L}^H = i[H,\dd]$ is anti-Hermitian for the KMS inner product:
\begin{equation*}
    \mc{L}^H = - (\mc{L}^H)^\star\,,
\end{equation*}
resulting in $\mc{L}_\alpha$ being irreversible.
However, we will see in \cref{rem:dbc} that the reversibility (detailed balance) of $\mc{L}_\alpha$ can be recovered in a generalized sense. The following lemma gives useful characterizations of $\ker(\mc{L})$ and the primitivity of the Lindbladian. 
\begin{lemma}[{\cite{wolf5quantum}*{Theorem 7.2}}]\label{lem:converg}
Suppose that the QMS $\mc{P}^\dag_t = e^{t \mc{L}^\dag}$ admits a full-rank invariant state $\si$, namely, $\mc{L}^\dag(\si) = 0$, where $\mc{L}$ is of the form \eqref{eq:lindbladform}. Then,  
\begin{equation*}
    \{H, L_j, L_j^\dag\}' = \ker(\mc{L})\,,
\end{equation*}
and hence the primitivity is equivalent to 
\begin{equation}  \label{eq:irredu}
    \{H, L_j, L_j^\dag\}' = \{z \mi\,;\ z \in \C\}\,,
\end{equation}
where the commutant $\{H, L_j, L_j^\dag\}'$ is defined by the operators that commute with $L_j$, $L_j^\dag$ and $H$. 
\end{lemma}

Let $\{V_j\}_{j = 1}^{J_D}$ be the operators associated with $\mc{L}^D$ given by \cref{prop:gradientform}. We define the noncommutative partial derivative by
\begin{equation*}
    \p_j X = [V_j, X]\,,\q \text{with}\q  \p_j^\dag X = [V_j^\dag, X]\,,
\end{equation*}
and the associated gradient $\na: \bh \to \bh^{J_D}$ by 
\begin{equation*} 
    \na X = (\p_1 X, \cdots, \p_{J_D} X)\,, \q \text{for} \ X \in \bh\,.
\end{equation*}
This allows us to write $\mc{E}_{\mc{L}^D}(X, Y) = \sum_{j} \l \p_j X ,  \p_j Y \r_{\si,1/2}$ and further obtain 
\begin{align} \label{eq:rep_gen}
    \mc{L}^D (X) = - \sum_{j = 1}^{J_D} \p_{j}^\star \p_j X = - \na^\star \na X\,.
\end{align}
Here, $\p_j^\star$ is the adjoint of $\p_j$ with respect to $\l \dd, \dd  \r_{\si,1/2}$:  
\begin{align} \label{eq:adjpjkms}
    \p_{j}^\star X = \Gamma_\si^{-1}\p_j^\dag \Gamma_\si X = \Gamma_\si^{-1} [V_j^\dag, \Gamma_\si X]\,.
\end{align}
Then, one can readily reformulate the Lindbladian $\mc{L}_\alpha$ in \cref{model:lalpha} into the hypocoercive form:
\begin{align} \label{eq:irreverseqms}
    \mc{L}_\alpha(X) = i \alpha [H,X] - \na^\star \na X\,,
\end{align}
which can be regarded as a quantum analog of underdamped Langevin dynamics; see \cref{rem:class1}.  
The following lemma is a direct consequence of \eqref{eq:rep_gen} and \eqref{eq:irreverseqms}. 

\begin{lemma} \label{lem:kernel}
For Lindbladian $\mc{L}^D$ in \eqref{eq:rep_gen}, there holds
\begin{align*}
    \ker(\mc{L}^D) = \ker(\na) = \{X\,;\ \p_j X = [V_j, X] = 0 \ \text{for all}\ j \}\,.
\end{align*}
For $\mc{L}_\alpha$ in \eqref{eq:irreverseqms}, there holds 
\begin{align*}
    \ker(\mc{L}_\alpha) = \ker(\mc{L}^D) \bigcap \ker(\mc{L}^H)\,.
\end{align*}
\end{lemma}

\begin{remark}
[Quantum analog of kinetic Fokker-Planck equation]
\label{rem:class1}
We recall some basics of the classical hypocoercive dynamics to better appreciate the connection of Lindblad dynamics in the form \eqref{eq:irreverseqms} with kinetic Fokker-Planck equations. Let $L$ be a generator for a Markov process with $\mu$ being the unique invariant measure. The Markov evolution $e^{t L}$ is (formally) in the hypocoercive form if 
\begin{equation} \label{eq:hypocoer}
    L = B - A^*A\,,
\end{equation}
for some operators $A, B$, where $B = - B^*$ is anti-symmetric. Here $A^*$ and $B^*$ denote the adjoints with respect to $\l \dd, \dd\r_{2,\mu}$. We refer the readers to \cite{villani2009hypocoercivity}*{Part I} for the rigorous definition of a hypocoercive operator. A well-studied hypocoercive model is the kinetic Fokker-Planck equation: given the potential energy $U(x)$ and the friction coefficient $\gamma > 0$, 
\begin{align*}
    \p_t \rho(t,x,v) = - v \dd \na_x \rho + \na_x U \dd \na_v \rho + \gamma(\Delta_v + \na_v\dd(v \rho))\,,
\end{align*}
where $\rho(t,\dd)$ denotes a distribution on $(x,v) \in \R^{2d}$. Under mild assumptions, it has a unique invariant probability measure: 
\begin{equation} \label{eq:invariantclass}
    \rd \rho_\infty(x, v) = \ud \mu(x) \ud \kappa(v)\,,
\end{equation}
where $\rd \mu = e^{-U(x)} \ud x/Z_U$ and $\rd \kappa = e^{-|v|^2/2} \ud v/(2\pi)^{d/2}$. Its corresponding backward Kolmogorov equation has the generator
$\mc{L}_{\rm kFP} = \mc{L}_{\rm ham} + \gamma \mc{L}_{\rm FD}$ with 
\begin{equation*}
    \mc{L}_{\rm ham} : = \na_v^* \na_x - \na_x^* \na_v\,,\q \mc{L}_{\rm FD}: = - \na_v^* \na_v\,.
\end{equation*}
Here $\na_x^*$ and $\na_v^*$ are the adjoints of $\na_x$ and $\na_v$ for the inner product induced by $\rho_\infty$. 

By above discussions, one can readily see that both $\mc{L}_\alpha$ and $\mc{L}_{\rm kFP}$ can fit into the form \eqref{eq:hypocoer}, and $\mc{L}_\alpha$ can be naturally viewed as a quantum analog of $\mc{L}_{\rm kFP}$. We also want to mention the recent work \cite{galkowski2024classical} on the correspondence between the Lindblad and Fokker-Planck evolutions, which shows that if $H$ and $L_j$ for $\mc{L}$ in \eqref{eq:lindbladform} are given by the semiclassical quantizations of classical observables, then the leading part of the semiclassical expansion of $\exp(t \mc{L})$ gives $\exp(t \mc{L}_{\rm kFP})$.
\end{remark}

\begin{remark}[Generalized detailed balance and time reversal] \label{rem:dbc}
We shall see that the reversibility of our focus $\mc{L}_\alpha$ can be recovered by a generalized 
detailed balance, as in the kinetic Fokker-Planck case. 
Let $\mu$ be a probability on
$\R^d$ and $\mc{F}$ be an involution (i.e., a bounded operator with $\mc{F}^2 = 1$) leaving $\mu$ invariant: $\mu \circ \mc{F} = \mu$. A Markov semigroup $e^{t \mc{L}}$ is generalized reversible with respect to $\mu$ up to $\mc{F}$ if and only if 
\begin{align} \label{eq:classgdbc}
    \mc{L}^* = \mc{L}(f\circ \mc{F})\circ \mc{F}\,,
\end{align}
where $\mc{L}^*$ is the adjoint of $\mc{L}$ for $L^2_\mu$. When $\mc{F}$ is the identity, the above definition reduces to the usual reversibility. One can show that the kinetic Fokker-Planck generator $\mc{L}_{\rm kFP}$ (see \cref{rem:class1}) is irreversible in the usual sense but satisfies the generalized reversibility with respect to $\rho_\infty$ up to $\mc{F}(x,v) = (x,-v)$ \cite{stoltz2010free}.

The quantum DBC (reversibility) can be generalized in various ways, which is typically formulated by a certain relationship between the Lindbladian $\mc{L}$ and its adjoint with respect to some inner product. The discussion below focuses on the KMS inner product, following \cite{fagnola2010generators}; a similar discussion for the GNS inner product can be found in \cite{fagnola2008detailed} and references therein. A Lindbladian satisfies the \emph{standard DBC} if there exists a self-adjoint operator $K$ such that  \cite{derezinski2006fermi}
\begin{equation} \label{eq:gdbc}
    \mc{L} - \mc{L}^\star = 2 i [K, \dd]\,,
\end{equation}
where $\mc{L}^\star$ is the adjoint of $\mc{L}$ for the KMS inner product. It is clear that $\mc{L}_\alpha$ introduced in \eqref{eq:irreverseqms} is detailed balanced in the sense of \eqref{eq:gdbc} with $K = H$. Another generalization of quantum DBC involves a time reversal $T$ and is more similar to the classical one \eqref{eq:classgdbc}. Here a time reversal $T$ means an antiunitary operator on $\mc{H}$ (i.e., $\l T x, T y\r = \l y, x\r$ for $x,y\in \mc{H}$) such that $T^{-1} = T$. We say a Lindbladian satisfies the \emph{standard DBC with time reversal $T$} if
\begin{align} \label{eq:gdbc2}
    \mc{L}^\star(X) = T\mc{L}(T X T^{-1}) T^{-1}\,.
\end{align}
Unfortunately, the definitions \eqref{eq:gdbc} and \eqref{eq:gdbc2} are not comparable in general. We refer the interested readers to \cite{guo2024designing} for further connections between the classical and quantum DBC. 
\end{remark}

\subsection{Spectral gap and relaxation time}\label{sec:relaxation}

Let $\mc{L}$ be a primitive Lindbladian that admits a full-rank invariant state $\si \in \dhh$. We will discuss various measures for quantifying the mixing properties of the Lindblad dynamics $\mc{P}_t = e^{t \mc{L}}$, in particular, the relations between the $L^2$-relaxation time, the spectral gap, and the singular value gap, as well as the symmetrization of $\mc{L}$.

We first recall the spectral gap of $\mc{L}$:
\begin{equation} \label{eq:spectralgap}
     \lad(\mc{L}) := \inf \left\{\Re(\lad)\,;\ \lad \in {\rm Spec}(-\mc{L}|_{\maf{H}})\right\},
\end{equation}
where the subspace $\maf{H}$ is defined in \eqref{def:spacezeroavg} and
$\mc{L}|_{\maf{H}}$ is the restriction of $\mc{L}$ on $\maf{H}$. For $\ep \in (0,1)$, the non-asymptotic $\ep$-relaxation time of $\mc{P}_t$ in the $L^2$-sense is given by
\begin{align} \label{eq:trelax}
    t_{\rm rel}(\ep) := \inf\left\{t \ge 0\,;\ \norm{\mc{P}_t X}_{2,\si} \le \ep \norm{X}_{2,\si}\ \, \text{for all}\ \, X \in \maf{H}\right\},
\end{align}
which, as a function $\ep \to t_{\rm rel}(\ep)$, is the inverse function of $t \to  \norm{\mc{P}_t}_{\maf{H} \to \maf{H}}$ (cf.\,\eqref{def:operatornorm2}, noting that $\maf{H}$ is invariant under $\mc{P}_t$). 
In the case of $\mc{L}$ satisfying $\si$-KMS DBC, it is easy to see 
\begin{equation}\label{eq:purexpo}
    \norm{\mc{P}_t}_{\maf{H} \to \maf{H}} = e^{- \lad(\mc{L}) t}\,.
\end{equation}
This implies that the relaxation time is proportional to the inverse of the spectral gap: 
\begin{equation*}
    t_{\rm rel}(\ep) = \log(\ep^{-1})/\lad(\mc{L}) \propto \lad(\mc{L})^{-1}\,.
\end{equation*}

When the detailed balance of $\mc{L}$ is broken, the pure exponential decay in \eqref{eq:purexpo} does not hold and one can only expect 
\begin{equation} \label{ex:expdecay}
      \norm{\mc{P}_t}_{\maf{H} \to \maf{H}} \le C e^{- \nu t}\,,
\end{equation}
for some $C \in [1,\infty)$ and $0 < \nu \le \lad(\mc{L})$. We define the sharp exponential decay rate by 
\begin{equation} \label{eq:sharprate}
    \nu_0 := \sup\{\nu > 0\,; \ \text{there exists $C \ge 1$ such that \eqref{ex:expdecay} holds}\}.
\end{equation}
By the semigroup theory of linear evolution equations \cite{engel2000one}*{Chapter IV}, we have 
\begin{equation} \label{eq:sharprate_gap}
   \lad(\mc{L}) = \nu_0 = - \lim_{t \to \infty} \frac{1}{t} \log  \norm{\mc{P}_t}_{\maf{H} \to \maf{H}}\,,
\end{equation}
which can be equivalently formulated as
  \begin{align*}
        \lad(\mc{L})^{-1} = \lim_{\ep \to 0} \frac{t_{\rm rel}(\ep)}{\log(\ep^{-1})}\,.
    \end{align*}
It means that in the asymptotic regime $\ep \to 0$, we still have $ t_{\rm rel}(\ep) \propto \lad(\mc{L})^{-1}$, while for finite $\ep > 0$, the inverse spectral gap only gives a lower bound of the relaxation time:
\begin{align} \label{eq:relax_gap}
 t_{\rm rel}(\ep) \ge \log(\ep^{-1})/\lad(\mc{L})\,,
\end{align}
due to the existence of the multiplicative constant $C \ge 1$ in \eqref{ex:expdecay}. 

\begin{remark}
    The supremum $\nu_0$ in \eqref{eq:sharprate}  may not always be attained. To be precise, given a primitive Lindblad dynamics $\mc{P}_t$, in general, one can only have\footnote{This happens when $\mc{L}$ has an eigenvalue $\lad \neq 0$ such that $\lad(\mc{L}) = \Re(-\lad)$ and its geometric multiplicity is strictly less than its algebraic multiplicity.}, for any $\epsilon > 0$, 
    \begin{equation*}
        \norm{\mc{P}_t}_{\maf{H} \to \maf{H}} \le C_\epsilon e^{- (\lad(\mc{L}) - \epsilon) t}\q \text{for some $C_\epsilon \ge 1$}\,,
    \end{equation*}
    see \cite{arnold2014sharp}*{Theorem 4.9}. The sharp decay rate $\nu_0 = \lad(\mc{L})$ can be recovered by allowing a time-polynomial prefactor $C(t)$: $ \norm{\mc{P}_t}_{\maf{H} \to \maf{H}} \le C(t) e^{- \lad(\mc{L}) t}$ \cite{arnold2020sharp}*{Theorem 2.8}. 
\end{remark}

Following the terminology in kinetic theory \cite{villani2009hypocoercivity}, we introduce the following definition. 

\begin{definition} \label{def:coercive}
A primitive Lindblad dynamics $\mc{P}_t = e^{t \mc{L}}$ is \emph{hypocoercive} if the estimate \eqref{ex:expdecay} holds for some $C > 1$ and $\nu > 0$, and $\mc{P}_t$ is called \emph{coercive} if \eqref{ex:expdecay} holds with $C = 1$. 
\end{definition}

\begin{remark}[Coercivity and symmetrized Lindbladian]\label{rem:coercive}
From \cref{eq:purexpo,ex:expdecay} with $C = 1$, we readily see that the primitive KMS-detailed balanced Lindbladian $\mc{L}$ is coercive. However, the converse is generally not true. In fact, by taking the derivative of $\norm{\mc{P}_t X}_{2,\si} \le e^{-\nu t} \norm{X}_{2,\si}$ for $X \in \maf{H}$ and Gr\"{o}nwall's inequality, we see that the coercivity of $\mc{L}$ is equivalent to the Poincar\'{e} inequality of the additively symmetrized Lindbladian: for $X \in \maf{H}$, 
\begin{equation*}
    \mc{E}_{\frac{\mc{L}+\mc{L}^\star}{2}}(X, X)\ge \nu \norm{X}_{2,\si}^2\,.
\end{equation*}
Moreover, let $X$ be the eigenvector of $\mc{L}$ associated with the spectral gap $\lad(\mc{L})$, that is, $\mc{L} X = \lad X$ with $\Re(\lad) = - \lad(\mc{L})$. We have $- \l X, \frac{\mc{L} + \mc{L}^\star}{2} X \r_{\si,1/2} = \lad(\mc{L}) \norm{X}_{2,\si}^2$, which implies 
\begin{equation} \label{eq:gapsym}
    \lad(\mc{L}) \ge \inf_{X \in \maf{H}\backslash \{0\}} \frac{ \mc{E}_{\frac{\mc{L}+\mc{L}^\star}{2}}(X, X)}{\norm{X}_{2,\si}^2}\,,
\end{equation}
meaning that the convergence rate from the coercivity of the symmetrized Lindbladian is generally not sharp. For the primitive model \eqref{model:lalpha}, there holds $\mc{L}_\alpha + \mc{L}_\alpha^\star = 2 \mc{L}^D$. In this case, the coercivity of $\mc{L}_\alpha$ is equivalent to that $\mc{L}^D$ is injective on $\maf{H}$, i.e., $\dim \ker(\mc{L}^D) = 1$. In other words, $\mc{L}_\alpha$ is hypocoercive if and only if $\dim \ker(\mc{L}^D) > 1$ (see \cref{ass:1} in \cref{sec:spacetime}). 

\end{remark}

To establish an upper bound on the relaxation time $t_{\rm rel}(\ep)$, we reformulate \eqref{ex:expdecay} as the following $T$-delayed version: for some $\nu > 0$ and $T \ge 0$, 
\begin{equation} \label{eq:delayedexponential}
    \norm{\mc{P}_t  X}_{2,\si} \le e^{-\nu(t-T)} \norm{X}_{2,\si}\,,\q \forall\, X \in \maf{H}\,,
\end{equation}
which can be implied by the $T$-average exponential decay with rate $\nu$ (see the proof of Theorem \ref{thm:ratest}): 
\begin{equation} \label{eq:average}
    \frac{1}{T} \int_t^{t + T} \norm{\mc{P}_s X}_{2,\si} \ud s \le e^{- \nu t}  \frac{1}{T} \int_0^{T} \norm{\mc{P}_s X}_{2,\si} \ud s\,, \q \forall\, X \in \maf{H}\,.
\end{equation}
With the above notions, it is standard to define, for $\ep \in (0,1)$,
\begin{align*}
t_0(\ep) = \inf\left\{\frac{1}{\nu}\log(\ep^{-1}) + T\,:\ \nu, T \ge 0 \ \, \text{such that \eqref{eq:delayedexponential} holds}\right\},     
\end{align*}
which is of the same order as $t_{\rm rel}(\ep)$:
\begin{equation} \label{eq:alterrex}
    t_{\rm rel}(\ep) \le t_0(\ep) \le 2 t_{\rm rel}(\ep)\,.
\end{equation}
Indeed, the left-hand side inequality of \eqref{eq:alterrex} is direct. For the right-hand side one, by definition \eqref{eq:trelax} of $t_{\rm rel}$, we see that \eqref{eq:delayedexponential} holds for $T = t_{\rm rel}(\ep)$ and $\nu = \log(\ep^{-1})/t_{\rm rel}(\ep)$: 
\begin{equation*}
      \norm{\mc{P}_t  X}_{2,\si} \le e^{-\nu(t-T)} \norm{X}_{2,\si} = \varepsilon^{\frac{t}{t_{\rm rel}(\ep)} - 1} \norm{X}_{2,\si}\,.
\end{equation*}
Then taking $t = 2 t_{\rm rel}(\ep)$ gives $ \norm{\mc{P}_t  X}_{2,\si} \le \varepsilon \norm{X}_{2,\si}$, that is, $t_0(\ep) \le 2 t_{\rm rel}(\ep)$. 
Therefore, to bound $t_{\rm rel}(\ep)$ from above, it suffices to find a lower bound of $\nu$ for a fixed $T$ such that the averaged decay \eqref{eq:average} holds, which is the main purpose of \cref{sec:spacetime}. 

The singular value estimate, corresponding to the multiplicative symmetrization of the generator, has also been widely used for analyzing the mixing property of a classical irreversible Markov process \cites{fill1991eigenvalue,chatterjee2023spectral}. Temme et al. \cite{temme2010chi} characterized the convergence of the discrete-time quantum Markov process in the $\chi^2$-divergence by the singular value gap of the quantum channel. Here, we briefly discuss the connections between the singular value gap of $\mc{L}$ and the relaxation time, which extends \cite{eberle2024non}*{Lemma 10}. We define 
\begin{align*}
    s(\mc{L}) := \lad\Big(\sqrt{\mc{L}^\star\mc{L}}\Big) = \inf_{X \in \maf{H} \backslash\{0\}} \frac{\norm{\mc{L} X}_{2,\si}}{\norm{X}_{2,\si}}\,,
\end{align*}
and we have $|\lad| \ge s(\mc{L})$ for any $\lad \in {\rm Spec}(\mc{L}|_{\maf{H}})$. It is also clear that 
\begin{align*}
\frac{1}{s(\mc{L})} = \norm{\mc{L}^{-1}}_{\maf{H} \to \maf{H}}  = \sup_{X \in \maf{H} \backslash\{0\}} \frac{\norm{\mc{L}^{-1} X}_{2,\si}}{\norm{X}_{2,\si}}\,.
\end{align*}
\begin{lemma}
Suppose that $\mc{P}_t = e^{t \mc{L}}$ is a primitive Lindblad dynamics satisfying the $T$-average exponential decay with rate $\nu$. Then it holds that $\frac{1}{\nu} + T \ge \frac{1}{s(\mc{L})}$, and hence $t_{\rm rel}(e^{-1}) \ge  \frac{1}{2 s(\mc{L})}$. 
\end{lemma}

\begin{proof}
Note that the $T$-average exponential decay yields the $2 T$-average one with the same rate: 
\begin{align*}
    \frac{1}{2T} \int_t^{t + 2T} \norm{\mc{P}_s X}_{2,\si} \ud s &\le \frac{1}{2T} \int_t^{t + T} \norm{\mc{P}_s X}_{2,\si} \ud s + \frac{1}{2T} \int_{t + T}^{t + 2T} \norm{\mc{P}_s X}_{2,\si} \ud s  \\
   & \le e^{- \nu t}  \frac{1}{2T} \int_0^{2T} \norm{\mc{P}_s X}_{2,\si} \ud s\,, \q \forall\, X \in \maf{H}\,.
\end{align*}
Then, by the $2T$-average exponential decay, we have, for $X \in \maf{H}$,
\begin{align*}
    \int_0^\infty \norm{\mc{P}_t X}_{2,\si} \ud t & = \int_0^\infty \frac{1}{2 T} \int_t^{t + 2T} \norm{\mc{P}_s X}_{2,\si} \ud s \ud t + \frac{1}{2 T} \int_0^{2 T} (2 T -t ) \norm{\mc{P}_t X}_{2,\si} \ud t \\
    & \le  \left(\int_0^\infty e^{- \nu t} \ud t + \frac{1}{2 T} \int_0^{2 T} (2 T -t ) \ud t \right) \norm{X}_{2,\si} \\
    & \le \left(\frac{1}{\nu} + T \right) \norm{X}_{2,\si}\,,
\end{align*}
where we use Fubini's theorem in the first line and $\norm{\mc{P}_t X}_{2,\si} \le \norm{X}_{2,\si}$ in the second line. 
It follows that the operator $ \int_0^\infty \mc{P}_t X \ud t$ 
for $X \in \maf{H}$ is well-defined, which is the inverse of the operator $- \mc{L}: \maf{H} \to \maf{H}$. Moreover, the following estimate holds:
\begin{equation*}
    \norm{\mc{L}^{-1} (X)}_{2,\si} \le \left(\frac{1}{\nu} + T \right) \norm{X}_{2,\si}\,,
\end{equation*}
which implies $1/s(\mc{L}) \le 1/\nu + T$. The proof is completed by \eqref{eq:alterrex}. 
\end{proof}

\section{Spectral gap under large coherent limit} \label{sec:limitbeha}

Let $\mc{L}_\alpha = \alpha \mc{L}^H + \mc{L}^D$ be a primitive Lindbladian given in \eqref{model:lalpha}. As discussed in \cref{sec:relaxation}, the sharp exponential $L^2$-convergence rate is characterized by the spectral gap $\lad(\mc{L}_\alpha)$. 
In this section, we will understand the asymptotic behavior of $\lad(\mc{L}_\alpha)$ as $\alpha \to + \infty$. 
We shall also see that adding a coherent term can accelerate the convergence of Lindblad dynamics in general.

We first reformulate the spectral gap \eqref{eq:spectralgap} of $\mc{L}_\alpha$ as follows: 
\begin{equation} \label{eq:respecgap}
    \lad(\alpha) := \inf \{\lad > 0\,;\ \exists \mu \in \R \ \, \text{such that}\ -\lad + i \mu \in {\rm Spec}(\mc{L}_\alpha|_{\maf{H}})\}\,.
\end{equation}
Suppose that the Hamiltonian $H$ has eigenvalues $\{\lad_i\}$. Then the eigenvalues of the commutator $[H,\dd]$ are given by the so-called Bohr frequencies:
\begin{equation} 
    B_H := \left\{\nu = \lad_i - \lad_j\,;\ \lad_i,\lad_j \in {\rm Spec}(H) \right\}\,.
\end{equation}
For each $\nu \in B_H$, we denote the associated eigenspace by
\begin{align*}
    \mc{B}_{\nu} := \{X \in \maf{H}\,;\ [H, X] = \nu X\}\,.
\end{align*}
For ease of exposition, we also introduce the minimal energy of $-\mc{L}^D$ restricted on $\mc{B}_\nu$: 
\begin{equation} \label{eq:restrgap}
    \lad_\nu = \inf\{\mc{E}_{\mc{L}^D}(X, X)\,;\ \norm{X}_{2,\si} = 1\,,\ X \in \mc{B}_\nu \}\,,\q \nu \in B_H\,, 
\end{equation}
and the corresponding eigenspace: 
\begin{align}  \label{eq:restrgapspace}
    \mc{M}_\nu = \left\{X \in \mc{B}_\nu\,; \ \mc{E}_{\mc{L}^D}(X, X) = \lad_\nu \norm{X}_{2,\si}^2 \right\}\,.
\end{align}
The main result of this section is \cref{thm:limitingspectral} below, which is analogous to \cite{franke2010behavior} for the spectral gap of the elliptic equation with growing drift. 

\begin{theorem} \label{thm:limitingspectral}
Let $\mc{L}_\alpha$ be a primitive Lindbladian of the form \eqref{model:lalpha} and $\lad(\alpha)$ be 
its spectral gap \eqref{eq:respecgap}. It holds that as $\alpha \to \infty$, the spectral gap $\lad(\alpha)$ converges to a finite value: 
\begin{align} \label{eq:limitspec}
    \lim_{\alpha \to \infty} \lad(\alpha) = \inf_{\nu \in B_H} \lad_\nu\,, 
\end{align}
where $\lad_\nu$ is defined in \eqref{eq:restrgap}.  
\end{theorem}

\begin{remark}
The formula \eqref{eq:limitspec} means that when $\alpha$ is large enough, the spectral gap is approximately given by $\inf_\nu \lad_\nu$, the estimation of which could be simpler than the one of $\lad(\alpha)$ if the eigendecomposition of $H$ is accessible.
The more precise asymptotic behavior of $\lad(\alpha)$ as $\alpha \to \infty$ can be derived by the higher-order perturbation theory \cite{kato2013perturbation}.
\end{remark}

We emphasize 
that since $\mc{L}_\alpha$ is primitive,
each $\lad_\nu$ is positive and hence
\begin{equation*}
     \lim_{\alpha \to \infty} \lad(\alpha) = \inf_{\nu \in B_H} \lad_\nu > 0\,.
\end{equation*}
Indeed, $\lad_\nu > 0$ for all Bohr frequencies means that there is no eigenvector of $\mc{L}^H$ lying in $\ker(\mc{L}^D)$, which, by \cite{achleitner2017multi}*{Proposition 2.6}, is equivalent to the following condition: 
\begin{equation} \label{eq:index}
    \sum_{j = 0}^{J} - (i\mc{L}^H)^j  \mc{L}^D (i\mc{L}^H)^j \ge K \, {\rm id}\,, \q \text{on}\ \maf{H}\,,
\end{equation}
for some $J \ge 0$ and $K > 0$. It can be viewed as a quantum analog of the finite rank H\"{o}rmander condition for hypoelliptic degenerate diffusion processes \cite{hormander1967hypoelliptic}. The smallest $J$ such that \eqref{eq:index} holds is called the \emph{hypocoercivity index} of $\mc{L}_\alpha$. Note also that $J = 0$ gives $- \mc{L}^D > 0$ on $\maf{H}$, that is, $\mc{L}_\alpha$ is coercive; see \cref{rem:coercive}. Thanks to \cite{achleitner2017multi}*{Lemma 2.8}, the condition \eqref{eq:index} holds if and only if all the eigenvalues $\lad$ of 
$\mc{L}_\alpha|_{\maf{H}}$ have negative real parts: $\Re(\lad) < 0$, which is guaranteed by the primitivity \eqref{eq:conver_qms} (see \cite{wolf5quantum}*{Theorem 6.7}).

\begin{corollary}[Accelerating Lindblad dynamics] \label{rem:accellind}
Let $\mc{L}^D$ be a primitive Lindbladian satisfying $\si$-KMS DBC, and define $\mc{L}_\alpha = \alpha \mc{L}^H + \mc{L}^D$ with $[H, \si] = 0$. Then, the spectral gap increases:
\begin{equation} \label{eq:gap}
    \lad(\alpha) \ge \lad(\mc{L}^D)\,.
\end{equation}
Suppose that $\mc{B}_D$ is the eigenspace of $-\mc{L}^D$ associated with the gap $\lad(\mc{L}^D)$. 
The inequality \eqref{eq:gap} strictly holds if and only if $\mc{B}_D \bigcap \mc{B}_\nu = \{0\}$ for any $\nu \in B_H$. In this case, the strict increase of the spectral gap still holds in the limit $\alpha \to \infty$:
\begin{equation} \label{eq:increaselimit}
    \lim_{\alpha \to \infty} \lad(\alpha) > \lad(\mc{L}^D)\,.
\end{equation}
\end{corollary}

\begin{proof}
The inequality \eqref{eq:gap} simply restates \cref{eq:gapsym} in \cref{rem:coercive}. Now, suppose $\lad(\alpha) = \lad(\mc{L}^D)$. Let $X_\alpha$ be an eigenvector of $\mc{L}_\alpha$ satisfying $\mc{L}_\alpha X_\alpha = \lad X_\alpha$ with $\Re(\lad) = - \lad(\alpha)= - \lad(\mc{L}^D)$. The argument for \cref{eq:gapsym}, along with Min-Max theorem, gives $\mc{L}^D(X_\alpha) = - \lad(\mc{L}^D) X_\alpha$, which implies that $X_\alpha$ is also an eigenvector of $\mc{L}^H$. This shows the \emph{if} direction, while the \emph{only if} direction is trivial. The claim \eqref{eq:increaselimit} follows from \cref{thm:limitingspectral}. 
\end{proof}

The proof of \cref{thm:limitingspectral} relies on the following two propositions. 

\begin{proposition} \label{prop:limit1}
Suppose that $\limsup_{\alpha \to \infty}\lad(\alpha) < \infty$. There holds 
\begin{align*}
        \liminf_{\alpha \to \infty}  \lad(\alpha) \ge \inf_{\nu \in B_H} \lad_\nu\,.
    \end{align*}
\end{proposition}

\begin{proposition} \label{prop:limit2}
Given $\alpha > 0$ and a Bohr frequency $\nu \in B_H$, let $\lad_\nu$ be given in \eqref{eq:restrgap} and denote by $B_\eta(-\lad_\nu + i \alpha \nu)$ the $\eta$-neighborhood of $-\lad_\nu + i \alpha \nu$ in the complex plane:
\begin{align*}
   B_\eta(-\lad_\nu + i \alpha \nu) := \{z \in \C\,;\ |z - (-\lad_\nu + i \alpha \nu)| \le \eta\}\,.
\end{align*}
Then, for any $\eta > 0$, there exists $\alpha_0$ large enough such that for all $\alpha \ge \alpha_0$, 
\begin{equation*}
    {\rm Spec}(\mc{L}_\alpha|\maf{H}) \bigcap B_\eta(- \lad_\nu + i \alpha \nu) \neq \emptyset. 
\end{equation*}
\end{proposition}

\begin{proof}[Proof of \cref{thm:limitingspectral}]
It suffices to note from \cref{prop:limit2} that for any $\nu \in B_H$, there holds 
\begin{equation*}
\limsup_{\alpha \to \infty}\lad(\alpha) \le \lad_\nu\,,
\end{equation*}
where $\lad_\nu$ is defined in \eqref{eq:restrgap}. Then the proof is readily completed by \cref{prop:limit1}.
\end{proof}

The rest of this section is devoted to the proofs of \cref{prop:limit1,prop:limit2}. We start with the easier lower bound of the limiting spectral gap (\cref{prop:limit1}). 

\begin{proof}[Proof of \cref{prop:limit1}]
Suppose that $- \lad(\alpha) + i \mu_\alpha \in \C$ is the eigenvalue of $\mc{L}_\alpha$ with $\lad(\alpha)$ being the spectral gap \eqref{eq:respecgap}, and $X^\alpha = X_1^\alpha + i X_2^\alpha \in \maf{H}$ with $\norm{X^\alpha}_{2,\si} = 1$ is the associated eigenvector:
\begin{align} \label{eq:eigen_equation}
    \mc{L}_\alpha (X_1^\alpha + i X_2^\alpha) = (-\lad(\alpha) + i \mu_\alpha ) X^\alpha\,,
\end{align}
where $X_1^\alpha, X_2^\alpha \in \bh$ are real-valued. It follows that 
\begin{align*}
    - \l X^\alpha,  \mc{L}_\alpha X^\alpha \r_{\si,1/2} = \lad(\alpha) - i \mu_\alpha\,.
\end{align*}
Since $\mc{L}^H$ is anti-Hermitian, we readily have 
\begin{equation} \label{eq:realpart}
     \mc{E}_{\mc{L}^D}(X^\alpha, X^\alpha) = - \Re\,\l X^\alpha,  \mc{L}_\alpha X^\alpha \r = \lad(\alpha)\,.
\end{equation}
In addition, by the boundedness of $\norm{X^\alpha}_{2,\si}$, there is a convergent subsequence of $\{X^\alpha\}$ (without loss of generality, still denoted by $X^\alpha$) such that 
\begin{align*}
  X_* : = \lim_{\alpha \to \infty} X^\alpha  \in \maf{H}\,.
\end{align*}
Note from \cref{eq:eigen_equation} that for any $Y \in \maf{H}$, 
\begin{align*}
  \mc{E}_{\mc{L}_\alpha}(Y,  X^\alpha) =  (\lad(\alpha) - i \mu_\alpha ) \l Y, X^\alpha\r_{\si,1/2}\,,
\end{align*}
which gives 
\begin{align*}
   - i \l Y, [H, X^\alpha] \r_{\si,1/2} + \frac{1}{\alpha} \mc{E}_{\mc{L}^D}(Y,X^\alpha) = \frac{1}{\alpha} ( \lad(\alpha) - i \mu_\alpha ) \l Y, X^\alpha\r_{\si,1/2}\,.
\end{align*}
Then, letting $\alpha \to \infty$ and recalling $\lim \sup_{\alpha \to \infty} \lad(\alpha) < \infty$, we find  
\begin{align*}
    i \l Y, [H, X_*] \r_{\si,1/2} = \left(\lim_{\alpha \to \infty} \frac{i \mu_\alpha}{\alpha}  \right)  \l Y, X_*\r_{\si,1/2}\,,
\end{align*}
which means that $\nu_* := \lim_{\alpha \to \infty} \frac{\mu_\alpha}{\alpha}$ exists and is a Bohr frequency with $X_* \in \mc{B}_{\nu_*}$:
\begin{align*}
    \mc{L}^{H} X_* = i \nu_* X_*\,.
\end{align*}
The proof is completed by \eqref{eq:realpart}:
\begin{equation*}
    \liminf_{\alpha \to \infty} \lad(\alpha) \ge \inf\{\mc{E}_{\mc{L}^D}(X, X)\,;\ \norm{X}_{2,\si} = 1\,,\ X \in \mc{B}_\nu \ \text{for some $\nu$}\}\,. \qedhere
\end{equation*}
\end{proof}

We next prove \cref{prop:limit2} on the stability of the spectrum of $\mc{L}_\alpha$. 

\begin{proof}[Proof of \cref{prop:limit2}]
We proceed by a contradiction argument. Suppose that there exists $\eta > 0$ and $\alpha_n \to \infty$ such that any element $z$ in the set 
\begin{equation*}
B_\eta(-\lad_\nu + i \alpha_n \nu) = \{ - \lad + i\mu\,;\ \lad = \lad_\nu + \ep \in \R \,,\ \mu = \alpha_n \nu + \d \in \R \,, \q \ep^2 + \d^2 \le \eta^2\}
\end{equation*}
is not an eigenvalue of $\mc{L}_\alpha|_{\maf{H}}$. We fix $Y \in \mc{M}_\nu$ \eqref{eq:respecgap} with normalization $\norm{Y}_{2,\si} = 1$. Then, for any $z = - \lad + i\mu \in B_\eta(-\lad_\nu + i \alpha_n \nu)$, there exists $X^{\alpha_n}_z \in \maf{H}$ satisfying 
\begin{align} \label{eq:resolventeq}
    (z - \mc{L}_{\alpha_n}) X^{\alpha_n}_z = Y\,,
\end{align}
since $z - \mc{L}_{\alpha_n}$ is invertible. Note that $z \in \C$ is parameterized by $(\ep,\d)$ and $-\lad_\nu + i \alpha_n \nu$.  

We first show that $\{X_z^{\alpha_n}\}$ is an uniformly bounded sequence: for small enough $\eta$, 
\begin{align} \label{claimbound}
    \limsup_{n \to \infty} \sup_{\ep^2 + \d^2  =  \eta^2} \norm{X^{\alpha_n}_z}_{2,\si} < \infty\,.
\end{align}
If it is not true, for a subsequence of $\alpha_n$ (still denoted by $\alpha_n$), there exists a sequence of  real $\ep_n, \d_n$ such that $\ep_n^2 + \d_n^2  =  \eta^2$, $z_n = - (\lad_\nu + \ep_n) + i (\alpha_n \nu + \d_n)$, and 
\begin{equation*}
  K_n: = \norm{X^{\alpha_n}_{z_n}}_{2,\si} \to \infty\q \text{as $n \to \infty$}\,.
\end{equation*}
We define $\w{X}_n = X^{\alpha_n}_{z_n}/K_n$, which is bounded and satisfies, by \eqref{eq:resolventeq}, for any $Z \in \maf{H}$, 
\begin{align} \label{eq:variationalresol}
    z_n \l Z,\w{X}_n\r_{\si,1/2} - i \alpha_n \l Z, [H, \w{X}_n] \r_{\si,1/2} + \mc{E}_{\mc{L}^D}(Z, \w{X}_n) = \frac{1}{K_n} \l Z, Y \r_{\si,1/2}\,.
\end{align}
By the boundedness of $\w{X}_n$, without loss of generality, we can assume that as $n \to \infty$, for some $\w{X}_* \in \maf{H}$ and real $\ep_*, \d_*$ satisfying $\ep_*^2 + \d_*^2  =  \eta^2$, there holds
\begin{equation*}
 \w{X}_n \to \w{X}_*\,,\q \ep_n \to \ep_*\,,\q \d_n \to \d_*\,.
\end{equation*}
Then, from \cref{eq:variationalresol}, we have 
\begin{equation*}
    \frac{z_n}{\alpha_n} \l Z,\w{X}_n\r_{\si,1/2} - i \l Z, [H, \w{X}_n] \r_{\si,1/2} + \frac{1}{\alpha_n} \mc{E}_{\mc{L}^D}(Z, \w{X}_n) = \frac{1}{\alpha_n K_n} \l Z, Y \r_{\si,1/2}\,,
\end{equation*}
and let $n \to \infty$ to obtain
\begin{equation*}
    i \nu \l Z,\w{X}_*\r_{\si,1/2} - i \l Z, [H, \w{X}_*] \r_{\si,1/2} = 0\,,
\end{equation*}
that is, $\w{X}_* \in \mc{B}_\nu$. We now  consider test operators $Z = \w{X}_n$ in \eqref{eq:variationalresol}
and take the real part: 
\begin{align*}
      - (\lad_\nu + \ep_n) \norm{\w{X}_n}_{2,\si}^2 + \mc{E}_{\mc{L}^D}(\w{X}_n, \w{X}_n) = \frac{1}{K_n} \Re\, \l \w{X}_n, Y \r_{\si,1/2}\,,
\end{align*}
which implies, by letting $n \to \infty$, 
\begin{align} \label{eq:limit_eig}
       \mc{E}_{\mc{L}^D}(\w{X}_*, \w{X}_*) =   (\lad_\nu + \ep_*) \norm{\w{X}_*}_{2,\si}^2\,.
\end{align}
It follows that $\ep_* \ge 0$ by $\w{X}_* \in \mc{B}_\nu$ and the definition \eqref{eq:restrgap} of $\lad_\nu$. On the other hand, we take test operators $Z \in \mc{M}_\nu$ (cf.\,\eqref{eq:restrgapspace}) in \eqref{eq:variationalresol}. A similar calculation gives 
\begin{align*}
  (\lad_\nu - i \alpha_n \nu + z_n) \l Z,\w{X}_n\r_{\si,1/2} = \frac{1}{K_n} \l Z, Y \r_{\si,1/2}\,,
\end{align*}
which further yields $(- \ep_* + i \d_*) \l Z,\w{X}_*\r_{\si,1/2} = 0$ by $n \to \infty$, and hence 
\begin{equation*}
     \w{X}_* \perp \mc{M}_\nu\,.
\end{equation*}
By the Min-Max theorem for the eigenvalue problem, it holds that $ \mc{E}_{\mc{L}^D}(\w{X}_*, \w{X}_*)/\norm{\w{X}_*}_{2,\si}^2$ is strictly greater than the second smallest eigenvalue of $- \mc{L}^D|_{\mc{B}_\nu}$. This contradicts with \eqref{eq:limit_eig} as $\ep_*$ could be arbitrarily small. We have proved the claim \eqref{claimbound}. 

Next, we fix a small enough $\eta$. By the boundedness of $X^{\alpha_n}_{z_n}$ proved in \eqref{claimbound}, up to a subsequence, we have $X^{\alpha_n}_{z_n} \to X_*$ for some $X_* \in \maf{H}$ as $n \to \infty$, where $X^{\alpha_n}_{z_n}$ solves the equation \eqref{eq:resolventeq} and 
$$ z_n = - (\lad_\nu + \ep) + i (\alpha_n \nu + \d) $$
with $\ep^2 + \d^2 = \eta^2$. Similarly to the argument for \eqref{claimbound}, we consider the equation \eqref{eq:resolventeq} with test operators $Z = Y \in \mc{M}_\nu$ and find 
\begin{align*} 
    z_n \l Y, X^{\alpha_n}_{z_n}\r_{\si,1/2} - i \alpha_n \l Y, [H,  X^{\alpha_n}_{z_n}] \r_{\si,1/2} + \mc{E}_{\mc{L}^D}(Y, X^{\alpha_n}_{z_n}) = \l Y, Y \r_{\si,1/2}\,,
\end{align*}
and by letting $n \to \infty$, 
\begin{align} \label{auxeq:1}
   (- \ep + i \d) \l Y, X_*\r_{\si,1/2} = 1\,.
\end{align}
The same analysis also shows that for $Z \in \mc{M}_\nu \perp Y$, there holds 
\begin{align*}
    \l Z, X_*\r_{\si,1/2} = 0\,.
\end{align*}
This means that the projection of $X_*$ on $\mc{M}_\nu$ is given by $- \frac{\ep + i \d}{\ep^2 + \d^2} Y$. We define 
\begin{align} \label{eq:projstate}
    \h{X}_* = X_* + \frac{\ep + i \d}{\ep^2 + \d^2} Y \,.
\end{align}
Moreover, we divide \cref{eq:resolventeq} with $z_n$ by $\alpha_n$ and obtain $X_* \in \mc{B}_\nu$ by letting $n \to \infty$. Then, taking test operators $Z = X^{\alpha_n}_{z_n}$ for the equation \eqref{eq:resolventeq}, we have 
\begin{equation*}
     -(\lad_\nu + \ep) \l X_*, X_*\r_{\si,1/2} + \mc{E}_{\mc{L}^D}(X_*, X_*) = \Re\, \l X_*, Y \r_{\si,1/2} = - \frac{\ep}{\ep^2 + \d^2}\,,
\end{equation*}
by \eqref{auxeq:1}. Plugging \eqref{eq:projstate} into the above formula gives, by orthogonality between $\h{X}_*$ and $Y$,
\begin{align*}
    -(\lad_\nu + \ep) \left( \norm{\h{X}_*}_{2,\si}^2 + \frac{1}{\ep^2 + \d^2}\right) + \mc{E}_{\mc{L}^D}(\h{X}_*, \h{X}_*) + \frac{\lad_\nu}{\ep^2 + \d^2} = - \frac{\ep}{\ep^2 + \d^2}\,,
\end{align*}
which implies 
\begin{align*}
         \mc{E}_{\mc{L}^D}(\h{X}_*, \h{X}_*)  = (\lad_\nu + \ep)  \norm{\h{X}_*}_{2,\si}^2 \,.
\end{align*}
Again, thanks to $\h{X}_* \perp \mc{M}_\nu$ and that $\ep$ could be arbitrarily small, it holds that 
\begin{equation*}
\h{X}_* = 0 \q \text{and} \q  X_* =  - \frac{\ep + i \d}{\ep^2 + \d^2} Y\,.
\end{equation*}

We are now ready to complete the proof. By assumption, for any small $\eta$, there is no spectrum of $\mc{L}_\alpha|_{\maf{H}}$ in $B_\eta(-\lad_v + i \alpha_n \nu)$. It follows that 
\begin{align*}
    \int_{\Gamma_c} (z - \mc{L}_{\alpha_n}|_{\maf{H}})^{-1} Y \ud z = 0\,,
\end{align*}
where $\Gamma_c = \{z \in \C\,; \ |z - (-\lad_v + i \alpha_n \nu)| = \eta\}$. Then, we have 
\begin{align} \label{eq:contou}
  \norm{2 \pi Y}_{2,\si}^2 & = \Big\| \int_{\Gamma_c} (z - \mc{L}_{\alpha_n}|_{\maf{H}})^{-1} Y \ud z - \int_{\Gamma_c} \frac{1}{z - (-\lad_v + i \alpha_n \nu)} \ud z Y  \Big\|_{2,\si}^2 \notag \\
  & \le \int_{\Gamma_c} \Big\| (z - \mc{L}_{\alpha_n}|_{\maf{H}})^{-1} Y  - \frac{1}{z - (-\lad_v + i \alpha_n \nu)}  Y  \Big\|_{2,\si}^2 \ud z \\
   & \le \int_{\Gamma_c} \Big\| X^{\alpha_n}_z  + \frac{\ep + i \d}{\ep^2 + \d^2}   Y  \Big\|_{2,\si}^2 \ud z\,. \notag
\end{align}
Recall that we have proved the boundedness of $X^{\alpha_n}_z$ with fixed $\ep$ and $\d$, and it converges to $- \frac{\ep + i \d}{\ep^2 + \d^2} Y$ as $n \to \infty$. Combining this with \eqref{eq:contou}, we have $Y = 0$, which contradicts with $\norm{Y}_{2,\si} = 1$. 
\end{proof}

\section{Convergence rate estimate via hypocoercivity} \label{sec:spacetime}

While it is shown in \cref{sec:limitbeha} that breaking the quantum detailed balance may increase the spectral gap and thereby help accelerate the mixing of QMS, the estimation of the spectral gap of the non-self-adjoint operator $\mc{L}_\alpha$ could be quite hard in general. In this section, we consider the hypocoercive primitive QMS \eqref{eq:irreverseqms} with \cref{ass:1} below, in which case we can derive an explicit and constructive decay rate estimate in the $L^2$-distance and the time-averaged one, bypassing the direct estimation of $\lad(\mc{L}_\alpha)$.
This analysis is motivated by the variational framework based on space-time Poinca\'e inequality, recently developed in \cites{albritton2019variational,cao2023explicit,brigati2023construct} for kinetic Fokker-Planck equations. 

\subsection{Setup and main result} For ease of exposition, we set the coupling parameter $\alpha = 1$ for $\mc{L}_\alpha$ (as we can always absorb $\alpha > 0$ in  $\mc{L}^H$) and consider the primitive Lindbladian:
\begin{align} \label{eq:qmslindblad}
    \mc{L} = \mc{L}^H + \mc{L}^D\,,
\end{align}
where $\mc{L}^H = i [H,\dd]$ and $\mc{L}^D$ is a $\si$-KMS detailed balanced Lindbladian. 

Our main result is the decay estimate in \cref{thm:ratest}, which is based on a novel space-time Poincar\'e  inequality for time-dependent operators (see \cref{thm:tspoincare}). 
The scaling of the convergence rate in various parameters, the optimal selection of $\alpha$, and the comparison with the recent work \cite{fang2024mixing} will be discussed in \cref{subsec:scaling}.

We recall the space $\maf{H}$ of operators with average zero \eqref{def:spacezeroavg}. Let $\Pi_0$ be the orthogonal projection (i.e., the conditional expectation \cites{bardet2017estimating}) to $\ker({\mc{L}^D|_{\maf{H}}})$ with respect to $\l \dd, \dd\r_{\si,1/2}$. For simplicity of notation, we will often abuse the notation and suppress the restriction to $\maf{H}$ in the sequel. We denote by $\Pi_+ = 1 - \Pi_0$ the complement projection, and then $\maf{H}$ can be decomposed as
\begin{equation} \label{eq:decomspace}
    \maf{H} = \maf{H}_0 \oplus \maf{H}_+\,, 
\end{equation}
where 
\begin{equation*}
    \maf{H}_0 := \ker(\mc{L}^D|_{\maf{H}})\q \text{and}\q \maf{H}_+ := \ran(\Pi_+) = \ker(\mc{L}^D|_{\maf{H}})^\perp\,.
\end{equation*}
Thanks to the finite dimensionality of $\mc{H}$ and the dissipativity of $\mc{L}^D$ with detailed balance condition, the Poincar\'{e} inequality for $\mc{L}^D$ holds on $\maf{H}_+$: there exists constant $\lad_D > 0$ such that 
\begin{equation} \label{eq:coerciveLd}
\mc{E}_{\mc{L}^D}(X,X) = - \l X, \mc{L}^D X\r_{\si,1/2} \ge \lad_D \norm{X}_{2,\si}^2\,,\q \forall\, X \in \maf{H}_+\,,
\end{equation}
where $\lad_D$ equals the spectral gap of $\mc{L}^D$. It follows that 
\begin{equation} \label{eq:poinwhole}
 \Re \mc{E}_{\mc{L}}(X,X) = - \Re \l X, (\mc{L}^H + \mc{L}^D) X\r_{\si,1/2} \ge \lad_D \norm{X}_{2,\si}^2\,, \q \forall\, X \in \maf{H}_+\,,
\end{equation}
since $\mc{L}^H$ is anti-Hermitian for $\l \dd , \dd\r_{\si,1/2}$. 

\begin{remark} \label{rem:time-avg}
The subspaces $\maf{H}_0$ and $\maf{H}_+$ correspond to non-decaying and decaying modes of the Lindblad dynamics $\exp(t \mc{L}^D)$, respectively. For the entire dynamics $\exp(t \mc{L})$ which is primitive, the Hamiltonian flow $\exp(t \mc{L}^H)$ 
facilitates the propagation of the dissipation on $\maf{H}_+$ to $\maf{H}_0$ (and hence the whole space $\maf{H}$), ensuring the convergence of $\exp(t \mc{L})$ to the equilibrium. However, the standard energy estimate, based on the Poincar\'{e} inequality \eqref{eq:poinwhole}, only implies the naive $L^2$-contraction, rather than the exponential decay, unless $\maf{H}_+ = \maf{H}$ (see \cref{rem:coercive}). It means that new proof techniques are needed. In particular, we shall consider the time-averaged evolution \eqref{eq:average}, instead of a single time slice (i.e., $T \to 0$ in $\frac{1}{T} \int_t^{t + T} \norm{\mc{P}_s X}^2_{2,\si} \ud s$); see also \cref{rem:sing}. 
\end{remark}

To estimate the exponential decay rate of Lindblad dynamics \eqref{eq:qmslindblad}, we introduce the following assumption, inspired by the hypocoercivity theory of classical kinetic equations. See \cref{rem:class1} and \cref{rem:analog} for detailed discussions. 

\begin{assumption} \label{ass:1}
The subspace $\ker({\mc{L}^D|_{\maf{H}}})$ is non-trivial, i.e., $\dim \ker({\mc{L}^D|_{\maf{H}}}) \ge 1$. In addition, the coherent term satisfies 
    \begin{equation} \label{assp:1}
        \Pi_0 \mc{L}^H \Pi_0 = 0 \,.
    \end{equation}
\end{assumption}

With the help of the decomposition \eqref{eq:decomspace} and \cref{ass:1}, the generator $\mc{L}$, restricted on the space $\maf{H}$, can be reformulated into the following block form on $\maf{H}_0 \oplus \maf{H}_+$: 
\begin{align} \label{eq:block}
    \mc{L}|_{\maf{H}} = \mm 0 & \mc{L}^H_{0+}\\ \mc{L}^H_{+0} & \mc{L}_{++}  \nn,
\end{align}
where 
\begin{align*}
    \mc{L}^H_{+0} := \Pi_+  \mc{L}^H \Pi_0 = \mc{L}^H \Pi_0\q \text{and}\q  \mc{L}^H_{0+} := \Pi_0  \mc{L}^H \Pi_+ = - (\mc{L}^H_{+0})^\star\,,
\end{align*}
and by the inequality \eqref{eq:poinwhole} and Lax-Milgram theorem, $\mc{L}_{++} : = \Pi_+ \mc{L} \Pi_+$ is invertible on $\maf{H}_+$ with $$\norm{\mc{L}_{++}^{-1}}_{\maf{H}_+\to \maf{H}_+} \le \lad_D^{-1}\,.$$

Noting that the QMS $\mc{P}_t$ of \eqref{eq:qmslindblad} is assumed to be primitive, it readily follows that $\mc{L}$ is invertible on $\maf{H}$, and by the block form \eqref{eq:block} under \cref{ass:1}, there necessarily holds 
\begin{align*}
    \ker(\mc{L}^H_{+0}) = \ker(\Pi_0) \,,
\end{align*}
equivalently, $\mc{L}^H_{+0}$ is injective on $\maf{H}_0$. It means that $(\mc{L}^H_{+0})^\star\mc{L}^H_{+0} X = 0$ for a $X \in \maf{H}_0$ if and only if $X = 0$. Thus, 
there exists $s_H > 0$ such that for any $X \in \maf{H}_0$, 
\begin{equation} \label{assp:2}
        \norm{\mc{L}_{+0}^H X}_{2,\si} \ge s_H \norm{X}_{2,\si}\,.
\end{equation}
Here, $s_H$ characterizes the singular value gap of $\mc{L}_{+0}^H$. Without loss of generality, we let $\lad_D$ and $s_H$ be the largest (optimal) constants such that \cref{eq:coerciveLd,assp:2} hold:
\begin{align} \label{eq:optconst}
    \lad_D := \inf_{X \in \maf{H}_+\backslash\{0\}} \frac{\mc{E}_{\mc{L}^D}(X,X)}{\norm{X}_{2,\si}^2}\,,\q s_H := \inf_{X \in \maf{H}_0\backslash\{0\}} \frac{ \norm{\mc{L}_{+0}^H X}_{2,\si} }{\norm{X}_{2,\si}}\,.
\end{align}

\begin{remark} \label{rem:analog}
As discussed in \cref{rem:coercive}, the condition $\dim \ker({\mc{L}^D|_{\maf{H}}}) \ge 1$ in \cref{ass:1} is equivalent to the hypocoercivity of $\mc{L}$. Indeed, if $\ker({\mc{L}^D|_{\maf{H}}}) = 0$, then $\Pi_0$ and $s_H$ have to be zero as well. In this case (i.e., $\dim \ker(\mc{L}^D) = 1$), the Lindbladian is coercive and the analysis in this section based on the space-time Poincar\'{e} inequality does not apply.  

The condition \eqref{assp:1} 
is motivated by \cite{dolbeault2015hypocoercivity}*{Assumption (H3)} for linear kinetic equations; see also \cite{bernard2022hypocoercivity}*{Assumption 2.1}.
For classical dynamics, inequalities \eqref{eq:coerciveLd} and \eqref{assp:2} are usually referred to as the microscopic and macroscopic coercivities, respectively, which are essentially the Poincar\'{e} inequalities of measures $\kappa$ and $\mu$ (see \eqref{eq:invariantclass}). 

Note that $\dim \ker({\mc{L}^D|_{\maf{H}}}) \ge 1$ is a quite natural condition in practical open quantum dynamics since the quantum noise or decoherence usually happens only locally on a large quantum system (see \cref{sec:infinite} for examples), while the assumption \eqref{assp:1} is a technical one and could be restrictive for applications. It would be interesting to relax \eqref{assp:1} and we leave it for future investigation.

\end{remark}

Before we state the main results, we introduce some additional notations that will be used in what follows. We consider the time interval $[0, T]$ with the normalized Lebesgue measure:
\begin{equation} \label{eq:normalles}
    \rd \lad(t) := \frac{1}{T}\chi_{[0,T]}(t) \ud t\,.
\end{equation}
We denote by $H_\lad^k([0,T])$ with $k \ge 0$ the standard Sobolev space of functions $f(t)$ on $[0,T]$, and by $\l f, g \r_{2,\lad} = \int_0^T \bar{f} g \ud \lad$ the inner product with respect to $\lad$. 

Given a full-rank quantum state $\si$, for a time-dependent operator $(X_t)_{t \in [0,T]} \subset \mc{B}(\mc{H})$, we define its state average, time average, and the state-time average by
\begin{equation*}
    \l X_t \r_\si := \tr(\si X_t)\,,\q \l X_t \r_{\lad} := \frac{1}{T}\int_0^T X_t \ud t\,,\q \l X_t \r_{\lad\otimes \si} := \l \l X_t\r_\si \r_\lad = \frac{1}{T} \int_0^T \tr(\si X_t) \ud t \,,
\end{equation*}
respectively. Here $\lad \otimes \si$ can be viewed as an augmented quantum state on the product space $[0,T] \times \mc{B}(\mc{H})$. Moreover, 
for $X_t, Y_t \in [0,T] \times \mc{B}(\mc{H})$, we define the inner product:
\begin{align*}
    \l X_t, Y_t \r_{2, \lad \otimes \si} := \frac{1}{T} \int_0^T \l X_t, Y_t\r_{\si,1/2} \ud t\,,
\end{align*}
where $\l \dd, \dd \r_{\si,1/2}$ is the KMS inner product. For $Y_t = X_t$, we write $\norm{X_t}_{2,\lad \otimes \si}^2 = \l X_t, X_t \r_{2, \lad \otimes \si}$. This allows us to introduce the Sobolev space of time-dependent operators: letting $\maf{S}$ be a subspace of $\mc{B}(\mc{H})$, for an integer $k \ge 0$, 
\begin{align*}
    H^k_{\lad \otimes \si}([0,T];\maf{S}): = \left\{X_t \in [0,T] \times \maf{S}\,;\  \norm{\p_t^j X_t}_{2,\lad \otimes \si} < \infty\,, \q j = 0,1,\ldots, k \right\}\,.
\end{align*}

We now present the main result of this section. 

\begin{theorem} \label{thm:ratest}
Under \cref{ass:1}, let $\lad_D$ and $s_H$ be constants introduced in \cref{eq:coerciveLd,assp:2}. For a given $X_0 \in \bh$, we denote 
$X_t: = \mc{P}_t X_0$. It holds that for any $T > 0$, the time average of $\norm{X_t}_{2,\si}$ over the interval $[t,t+T]$ exponentially decays in $t$:
\begin{align} \label{eq:averagedecay}
    \frac{1}{T}\int_t^{t + T} \norm{X_\tau - \l X_\tau\r_\si}_{2,\si}^2 \ud \tau  \le e^{-\nu t} \frac{1}{T}\int_0^{T} \norm{X_\tau - \l X_\tau\r_\si}_{2,\si}^2 \ud \tau\,,
\end{align}
with convergence rate:
\begin{align} \label{eq:ratelower}
    \nu = \frac{\lad_D}{C_{1,T}^2 + \lad_D C_{2,T}^2}\,. 
\end{align}
Here constants $C_{1,T}$ and $C_{2,T} := C_{2,T,\beta = 0}$ are defined in \cref{thm:tspoincare} below. In addition, we have 
\begin{equation} \label{eq:expdecay}
      \norm{X_{t} - \l X_t \r_{\si}}_{2,\si}^2 \le C_T e^{-\nu t} \norm{X_0 - \l X_0 \r_{\si}}_{2,\si}^2\,,\q \forall\, t \ge 0\,, 
\end{equation}
where 
\begin{align*}
  C_T := e^{\nu T} \le e^{\mathcal{O}\left(\frac{1}{\lad_D T + (\lad_D T)^{-1}} \right)} = \mc{O}(1)\,.
\end{align*} 
\end{theorem}

Note that both $\nu$ and $C_T$ depend on the parameter $T$, which can be chosen to optimize the decay estimate (see \cref{subsec:scaling}). The proof follows from a quantum space-time Poincar\'{e} inequality established in \cref{thm:tspoincare} below and the standard energy estimate. For ease of readability, we shall defer all the proofs and some related technical lemmas to \cref{sec:lemaproof}. 

\begin{theorem} \label{thm:tspoincare}
Let $s_H$ be the constant given in \cref{assp:2} under \cref{ass:1}. It holds that 
for any constant $\beta > 0$ and $X_t \in L^2_{\lad \otimes \si}([0,T]; \mc{B}(\mc{H}))$, 
\begin{align*} 
    \norm{X_t - \l X_t \r_{\lad \otimes \si}}_{2,\lad \otimes \si} \le C_{1,T} \norm{(1 - \Pi_0) X_t}_{2,\lad \otimes \si} + C_{2,T,\beta}\norm{(\beta - \mc{L}^D)^{-1/2}(-\p_t + \mc{L}^H) X_t}_{2, \lad \otimes \si}\,,
\end{align*}
where 
\begin{equation} \label{constant1}
    C_{1,T} := \sqrt{2} + 14 + \frac{12\sqrt{5}}{s_H T} + \frac{24}{(s_H T)^2} +  \sqrt{2} \max\left\{\frac{6}{s_H^2 T} + \frac{3}{s_H}, \frac{T}{\pi} \right\} \Norm{\mc{L}^H \Pi_+}_{\maf{H} \to \maf{H}} \,,
\end{equation}
and
\begin{equation} \label{constant2}
    C_{2,T,\beta} := \sqrt{2} \Norm{(\beta -  \mc{L}^D)^{1/2}}_{2 \to 2} \left(\max\left\{\frac{6}{s_H^2 T} + \frac{3}{s_H}, \frac{T}{\pi} \right\} + \max\left\{\frac{1}{s_H}, \frac{T}{\pi} \right\} \right)\,. 
\end{equation}
\end{theorem}

\begin{remark} \label{rem:sing}
The blow-up of constants $C_{1,T}$ and $C_{2,T,\beta}$ in \eqref{constant1} and \eqref{constant2} as $T \to 0$ is expected from the hypocoercivity of the Lindbladian \eqref{eq:qmslindblad}. Indeed, if there exists some constant $C > 0$ such that $\limsup_{T \to 0} C_{1, T} + C_{2, T, \beta} \le C$. Then from \eqref{eq:averagedecay}, we can readily derive, by letting $T \to 0$, 
\begin{equation*}
       \norm{X_{t} - \l X_t \r_{\si}}_{2,\si}^2 \le  e^{-\nu t} \norm{X_0 - \l X_0 \r_{\si}}_{2,\si}^2\,,\q \forall\, t \ge 0\,, 
\end{equation*}
for some $\nu > 0$, that is, \eqref{eq:expdecay} holds with prefactor $C_T = 1$. This means that $\mc{L}$ in \eqref{eq:qmslindblad} is coercive, equivalently, $\dim \ker(\mc{L^D}) = 1$, which contradicts with \cref{ass:1} (see \cref{rem:coercive}). 
\end{remark}

\subsection{Scaling of the convergence rate} \label{subsec:scaling}

We next discuss the scaling of the exponential convergence rate $\nu$ in \eqref{eq:expdecay} with respect to various model parameters. This allows us to optimally select the parameters to enhance the dissipation of QMS. 

\medskip

\noindent 
\emph{Optimizing time interval length $T$.} Recall that in the exponential decay \eqref{eq:expdecay}, both the multiplicative constant $C_T$ and the convergence rate $\nu$ depend on the choice of $T$. By \cref{constant1,constant2}, it holds that $C_{1, T}, C_{2, T, \beta} \to + \infty$, if we set $T \to 0$ or $T \to + \infty$ (see also \cref{rem:sing}). Such choice is not optimal, as in this case, the convergence rate estimate \eqref{eq:ratelower} reduces to the trivial one $\nu = 0$, and \cref{eq:expdecay} is simply the contraction: $ \norm{X_{t} - \l X_t \r_{\si}}_{2,\si} \le \norm{X_0 - \l X_0 \r_{\si}}_{2,\si}$. Thus the optimum should be achieved somewhere intermediate.  We note that this parameter is only used in the estimate of convergence, and does not affect the actual QMS.

The optimal $T$ for maximizing the rate estimate could be formally written as 
\[
T = \arg\max_T \frac{\lad_D}{C_{1, T}^2 + \lad_D C_{2, T}^2}\,.
\]
However, it is hard to derive its explicit formula. We will consider a few asymptotic regimes here. 
When $s_H \to 0$, using the ansatz $T = \Theta(s_H^t)$ with $t \in \R$, we have that if $t \ge -1$,  
\begin{align*}
    C_{1,T} = \Theta\left(\frac{1}{s_H^{2 + 2t}} + \frac{1}{s_H^{2 + t}} \right)\,,\q C_{2,T} = \Theta\left(\frac{1}{s_H^{2 + t}} \right)\,,
\end{align*}
and if $t < - 1$, 
\begin{align*}
    C_{1,T} = C_{2,T} = \Theta\left(s_H^t \right)\,, 
\end{align*}
which implies 
\begin{equation*}
    \frac{\lad_D}{C_{1,T}^2 + \lad_D C_{2,T}^2} = \begin{dcases}
        \Theta(s_H^{-2t}) & \text{if $t \le -1$}\,,\\
        \Theta(s_H^{4+2t}) & \text{if $0 \ge t \ge -1$}\,, \\
        \Theta(s_H^{4+4t}) & \text{if $t \ge 0$}\,.
    \end{dcases}
\end{equation*}
Similarly, when $s_H = \Theta(\Norm{\mc{L}^H \Pi_+}_{\maf{H} \to \maf{H}}) \to + \infty$, we have that if $t \ge -1$, then $C_{1, T} = \Theta(s_H^{1 + t})$ and $C_{2, T} = \Theta(s_H^t)$, while if $t < -1$, then $C_{1, T} = \Theta(s_H^{-2 t - 2})$ and  $C_{2, T} = \Theta(s_H^{- t - 2})$. It follows that 
\begin{equation*}
    \frac{\lad_D}{C_{1,T}^2 + \lad_D C_{2,T}^2} = \begin{dcases}
        \Theta(s_H^{4+4t}) & \text{if $t \le -1$}\,,\\
        \Theta(s_H^{-2-2t}) & \text{if $0 \ge t \ge -1$}\,, \\
        \Theta(s_H^{-2t}) & \text{if $t \ge 0$}\,.
    \end{dcases}
\end{equation*}
This suggests that the optimal $T$ is proportional to $s_H^{-1}$ (at least when $s_H \ll 1$ and $s_H \gg 1$). 
Therefore, we set $T = 3/s_H$, and a direct computation gives 
\begin{equation*} 
    C_{1,T} \le 28 + \frac{5 \sqrt{2}}{s_H}  \Norm{\mc{L}^H \Pi_+}_{\maf{H} \to \maf{H}} \,,\q  C_{2,T} =  \big\|(  \mc{L}^D)^{1/2}\big\|_{2 \to 2} \frac{6 \sqrt{2}}{s_H} \,.
\end{equation*}
Then, \cref{thm:ratest} can be simplified as follows. 

\begin{corollary} \label{coro:simplified}
    Under \cref{ass:1}, let $\lad_D$ and $s_H$ be constants in \cref{eq:optconst}. For a $X_0 \in \bh$, we denote 
$X_t: = \mc{P}_t X_0$. Then, it holds that 
\begin{equation} \label{eq:expdecay2}
      \norm{X_{t} - \l X_t \r_{\si}}_{2,\si}^2 \le C_T e^{-\nu t} \norm{X_0 - \l X_0 \r_{\si}}_{2,\si}^2\,,\q \forall\, t \ge 0\,, 
\end{equation}
with the multiplicative constant  $C_T = e^{3\nu/s_H} = \mc{O}(1)$ and the convergence rate:
\begin{equation} \label{rate}
    \begin{aligned} 
    \nu & = \frac{\lad_D s_H^2}{(28 s_H + 5 \sqrt{2} \Norm{\mc{L}^H \Pi_+}_{\maf{H} \to \maf{H}})^2 + 72 \lad_D \big\|\mc{L}^D\big\|_{2 \to 2}}  \\ &= \begin{dcases}
        \Theta(\lad_D) & \text{if $s_H = \Theta(\Norm{\mc{L}^H \Pi_+}_{\maf{H} \to \maf{H}}) \gg 1$}\,,\\
        \Theta\left(\frac{\lad_D s_H^{2}}{\norm{\mc{L}^H \Pi_+}_{\maf{H} \to \maf{H}}^2 +  \lad_D \|\mc{L}^D \|_{2 \to 2}}\right) & \text{if $s_H \ll 1$}\,.
    \end{dcases}
\end{aligned}
\end{equation}
\end{corollary}

\begin{remark}\label{rem:upplowbb}
The lower bound estimate \eqref{rate} for the rate $\nu$ is monotonically decreasing in $\norm{\mc{L}^D}_{2 \to 2}$ and $\norm{\mc{L}^H \Pi_+}_{\maf{H} \to \maf{H}}$ and increasing in $\lad_D$ and $s_H$.
Thus, for a given concrete model, we need to further bound $\lad_D$ and $s_H$ from below and $\norm{\mc{L}^D}_{2 \to 2}$ and $\norm{\mc{L}^H \Pi_+}_{\maf{H} \to \maf{H}}$ from above, which should be considered case by case.

An elementary upper bound for $\norm{\mc{L}^H \Pi_+}_{\maf{H} \to \maf{H}}$ can be given by $\lad_{\max} (H) - \lad_{\min}(H)$, which we will use in \cref{sec:example} where a few examples are considered. Here $\lad_{\max/\min}(H)$ denotes the maximal/minimal eigenvalue of $H$. In the case of local commuting Hamiltonians with local jumps, the norms $\|\mc{L}^D\|_{2 \to 2}$ and $\|\mc{L}^H \Pi_+\|_{\maf{H} \to \maf{H}}$ would typically scale polynomially in the number of qubits $n$. If there further holds $\lad_D = s_H = \Omega({\rm poly\,}(n^{-1}))$, then by standard arguments as in \cite{temme2010chi}, we can conclude the polynomial mixing time $t_{\rm mix} = \mc{O}({\rm poly\,}(n)$ (i.e., fast mixing). The detailed discussion of the dimension dependence of involved constants is beyond the scope of this work.
\end{remark}

\noindent 
\emph{Choice of coupling parameter $\alpha$.} 
We proceed to discuss the scaling of the convergence rate $\nu$ in magnitudes of $\mc{L}^H$ and $\mc{L}^D$. Note that the choice of the parameter indeed changes the QMS.

For clarity, we will consider the Lindbladian $\mc{L}_\alpha = \alpha \mc{L}^H + \mc{L}^D$ in \eqref{model:lalpha}, and suppose that \cref{ass:1} holds and operators $\mc{L}^H$ and $\mc{L}^D$ satisfy inequalities \eqref{eq:coerciveLd} and \eqref{assp:2} with constants $\lad_D$ and $s_H$. 
Note that \cref{thm:ratest} directly applies to the QMS $\exp(t \mc{L}_\alpha)$. We also introduce the rescaled Lindbladian $\w{\mc{L}}_\gamma = \mc{L}^H + \gamma \mc{L}^D$ with $\gamma = \alpha^{-1}$ that satisfies 
\begin{equation*}
    e^{t \mc{L}_\alpha} = e^{\alpha t \w{\mc{L}}_\gamma}\,. 
\end{equation*}
It follows that if the exponential decay \eqref{eq:expdecay} holds for the flow $X_t = e^{t \w{\mc{L}}_\gamma}X_0$ with the prefactor $C_T$ and the convergence rate $\nu(\w{\mc{L}}_\gamma)$, then $X_t = e^{t \mc{L}_\alpha}X_0$ would satisfy \eqref{eq:expdecay} with the same $C_T$ and rate $\nu(\mc{L}_\alpha) = \alpha \nu(\w{\mc{L}}_\gamma)$. Below, we let $C_{1,T}$ and $C_{2,T} = C_{2,T,\beta = 0}$ be constants defined in \eqref{constant1}-\eqref{constant2} based on $\lad_D$ and $s_H$ in \eqref{eq:optconst}.

We first apply \cref{thm:ratest} to $\w{\mc{L}}_\gamma$ and find 
\begin{equation*}
     \nu(\w{\mc{L}}_\gamma) \ge \frac{\gamma \lad_D}{C_{1,T}^2 + \gamma^2 \lad_D C_{2,T}^2}\,,
\end{equation*}
which is maximized at 
\begin{equation*}
    \gamma = \frac{C_{1,T}}{\sqrt{\lad_D}C_{2,T}}\,.
\end{equation*} 
Then, we readily obtain 
\begin{align} \label{eq:ratewithalpha}
    \nu(\mc{L}_\alpha) = \alpha \nu(\w{\mc{L}}_\gamma) \ge \frac{\lad_D}{C_{1,T}^2 + \alpha^{-2} \lad_D C_{2,T}^2} \nearrow \frac{\lad_D}{C_{1,T}^2} \q \text{as}\q \alpha \to \infty\,,
\end{align}
i.e., the lower bound estimate of $\nu$ is monotonically increasing when $\alpha \to \infty$. 
We note that the scaling behaviors of $\nu(\mc{L}_\alpha)$ and $\nu(\w{\mc{L}}_\gamma)$ are quite different. 
Some further direct estimation with the above discussions and \cref{coro:simplified} gives the following result.

\begin{corollary} \label{coro:ratepara}
     Under \cref{ass:1}, let $\lad_D$ and $s_H$ be constants in \cref{eq:optconst} with associated constants $C_{1,T}$ and $C_{2,T} = C_{2,T,\beta = 0}$ in \eqref{constant1}-\eqref{constant2}.  
     \begin{itemize}
         \item  For the Lindbladian $\w{\mc{L}}_\gamma = \mc{L}^H + \gamma \mc{L}^D$ with $\gamma > 0$, the convergence rate estimate \eqref{rate} is maximized at $\gamma_* = \frac{C_{1,T}}{\sqrt{\lad_D}C_{2,T}}$, and there holds
    \begin{equation*}
       \max_{\gamma > 0} \nu(\w{\mc{L}}_\gamma) \ge \frac{\sqrt{\lad_D}}{C_{1,T} C_{2,T}} \underset{T = 3/s_H}{=}  \Theta\left(\frac{\sqrt{\lad_D}s_H^{2}}{\big( s_H + \norm{\mc{L}^H \Pi_+}_{\maf{H} \to \maf{H}}\big) \sqrt{\|\mc{L}^D \|_{2 \to 2}}}\right)\,.
    \end{equation*}
    \item For the Lindbladian $\mc{L}_\alpha = \alpha \mc{L}^H + \mc{L}^D$ with $\alpha > 0$, the convergence rate estimate \eqref{rate} is increasing in $\alpha$, and there holds 
    \begin{equation*}
       \liminf_{\alpha \to \infty}\nu(\mc{L}_\alpha) \ge \frac{\lad_D}{C_{1,T}^2} \underset{T = 3/s_H}{=} \Theta\left(\frac{\lad_D s_H^2}{s_H^2 + \norm{\mc{L}^H \Pi_+}^2_{\maf{H} \to \maf{H}}}\right)\,.
    \end{equation*}
     \end{itemize}
\end{corollary}

\begin{remark}
Let us view $\mc{L}^H$ and $\mc{L}^D$ as the generators of the unitary system evolution and the decoherence process induced by local quantum noise, respectively. $\gamma > 0$ represents the noise strength. The above derivation shows that the convergence rate of the total dynamics $\exp(t \w{\mc{L}}_\gamma)$ is not monotone in $\gamma > 0$ (it will first increase and then decrease). In the regime of small noise, i.e., $\gamma \ll 1$, it is physically expected that the stronger the noise is, the faster the dynamic decays. Here, the interesting phenomenon is that when noise is large enough, i.e., $\gamma \gg 1$, the strong noise gives the fast local decoherence of $\exp(t \gamma \mc{L}_D)$, but it meanwhile slows down the overall decoherence rate $\nu(\w{\mc{L}}_\gamma)$. This is consistent with the numerical observation in \cite{laracuente2022self}.
\end{remark}

\medskip

\noindent 
\emph{Comparison with \cite{fang2024mixing} via DMS approach.} 
We finally compare our result \cref{thm:ratest} with the recent work \cite{fang2024mixing}  motivated by \cite{dolbeault2015hypocoercivity} (the analysis framework is quite different from the current work), which also assumes \cref{ass:1} for the Lindbladian $\mc{L} = \mc{L}^H  + \mc{L}^D$.

Let constants $s_H$ and $\lad_D$ be given as in \eqref{eq:optconst}. For some $\eta > 0$, we define the operator: 
\begin{align*}
    \mc{A} := \left(\eta + (\mc{L}^H \Pi_0)^\star (\mc{L}^H \Pi_0)\right)^{-1} (\mc{L}^H \Pi_0)^\star\,,
\end{align*}
and the Lyapunov
functional for $\epsilon \in (0,1)$: 
\begin{align*}
    \mf{L}(X) = \frac{1}{2} \norm{X}_{2,\si}^2 - \epsilon \, \Re \l  \mc{A} X, X \r_{\si,1/2}\,.
\end{align*}
Let the constant $C_M(\eta)$ be given by 
\begin{equation*}
    \norm{\mc{A}\mc{L}^H \Pi_+ X}_{2,\si} +  \norm{\mc{A}\mc{L}^D  X}_{2,\si} \le C_M(\eta) \norm{\Pi_+ X}_{2,\si}\,.
\end{equation*}
Then, \cite{fang2024mixing}*{Theorem 1} gives the exponential decay estimate:
\begin{equation*}
    \norm{e^{t \mc{L}} X}_{2,\si} \le C e^{-\nu t} \norm{X}_{2,\si},,\q \text{for}\ X \in \maf{H}\,,
\end{equation*}
where 
\begin{equation} \label{dmseq1}
    C = \left(\frac{1+ \epsilon}{1 - \epsilon}\right)^{1/2}\,, \q \nu = \min\left\{\frac{1}{4}\frac{\lad_D}{1 + \eps}, \frac{1}{3}\frac{\eps s_H^2}{(1 + \eps)(\eta + s_H^2)} \right\},
\end{equation}
with
\begin{equation} \label{dmseq2}
    \eps = \frac{1}{2}\min\left\{ \frac{\lad_D s_H^2}{(\eta + s_H^2)(1 + C_M(\eta))^2},1  \right\}\,.
\end{equation}
To compare with our results, we let $\lad_D, s_H \ll 1$, and note, by using \cite{fang2024mixing}*{Lemma 2}, 
\begin{equation*}
    C_M(\eta) \le \frac{1}{2 \sqrt{\eta}} \left(\norm{\mc{L}^H \Pi_+}_{\maf{H} \to \maf{H}} + \norm{\mc{L}^D}_{2 \to 2}\right)\,.
\end{equation*}
It follows that the optimal $\eta$ in this regime is $\eta = \Theta(s_H^2)$, and there holds \cite{fang2024mixing}*{Corollary 4}
\begin{equation} \label{eqfang}
    C = \Theta(1)\,,\q \nu = \Theta \left(\frac{\lad_D s_H^2}{\norm{\mc{L}^H \Pi_+}^2_{\maf{H} \to \maf{H}} + \norm{\mc{L}^D}^2_{2 \to 2}}\right)\,,
\end{equation}
which matches the rate estimate in \cref{coro:simplified} in the sense that both give $\nu = \Theta (\lad_D s_H^2)$. 

However, beyond the regime $\lad_D, s_H \ll 1$, 
the explicit form of the rate estimate as \eqref{eqfang} is not available. 
Accurate estimations for $\lad_D$, $s_H$, and $C_M(\eta)$ are necessary to find the minimum in \cref{dmseq1,dmseq2}.
One advantage of our results is that they do not depend on the twisted norm, and their scaling with respect to the model parameters is fully explicit. This allows us to determine the optimal $\gamma$ and $\alpha$ for the convergence rate estimation \eqref{dmseq1} in \cref{coro:ratepara} above. It would be interesting to conduct a detailed case study for various physical models to understand the sharpness of our decay rate estimate compared to that in \cite{fang2024mixing}.




\subsection{Proofs and Lemmas}\label{sec:lemaproof}

This section is devoted to the proofs of \cref{thm:ratest} (exponential decay of time-averaged $L^2$-distance) and 
\cref{thm:tspoincare} (quantum space-time Poincar\'{e} inequality).

For this, we need some auxiliary technical lemmas, in particular, the key \cref{lem:diverest}, where we consider a divergence-type equation, whose solutions provide some crucial test functions used for the proof of \cref{thm:tspoincare}. We start with lifting \cref{assp:2} to the space of time-dependent operators, which follows from a standard tensorization technique. For ease of exposition, we define the subspace of $\mc{B}(\mc{H})$:
\begin{equation*}
   \wb{\maf{H}}_0 := \maf{H}_0 \oplus {\rm Span}\{\bf 1\}\,.
\end{equation*}

\begin{lemma} \label{lem:tspoin1}
Let $s_H$ be the constant in \eqref{assp:2} under \cref{ass:1}. For $X_t \in H^1_{\lad \otimes \si}([0,T]; \wb{\maf{H}}_0)$, there holds 
    \begin{align} \label{eq:tspoin1}
        \norm{X_t - \l X_t\r_{\lad \otimes \si}}^2_{2,\lad \otimes \si} \le \max\left\{\frac{1}{s_H^2}, \frac{T^2}{\pi^2} \right\} \left(\Norm{ \mc{L}_{+0}^H  X_t}_{2, \lad \otimes \si}^2 + \big\|\p_t X_t\big\|^2_{2,\lad \otimes \si}   \right).
    \end{align}
\end{lemma}

\begin{proof}
We first note 
\begin{equation*}
    \norm{X_t - \l X_t\r_{\lad \otimes \si}}^2_{2,\lad \otimes \si} = \norm{ X_t - \l X_t\r_{\lad}}^2_{2,\lad \otimes \si}  + \norm{\l X_t \r_{\lad} - \l X_t\r_{\lad \otimes \si}}_{2,\lad \otimes \si}^2\,,
\end{equation*}
thanks to, for any $X_t \in L^2_{\lad \otimes \si}([0,T];\mc{B}(\mc{H}))$, 
$$\l X_t - \l X_t\r_{\lad} ,  \l X_t \r_{\lad} - \l X_t\r_{\lad \otimes \si} \r_{2,\lad \otimes \si} = 0\,.$$
It is easy to see $\l X_t \r_{\lad} - \l X_t\r_{\lad \otimes \si} \in \maf{H}_0$ for $X_t \in H^1_{\lad \otimes \si}([0,T]; \wb{\maf{H}}_0)$. By \cref{assp:2}, we have 
\begin{equation*}
   \norm{  \l X_t \r_{\lad} - \l X_t\r_{\lad \otimes \si} }_{2,\lad \otimes \si}^2 \le  \frac{1}{s^2_H} \norm{\l \mc{L}_{+0}^H  X_t \r_\lad}_{2,\si}^2 \le \frac{1}{s^2_H} \norm{ \mc{L}_{+0}^H X_t}_{2, \lad \otimes \si}^2\,,
\end{equation*}
where the last inequality follows from the convexity of $\norm{\dd}_{2,\si}^2$. Moreover, by the standard Poincar\'{e} inequality for $H^1$-functions \cite{evans2022partial}:
\begin{equation*}
    \left\|f - \int_0^T f(s) \ud \lad\right\|_{L^2([0,T])} \le \frac{T}{\pi} \norm{f'}_{L^2([0,T])}\,,
\end{equation*}
there holds
\begin{align*}
    \norm{ X_t - \l  X_t \r_\lad}^2_{2,\lad \otimes \si} & = \int_0^T \norm{ X_t - \l  X_t \r_\lad}_{2,\si}^2 \ud \lad \\ 
    & \le  \frac{T^2}{\pi^2} \int_0^T \norm{\p_t X_t}_{2,\si}^2  \ud \lad \le \frac{T^2}{\pi^2}  \norm{\p_t X_t}_{2,\lad \otimes \si}^2\,.
\end{align*}
Then, we have 
\begin{equation*}
     \norm{X_t - \l X_t\r_{\lad \otimes \si}}^2_{2,\lad \otimes \si} \le \frac{1}{s^2_H} \norm{ \mc{L}_{+0}^H  X_t}_{2, \lad \otimes \si}^2 + \frac{T^2}{\pi^2}  \norm{\p_t X_t}_{2,\lad \otimes \si}^2\,,
\end{equation*}
which gives the desired \eqref{eq:tspoin1}. 
\end{proof}

We consider the following abstract elliptic operator:
\begin{equation} \label{eq:time-laplace}
    \maf{L}^H = 
    - \p_{tt} + \big(\mc{L}_{+0}^H\big)^\star \mc{L}_{+0}^H\,,
\end{equation}
and the space:
\begin{equation*}
\mc{V} := \left\{X_t \in H^1_{\lad \otimes \si}([0,T]; \wb{\maf{H}}_0)\,;\ \l X_t\r_{\lad \otimes \si} = 0 \right\}\,.
\end{equation*}
The well-posedness of the operator $\maf{L}^H$ is established in the following lemma.

\begin{lemma} \label{lem:prioriest}
Let $s_H$ be the constant given in \eqref{assp:2} under \cref{ass:1}. Given $F_t \in H^{-1}_{\lad \otimes \si}([0,T]; \wb{\maf{H}}_0)$ with $\l F_t \r_{\lad \otimes \si} = 0$, 
consider the Neumann boundary value problem: 
\begin{align} \label{eq:elliptic}
    \maf{L}^H X_t = F_t\,, \q  \p_t X_t|_{t = 0, T} = 0\,.
\end{align}
Then there is a unique weak solution $X_t \in \mc{V}$ such that 
\begin{equation} \label{eq:elliptic2}
    \l \p_t Y_t, \p_t X_t \r_{2, \lad \otimes \si} + \l \mc{L}_{+0}^H Y_t, \mc{L}_{+0}^H X_t \r_{2, \lad \otimes \si} = \l F_t, Y_t \r_{2, \lad \otimes \si}\,, \q \forall\, Y_t \in H^1_{\lad \otimes \si}([0, T]; \wb{\maf{H}}_0)\,.
\end{equation}
In addition, if $F_t \in L^2_{\lad \otimes \si}([0,T]; \wb{\maf{H}}_0)$, the solution $X_t$ satisfies $X_t \in H^2_{\lad \otimes \si}([0,T]; \wb{\maf{H}}_0)$ with the estimate:
\begin{align} \label{eq:prioriest}
    \norm{\p_t X_t}^2_{2,\lad \otimes \si} + \norm{\mc{L}_{+0}^H X_t}_{2,\lad \otimes \si}^2 \le \max\left\{\frac{1}{s^2_H}, \frac{T^2}{\pi^2} \right\} \norm{F_t}_{2,\lad \otimes \si}^2\,, 
\end{align}
and 
\begin{align} \label{eq:prioriest2}
     \big\|\p_{tt} X_t\big\|^2_{2,\lad \otimes \si} + \Norm{(\mc{L}_{+0}^H)^\star\mc{L}_{+0}^H X_t}_{2,\lad \otimes \si}^2 + 2 \big\| \p_t \mc{L}_{+0}^H X_t \big\|_{2,\lad \otimes \si}^2 = \big\| \maf{L}^H X_t \big\|_{2,\lad \otimes \si}^2 = \big\| F_t \big\|_{2,\lad \otimes \si}^2 \,.
\end{align}
\end{lemma}

\begin{proof}
We consider the bilinear form on $\mc{V}$ associated with the equation \eqref{eq:elliptic}:
    \begin{equation*}
        \mc{Q}(Y_t, X_t) =  \big\l \p_t Y_t, \p_t X_t \big\r_{2, \lad \otimes \si} + \big\l \mc{L}_{+0}^H Y_t, \mc{L}_{+0}^H X_t \big\r_{2, \lad \otimes \si}\,. 
    \end{equation*}
    One can easily see that $\mc{Q}(\dd,\dd)$ is an inner product on $\mc{V}$. In fact, if $\mc{Q}(X,X) = 0$, we have $\p_t X_t = 0$ and $\mc{L}_{+0}^H X = 0$. The latter condition $\mc{L}_{+0}^H X = 0$ gives that for a.e.\,$t \in [0,T]$, $X_t$ is proportional to the identity ${\bf 1}$. Then the former one $\p_t X_t = 0$, along with $\l X_t\r_{\lad \otimes \si} = 0$, implies $X_t = 0$. Further, \cref{lem:tspoin1} shows that  the bilinear form $\mc{Q}(\dd,\dd)$ is coercive on $\mc{V}$: for some $C > 0$, 
    \begin{align*}
        \mc{Q}(X_t,X_t) \ge C \left(\norm{X_t}^2_{2,\lad \otimes \si} + \norm{\p_t X_t}^2_{2,\lad \otimes \si}\right), \q \forall\, X_t \in \mc{V}\,.
    \end{align*}
    The existence of the weak solution to \eqref{eq:elliptic} is guaranteed by the Lax-Milgram theorem with the 
    estimate \eqref{eq:prioriest}. Finally, note that $\maf{L}^H X_t = F_t$ is equivalent to $$- \p_{tt} X_t = F_t - \left(\mc{L}_{+0}^H\right)^\star \mc{L}_{+0}^H X_t\,,$$ which gives $X \in H^2_{\lad \otimes \si}([0,T];\wb{\maf{H}}_0)$. The estimate \eqref{eq:prioriest2} follows from 
    \begin{align*}
        \big\| \maf{L}^H X_t \big\|_{2,\lad\otimes \si}^2 & =  \big\l (-\p_{tt} + (\mc{L}_{+0}^H)^\star\mc{L}_{+0}^H)X_t, (-\p_{tt} + (\mc{L}_{+0}^H)^\star\mc{L}_{+0}^H)X_t  \big\r_{2,\lad\otimes \si} \\
        & =  \big\|\p_{tt} X_t\big\|^2_{2,\lad \otimes \si} + \Norm{(\mc{L}_{+0}^H)^\star\mc{L}_{+0}^H X_t}_{2,\lad \otimes \si}^2 + 2\, \Re \left\l -\p_{tt} X_t,  (\mc{L}_{+0}^H)^\star\mc{L}_{+0}^H X_t  \right\r_{2,\lad\otimes \si}\,,
    \end{align*}
    and by the commutativity between $\p_t$ and $\mc{L}_{+0}^H$, 
    \begin{equation*}
        \left\l -\p_{tt} X_t,  (\mc{L}_{+0}^H)^\star\mc{L}_{+0}^H X_t  \right\r_{2,\lad\otimes \si} =  \left\l \p_{t} \mc{L}_{+0}^H X_t,  \p_{t} \mc{L}_{+0}^H  X_t  \right\r_{2,\lad\otimes \si}\,. \qedhere
    \end{equation*}
\end{proof}

The next technical lemma is essential and deals with a divergence-type equation \eqref{eq:diver} associated with the elliptic one \eqref{eq:elliptic}, which is inspired by the arguments in \cite{cao2023explicit}*{Lemma 2.6} and
\cite{brigati2023construct}*{Section 6} for the underdamped Langevin dynamics. The main purpose of this lemma is to construct necessary test operators $(Y_t^{(0)}, Y_t^{(1)})$, vanishing at time points $t = 0, T$, required for the proof of \cref{thm:tspoincare} (see \cref{auxeq:space1,auxeq:space2,auxeq:space3} below). Letting $X_t$ be the solution to \eqref{eq:elliptic}, it is easy to see that $(Y_t^{(0)}, Y_t^{(1)}) = (\p_t X_t, \mc{L}_{+0}^H X_t)$ satisfies \eqref{eq:diver}, except the one $Y_t^{(1)} = 0$ at $t = 0, T$ unless $F_t$ is orthogonal to the subspace $\hh = \{X_t\,;\  \maf{L}^H X_t = 0\}$. In the case of $F_t \in \hh$, the construction of 
the solution $(Y_t^{(0)}, Y_t^{(1)})$ to \eqref{eq:diver} with desired estimates is subtle and relies on the method of separation of variables. Combining these two parts of solutions finishes the proof of \cref{lem:diverest}.
 
\begin{lemma} \label{lem:diverest}
Let $s_H$ be the constant given in \eqref{assp:2} under \cref{ass:1}. For $F_t \in L^2_{\lad \otimes \si}([0,T];\wb{\maf{H}}_0 )$ with $\l F_t \r_{\lad\otimes \si} = 0$, there exists $Y^{(0)}_t \in  H^1_{\lad \otimes \si}([0,T];\wb{\maf{H}}_0)$ and $Y^{(1)}_t \in  H^1_{\lad \otimes \si}([0,T]; \maf{H}_+)$ that satisfy the following equation with Dirichlet boundary condition in $t$:
\begin{equation} \label{eq:diver}
    -\p_t Y^{(0)}_t + \left(\mc{L}_{+0}^H \right)^{\star} \ Y_t^{(1)} = F_t\,, \q \text{$Y^{(j)}_t|_{t = 0, T} = 0$ for $j = 0, 1$}\,.
\end{equation}
In addition, we have 
\begin{equation} \label{eq:mainest}
\begin{aligned}
    & \big\|Y^{(0)}_t\big\|_{2,\lad \otimes \si} \le \sqrt{2} \max \left\{\frac{1}{s_H}, \frac{T}{\pi} \right\}  \Norm{F_t}_{2,\lad \otimes \si}\,, \\ 
& \big\|Y^{(1)}_t\big\|_{2,\lad \otimes \si} \le  \sqrt{2} \max\left\{\frac{6}{s_H^2 T} + \frac{3}{s_H}, \frac{T}{\pi} \right\}  \Norm{F_t}_{2,\lad \otimes \si}\,, 
\end{aligned} 
\end{equation}
and 
\begin{equation} \label{eq:mainestdev}
     \begin{aligned}
         & \big\|\p_t Y^{(0)}_t\big\|_{2,\lad \otimes \si} \le \sqrt{26} \left(1 + \frac{1}{\sqrt{26}} + \frac{2}{s_H T} \right)  \Norm{F_t}_{2,\lad \otimes \si}\,,\\
         & \big\|\p_t Y^{(1)}_t\big\|_{2,\lad \otimes \si} \le 12 \left( \frac{13}{12} + \frac{\sqrt{5}}{s_H T} + \frac{2}{(s_H T)^2} \right) \Norm{F_t}_{2,\lad \otimes \si}\,.
     \end{aligned}
\end{equation}
\end{lemma}

\begin{proof}
{\bf Decomposition of source $F_t$.} 
We recall the operator $\maf{L}^H$ defined in \eqref{eq:time-laplace}. 
Let $\hh$ be the subspace of $L^2_{\lad \otimes \si}([0,T]; \wb{\maf{H}}_0)$ consisting of those $X_t \in H^2_{\lad \otimes \si}([0,T]; \wb{\maf{H}}_0)$ satisfying $\maf{L}^H X_t = 0$.
Then, for a given $F_t \in L^2_{\lad \otimes \si}([0,T];\wb{\maf{H}}_0 )$ with $\l F_t \r_{\lad\otimes \si} = 0$, we decompose
\begin{equation} \label{eq:sourcedecomp}
    F_t = F^h_t + F^\perp_t \q \text{with}\q F^h_t \in \hh\,,\ F^\perp_t \perp \hh\,,
\end{equation}
with respect to the inner product $\l \dd, \dd \r_{2, \lad \otimes \si}$. Since a constant operator $X_t = c {\bf 1}$ with $c \in \C$ belongs to $\hh$, we have $\l F^\perp_t\r_{\lad \otimes \si} = 0$ and thus $\l F^h_t \r_{\lad \otimes \si} = \l F - F^\perp_t \r_{\lad \otimes \si} = 0$. We shall construct the solution to the equation \eqref{eq:diver} with source $F^h_t$ or $F^\perp_t$ separately. 

\medskip

\noindent
{\bf Construction of solution for $F^\perp$.}
We first consider the easier case of $F^\perp_t$. Let $X_t$ be the solution to the variational equation \eqref{eq:elliptic2} with source $F^\perp_t$. We take the test function $Z_t \in \hh$ and obtain
\begin{equation*}
 0 = \left\l Z_t, F_t^\perp \right\r_{2, \lad \otimes \si} = \left\l \maf{L}^H Z_t, X_t \right\r_{2,\lad \otimes \si} + \left\l \p_t Z_t, X_t \right\r_{\si,1/2}|_{t = T} - \left\l \p_t Z_t, X_t \right\r_{\si,1/2}|_{t = 0}\,,
\end{equation*}
using integration by parts, thanks to the $H^2$-regularity of $Z_t$ in $t$. Noting that $ \maf{L}^H Z_t = 0$ and $\p_t Z_t$ is arbitrary at $t = 0, T$ (see also the decomposition \eqref{auxeq:depsource} below for an operator in $\mathbb{H}$), we can conclude  $X_t|_{t = 0, T} = 0$. We then define 
\begin{align*}
Y^{(0)}_t = \p_t X_t\,,\q  Y^{(1)}_t = \mc{L}_{+0}^H X_t\,,
\end{align*}
with the following estimate by \cref{lem:prioriest}, 
\begin{equation} \label{eq:est1}
   \big\|Y^{(0)}_t\big\|^2_{2,\lad \otimes \si} + \big\|Y^{(1)}_t\big\|_{2,\lad \otimes \si}^2 \le \max\left\{\frac{1}{s^2_H}, \frac{T^2}{\pi^2} \right\} \Norm{F^\perp_t}_{2,\lad \otimes \si}^2\,,
\end{equation}
and 
\begin{align} \label{eq:devest1}
     \big\|\p_t Y^{(0)}_t\big\|^2_{2,\lad \otimes \si} + \big\|\p_t Y^{(1)}_t\big\|_{2,\lad \otimes \si}^2 \le \Norm{\maf{L}^H X_t}_{2,\lad\otimes \si}^2 \le \Norm{F^\perp_t}_{2,\lad \otimes \si}^2\,.
\end{align}

\medskip

\noindent
{\bf Construction of solution for $F^h$.}
We proceed to consider the case of $F^h_t$. For this, recall that $\left(\mc{L}_{+0}^H\right)^\star \mc{L}_{+0}^H = \Pi_0 \left(\mc{L}^H\right)^2 \Pi_0$ is a self-adjoint operator on $\maf{H}_0$. We assume that its orthonormal eigenvectors $\{W_\mu\}$ 
with respect to KMS inner product $\l \dd , \dd \r_{\si,1/2}$ are given by 
\begin{align*} 
    \left(\mc{L}_{+0}^H\right)^\star \mc{L}_{+0}^H W_\mu = \mu^2 W_\mu\,,
\end{align*}
where $\mu^2 \ge s_H^2 > 0$ by \cref{assp:2}. With the help of $\{W_\mu\}$, we have the following decomposition:
\begin{equation}  \label{auxeq:depsource}
      F^h_t = f_0(t) + \sum_\mu f_\mu(t) W_\mu\,,
\end{equation}
with $\l f_0(t) \r_\lad = 0$. Then the condition $\maf{L}^H F_t^h = 0$ gives 
\begin{align*}
     0 = - f_0''(t) + \sum_\mu \left(-f_\mu'' + \mu^2 f_\mu(t)\right) W_\mu\,,
\end{align*}
further implying 
\begin{equation*}
    f_0''(t) = 0 \qquad \text{and} \qquad f_\mu'' - \mu^2 f_\mu(t) = 0\,.    
\end{equation*}
Solving above ODEs yields $f_0(t) = c(t - \frac{T}{2})$ for some constant $c \in \C$, and that $f_\mu(t)$ can be decomposed as: for some $c_+^\mu, c_-^\mu \in \C$, 
\begin{equation} \label{eq:generalflad}
    f_\mu(t) = c_+^\mu e^{- \mu t} + c_-^\mu e^{-\mu(T - t)}\,,
\end{equation}
where coefficients  $c_{\pm}^\mu$ are computed below. 
Due to linearity, to construct the solution to \cref{eq:diver} for $F^h_t$, it suffices to construct for each term in the decomposition \eqref{auxeq:depsource} and add them together.

\medskip

\noindent
\textit{Computation of coefficients $c_{\pm}^\mu$.} We now calculate the constants $c_{\pm}^\mu$ in \eqref{eq:generalflad}. For convenience, we define constants:
\begin{equation} \label{def:constant}
    v_T := \norm{e^{-\mu t}}_{2,\lad}^2 = \norm{e^{-\mu (T - t)}}_{2,\lad}^2 =  \int_0^T e^{- 2 \mu t} \ud \lad(t) = \frac{1 - \ell^2}{2 \mu T}\,,  \q \ell := e^{- \mu T}\,,
\end{equation}
and functions ${\bf e}_{\mu, \pm}$ as linear combinations of $e^{- \mu t}$ and $e^{- \mu(T - t)}$:
\begin{equation*}
  {\bf e}_{\mu,+}(t) := e^{- \mu t} - \frac{\l e^{- \mu t}, e^{- \mu(T - t)}\r_{2,\lad}}{\norm{e^{- \mu(T - t)}}_{2,\lad}^2}e^{- \mu(T - t)} = e^{-\mu t} -  v_T^{-1} e^{- \mu(2 T - t)}\,,
\end{equation*}
and 
\begin{equation*}
    {\bf e}_{\mu,-}(t) := e^{- \mu(T - t)} - \frac{\l e^{- \mu t}, e^{- \mu(T - t)}\r_{2,\lad}}{\norm{e^{- \mu t}}_{2,\lad}^2}e^{- \mu t} = e^{- \mu(T - t)} - v_T^{-1} e^{- \mu (t + T)}\,.
\end{equation*}
Taking the inner product of \cref{eq:generalflad} with $e^{-\mu t}$ and $e^{-\mu (T - t)}$ and solving the resulting linear equations in $c_{\pm}^\mu \in \C$, we find 
\begin{equation*}
   c_{+}^\mu := \frac{\l f_\mu, {\bf e}_{\mu, +}\r_{2\,\lad}}{\l {\bf e}_{\mu, +}, e^{-\mu t}\r_{2,\lad}}\,,\q  c_{-}^\mu: = \frac{\l f_\mu, {\bf e}_{\mu, -}\r_{2\,\lad}}{\l {\bf e}_{\mu, -}, e^{-\mu (T - t)}\r_{2,\lad}}\,.
\end{equation*}
Before we proceed, we estimate 
\begin{equation} \label{eq:sourceestimate}
    \begin{aligned}
     \norm{c_{+}^\mu e^{-\mu t}}_{2,\lad}^2 = |c_{+}^\mu|^2 v_T \le \left(\frac{\norm{{\bf e}_{\mu, +}}_{2,\lad}}{\l {\bf e}_{\mu, +}, e^{-\mu t}\r_{2,\lad}} \right)^2 \norm{f_\mu}_{2,\lad}^2 v_T = \frac{v_T^2}{v_T^2 - \ell^2} \norm{f_\mu}_{2,\lad}^2\,,
    \end{aligned}
\end{equation}
by a direct computation: 
\begin{equation*}
    \norm{{\bf e}_{\mu, +}}_{2,\lad}^2 = \l {\bf e}_{\mu, +}, e^{-\mu t}\r_{2,\lad} = v_T - v_T^{-1} \ell^2\,.
\end{equation*}
It is also easy to see that $v_T^2/(v_T^2 - \ell^2)$ is monotonically decreasing in $\mu T$ and 
\begin{equation*}
    \frac{v_T^2}{v_T^2 - \ell^2} \to \begin{dcases}
        1\,,&   {\rm as}\q \mu  T \to \infty\,, \\
        \frac{3}{(\mu T)^2}(1 + \mathcal{O}(\mu T))\,,&  {\rm as}\q \mu  T \to 0\,. \\
    \end{dcases}
\end{equation*}
Thus, for any $\mu, T > 0$, there holds
\begin{equation} \label{eq:constest}
    \frac{v_T^2}{v_T^2 - \ell^2} \le  \frac{4}{(\mu T)^2} + 1 \,.
\end{equation}
A similar estimate holds for $ \norm{c_{-}^\mu e^{-\mu (T -t)}}_{2,\lad}$: 
\begin{align*}
     \norm{c_{-}^\mu e^{-\mu (T - t)}}_{2,\lad}^2 = |c_{-}^\mu|^2 v_T \le  \frac{v_T^2}{v_T^2 - \ell^2} \norm{f_\mu}_{2,\lad}^2\,.
\end{align*}

\medskip

\noindent
\textit{Construction of solution for $F^h = f_0(t)$.}
We start with constructing the solution for $f_0(t)$, which is direct. We set $Y^{(0)}_t = - \int_0^t f_0(\tau) \ud \tau$ and $Y^{(1)}_t = 0$. It is easy to see, by $Y_t^{(0)}|_{t = 0, T} = 0$,  
\begin{equation} \label{eq:est0}
    \big\|Y^{(0)}_t\big\|_{2,\lad \otimes \si} \le \frac{T}{\pi} \norm{f_0}_{2,\lad \otimes \si}\,,
\end{equation}
and by definition of $Y_t^{(0)}$, 
\begin{equation} \label{eq:devest0}
   \big\|\p_t Y^{(0)}_t\big\|_{2,\lad \otimes \si} = \norm{f_0}_{2,\lad \otimes \si}\,.
\end{equation}

\medskip

\noindent
\textit{Construction of solution for $F^h = e^{-\mu t} W_\mu$.} 
We next consider the case of $f_\mu(t)W_\mu$ with $f_\mu$ being of the form \eqref{eq:generalflad}. By linear combination, we only need to consider $f_\mu(t) = e^{-\mu t}$ or $e^{-\mu (T - t)}$. We shall focus on the former case, as the latter one can be similarly handled by reversing the time. To be precise, we shall construct the solution to 
\begin{align*}
     -\p_t Y^{(0)}_t + \left(\mc{L}_{+0}^H \right)^{\star} \ Y_t^{(1)} = e^{-\mu t} W_\mu\,.
\end{align*}
We consider the following ansatz: $Y^{(0)}_t = \psi_{0,\mu}(t) W_\mu$ and $Y^{(1)}_t = \psi_{1,\mu}(t) \mc{L}^H_{+0} W_\mu$ for some functions $\psi_{0,\mu}$ and $\psi_{1,\mu}$ satisfying $\psi_{0,\mu}|_{t = 0,T} = \psi_{1,\mu}|_{t = 0,T} = 0$ and the equation:  
\begin{equation} \label{eq:psi0psi1}
    - \psi_{0,\mu}'(t) + \mu^2 \psi_{1,\mu}(t) = e^{- \mu t}\,.
\end{equation}
We first compute
\begin{equation} \label{auxeq:source}
    \norm{e^{-\mu t} W_\mu}_{2,\lad \otimes \si}^2 = \frac{1}{T} \int_0^T e^{-2 \mu t} \ud t = v_T\,.
\end{equation}
Motivated by the proof of \cite{cao2023explicit}*{Lemma 2.6}, we consider the following ansatz: 
\begin{align} \label{eq:soleqphi}
    \psi_{0,\mu}(t) = p_{0,\mu}(e^{-\mu t})\,,\q \psi_{1,\mu}(t) = p_{1,\mu}(e^{-\mu t})\,,
\end{align}
where $p_{0,\mu}$ and $p_{1,\mu}$ are polynomial functions to be constructed. Note from the equation \eqref{eq:psi0psi1} that (recalling $\ell = e^{-\mu T}$ in \cref{def:constant}) 
\begin{equation} \label{eq:constraint}
    \int_0^T \psi_{1,\mu}(t) \ud t =  \int_\ell^1 \frac{p_{1,\mu}(x) }{\mu x} \ud x = \int_0^T \frac{e^{-\mu t} + \psi_{0,\mu}'(t)}{\mu^2} \ud t = \frac{1 - \ell}{\mu^3}\,.
\end{equation}
We define 
\begin{equation} \label{eq:p1mu}
     p_{1,\mu}(x) =  \mu^{-2} g_\mu(x)\,, \q g_{\mu}(x) = \frac{6}{(\ell - 1)^2} x (1- x)(x - \ell)\,,
\end{equation}
which satisfies the constraint \eqref{eq:constraint} and $p_{1,\mu}(1) = p_{1,\mu}(\ell) = 0$. We then compute 
\begin{equation} \label{auxeq:apsi}
    \psi_{0,\mu}(t) = \int_0^t \mu^2 p_{1,\mu}(e^{-\mu s}) - e^{- \mu s} \ud s = \frac{1}{\mu} \int_{e^{-\mu t}}^1 \frac{g_{\mu}(x) - x}{x} \ud x = p_{0,\mu}(e^{-\mu t} )\,,
\end{equation}
and by \cref{eq:p1mu},  
\begin{align} \label{auxeq:appp}
   \mu p_{0,\mu}(s) = \int_s^1 \frac{6}{(\ell - 1)^2} (1- x)(x - \ell) - 1 \ud x = \frac{(s-1)(\ell - 2 s + 1)(\ell - s)}{(\ell - 1)^2}\,.
\end{align}
By the construction, we have, for $x \in [\ell, 1]$, 
\begin{equation*}
   0 \le  g_\mu(x) \le \frac{3}{2} x\,,\q  g_\mu'(x) = \frac{6}{(\ell - 1)^2} (-3 x^2 - \ell + 2 x (1 + \ell))\,.
\end{equation*}
It follows that, by \cref{eq:constraint}, 
\begin{align} \label{auxeq11}
    \norm{\psi_{1,\mu}}_{2,\lad}^2 = \frac{1}{T \mu^5}\int_\ell^1  \frac{g_\mu(x)^2}{x} \ud x \le \frac{9}{4 T \mu^5}\int_\ell^1  x \ud x \le \frac{9(1 - \ell^2)}{8 T \mu^5} = \frac{9}{4 \mu^4} v_T \,,
\end{align}
and 
\begin{equation} \label{auxeq12}
\begin{aligned}
    \norm{\psi_{1,\mu}'}_{2,\lad}^2 = \frac{1}{T \mu^3}\int_\ell^1 g'(x)^2 x \ud x & = \frac{36}{T \mu^3 (\ell - 1)^4}  \int_\ell^1 (-3 x^2 - \ell + x (2 + 2 \ell))^2 x \ud x \\
    & = \frac{6}{5 T \mu^3 (\ell - 1)^4}  (1 - \ell)^3 (3 + 2 \ell + 2 \ell^2 + 3 \ell^3) \le \frac{12}{T \mu^3 (1 - \ell)}\,.
\end{aligned}
\end{equation}
Noting $ \psi_{0,\mu}'(t) = \mu^2 \psi_{1,\mu}(t) - e^{\mu t}$ by \cref{eq:psi0psi1}, a straightforward computation gives 
\begin{align} \label{auxeq21}
    \norm{\psi_{0,\mu}'}_{2,\lad}^2 \le 2 \mu^4 \norm{\psi_{1,\mu}}_{2,\lad}^2 + 2 \norm{e^{\mu t}}_{2,\lad}^2 \le 2 \mu^4 \frac{9}{4 \mu^4} v_T + 2 v_T = \frac{13}{2} v_T\,.
\end{align}
Then, by \cref{auxeq:apsi,auxeq:appp}, we have 
\begin{equation} \label{auxeq22}
\begin{aligned}
     \norm{\psi_{0,\mu}}_{2,\lad}^2 = \frac{1}{T \mu} \int_\ell^1 \frac{p_{0,\mu}(x)^2}{x} \ud x & = \frac{(1-\ell)^3}{T \mu^3} \int_0^1  \frac{ (s-1)^2(2s - 1)^2 s^2}{s + \ell(1-s)} \ud s \\
     & \le \frac{(1-\ell)^3}{T \mu^3} \int_0^1 (s-1)^2(2s - 1)^2 s \ud s = \frac{(1-\ell)^2}{30(1+\ell) \mu^2} v_T\,.
\end{aligned}
\end{equation}

\medskip

\noindent
\textit{Estimate of norms.} Recalling the decomposition \eqref{auxeq:depsource}-\eqref{eq:generalflad} for $F^h_t$, with the help of $\psi_{0,\mu}$ and $\psi_{1,\mu}$ defined above, 
 the solution to \eqref{eq:diver} with source $F^h_t = \sum_\mu f_\mu(t) W_\mu$  can be constructed as follows: 
\begin{align} \label{eq:repy0}
    Y_t^{(0)} = - \int_0^t f_0(\tau) \ud \tau + \sum_\mu \left(c_+^\mu\psi_{0,\mu}(t) - c_-^\mu \psi_{0,\mu}(T - t)\right) W_\mu\,, 
\end{align}
and 
\begin{align} \label{eq:repy1}
    Y_t^{(1)} = \sum_\mu \left(c_+^\mu\psi_{1,\mu}(t) + c_-^\mu \psi_{1,\mu}(T - t)\right) \mc{L}^H_{+0} W_\mu\,.
\end{align}
We now estimate, by \cref{eq:est0}, 
\begin{align*}
     \big\|Y^{(0)}_t\big\|^2_{2,\lad \otimes \si}  & \le \frac{T^2}{\pi^2}\norm{f_0}_{2,\lad}^2 + \sum_\mu \norm{c_+^\mu\psi_{0,\mu}(t) - c_-^\mu \psi_{0,\mu}(T - t)}_{2,\lad}^2  \\
    & \le \frac{T^2}{\pi^2}\norm{f_0}_{2,\lad}^2 + 2 \sum_\mu (|c_+^\mu|^2 + |c_-^\mu|^2) \norm{\psi_{0,\mu}(t)}_{2,\lad}^2 \,,
\end{align*}
and similarly, 
\begin{align*}
     \big\|Y^{(1)}_t\big\|_{2,\lad \otimes \si}^2 & \le \sum_\mu \norm{c_+^\mu\psi_{1,\mu}(t) + c_-^\mu \psi_{1,\mu}(T - t)}^2_{2,\lad} \mu^2 \\
    & \le  2 \mu^2 \sum_\mu (|c_+^\mu|^2 + |c_-^\mu|^2) \norm{\psi_{1,\mu}(t)}_{2,\lad}^2\,.
\end{align*}
Then, thanks to \cref{auxeq22}, we have, by $\mu^2 \ge s^2_H$ and \cref{eq:sourceestimate}, 
\begin{equation}\label{eq:est2}
    \begin{aligned}
     \sum_\mu |c_+^\mu|^2 \norm{\psi_{0,\mu}(t)}_{2,\lad}^2   \le & \sum_\mu |c_+^\mu|^2 \frac{(1-\ell)^2}{30(1+\ell)} \frac{1}{s^2_H}  v_T \\
    \le & \sum_\mu  \frac{(1-\ell)^2}{30(1+\ell)} \frac{1}{s^2_H} \frac{v_T^2}{v_T^2 - \ell^2} \norm{f_\mu}_{2,\lad}^2  \\
    \le &\, \frac{1}{20} \frac{1}{s^2_H} \sum_\mu   \norm{f_\mu}_{2,\lad}^2 \,,
    \end{aligned}
\end{equation}
where the third inequality follows from 
\begin{equation*}
     \frac{(1-\ell)^2}{1+\ell} \frac{v_T^2}{v_T^2 - \ell^2}  \le \frac{3}{2}\,,\q  \mu T > 0\,.
\end{equation*}
Similarly, by \cref{eq:constest,auxeq11}, we have 
\begin{equation}\label{eq:est3}
    \begin{aligned}
     \sum_\mu |c_+^\mu|^2  \mu^2 \norm{\psi_{1,\mu}(t)}_{2,\lad}^2   \le & \sum_\mu |c_+^\mu|^2 \frac{9}{4 s_H^2} v_T \le \sum_\mu \frac{9}{4 s^2_H}  \frac{v_T^2}{v_T^2 - \ell^2} \norm{f_\mu}_{2,\lad}^2  \\
    \le & \, 9 \left(\frac{1}{s_H^4 T^2} + \frac{1}{4 s^2_H} \right) \sum_\mu \norm{f_\mu}_{2,\lad}^2\,.
    \end{aligned}
\end{equation}
The estimate for the part associated with $c_-^\mu$ can be dealt with similarly. We  then conclude
\begin{equation} \label{eq:cond1}
     \big\|Y^{(0)}_t\big\|^2_{2,\lad \otimes \si} \le \max \left\{\frac{1}{5 s^2_H} , \frac{T^2}{\pi^2} \right\} \norm{F_t^h}_{2,\lad \otimes \si}^2\,, \q \big\|Y^{(1)}_t\big\|_{2,\lad \otimes \si}^2 \le 9 \left(\frac{4}{s_H^4 T^2} + \frac{1}{s^2_H} \right) \norm{F_t^h}_{2,\lad \otimes \si}^2\,,
\end{equation}
by $\norm{F_t^h}_{2,\lad \otimes \si}^2 = \norm{f_0}_{2,\lad}^2 + \sum_\mu \norm{f_\mu}_{2,\lad}^2$. In the same manner, we estimate 
\begin{equation*}
     \big\|\p_t Y^{(0)}_t\big\|^2_{2,\lad \otimes \si} \le \norm{f_0}_{2,\lad}^2 + 2 \sum_\mu (|c_+^\mu|^2 + |c_-^\mu|^2) \norm{\psi_{0,\mu}'(t)}_{2,\lad}^2 \,,
\end{equation*}
and 
\begin{equation*}
       \big\|\p_t Y^{(1)}_t\big\|_{2,\lad \otimes \si}^2 \le 2 \mu^2 \sum_\mu (|c_+^\mu|^2 + |c_-^\mu|^2)   \norm{\psi_{1,\mu}'(t)}_{2,\lad}^2\,.
\end{equation*}
By using \cref{auxeq12,auxeq21}, as well as \cref{eq:sourceestimate} with \cref{eq:constest}, we find 
\begin{equation*}  
    \begin{aligned}
       \sum_\mu |c_+^\mu|^2 \norm{\psi_{0,\mu}'(t)}_{2,\lad}^2  
       \le & \sum_\mu    \frac{13}{2}  \frac{v_T^2}{v_T^2 - \ell^2} \norm{f_\mu}_{2,\lad}^2 \\
        \le & \sum_\mu  \left(\frac{13}{2} + \frac{26}{(\mu T)^2}\right) \norm{f_\mu}_{2,\lad}^2\,.
    \end{aligned}
\end{equation*}
and 
\begin{equation*}  
    \begin{aligned}
       \sum_\mu |c_+^\mu|^2 \mu^2 \norm{\psi_{1,\mu}'(t)}_{2,\lad}^2 
       \le & \sum_\mu \frac{12}{T \mu (1 - \ell) v_T}    \frac{v_T^2}{v_T^2 - \ell^2} \norm{f_\mu}_{2,\lad}^2 \\
        \le &\, 36 \sum_\mu  \left( 1 + \frac{5}{(\mu T)^2} + \frac{4}{(\mu T)^4} \right) \norm{f_\mu}_{2,\lad}^2\,,
    \end{aligned}
\end{equation*}
where the last inequality is implied by 
\begin{equation*}
   T \mu (1 - \ell) v_T = \frac{(1-\ell)^2(1+\ell)}{2} \ge \frac{(1-\ell)^2}{2} = \frac{(\mu T)^2}{2}\,,\q \text{as}\q \mu T \to 0\,,
\end{equation*}
and hence 
\begin{equation*}
    \frac{1}{T \mu (1 - \ell) v_T} \le 3 + \frac{3}{(\mu T)^2}\,.
\end{equation*}
Therefore, similarly as for \cref{eq:cond1}, by $\mu^2 \ge s^2_H$, we have 
\begin{equation*}
    \big\|\p_t Y^{(0)}_t\big\|^2_{2,\lad \otimes \si}  \le 52 \left(\frac{1}{2} + \frac{2}{(s_H T)^2} \right) \norm{F_t^h}_{2,\lad \otimes \si}^2\,, 
\end{equation*}
and 
\begin{equation*}
    \big\|\p_t Y^{(1)}_t\big\|_{2,\lad \otimes \si}^2 \le 144 \left( 1 + \frac{5}{(s_H T)^2} + \frac{4}{(s_H T)^4} \right) \norm{F_t^h}_{2,\lad \otimes \si}^2\,.
\end{equation*}
The proof is completed by \cref{eq:est1,eq:devest1}. 
\end{proof}

With the help of the above lemmas, we are ready to give the proof of \cref{thm:tspoincare}. 

\begin{proof}[Proof of \cref{thm:tspoincare}]
We first note, by triangle inequality,   
\begin{equation} \label{eq:initial}
    \norm{X_t - \l X_t \r_{\lad \otimes \si}}_{2,\lad \otimes \si} \le \norm{\Pi_0 X_t - \l X_t \r_{\lad \otimes \si}}_{2,\lad \otimes \si} + \norm{(1 - \Pi_0)X_t}_{2,\lad \otimes \si}\,.
\end{equation}
It suffices to estimate $\norm{\Pi_0 X_t - \l X_t \r_{\lad \otimes \si}}_{2,\lad \otimes \si}$. 
Let $(Y^{(0)}_t, Y^{(1)}_t)$ be a solution constructed in \cref{lem:diverest} for the divergence-type equation \eqref{eq:diver} with source $F_t = \Pi_0 X_t - \l X_t \r_{\lad \otimes \si}$.   
Then, we can compute  
\begin{equation} \label{auxeq:space1}
    \begin{aligned}
          \norm{\Pi_0 X_t - \l X_t \r_{\lad \otimes \si}}_{2,\lad \otimes \si}^2 & = \left\l \Pi_0 X_t, - \p_t Y^{(0)}_t + \left(\mc{L}_{+0}^H\right)^{\star} Y^{(1)}_t \right\r_{2, \lad \otimes \si} \\
    & = \left\l \p_t \Pi_0 X_t, Y^{(0)}_t \right\r_{2, \lad \otimes \si} + \left\l \mc{L}_{+0}^H \Pi_0 X_t, Y^{(1)}_t \right\r_{2, \lad \otimes \si}\\
     & = \left\l \p_t \Pi_0 X_t - \mc{L}^H \Pi_0 X_t, Y^{(0)}_t - Y^{(1)}_t \right\r_{2, \lad \otimes \si}\,,
    \end{aligned}
\end{equation}
where the first equality is from 
\begin{align*}
    \left\l {\bf 1}, - \p_t Y^{(0)}_t + \left(\mc{L}_{+0}^H\right)^{\star} Y^{(1)}_t \right\r_{2, \lad \otimes \si} = 0\,,
\end{align*}
the second equality is by integration by parts for $t$ with $ Y_t^{(0)}|_{t = 0, T} = 0$, while
the third one uses \cref{ass:1} and the orthogonality from $Y^{(0)}_t \in \maf{H}_0$ and $Y^{(1)}_t \in \maf{H}_+$:
\begin{equation*}
\left\l \p_t \Pi_0 X_t,  Y^{(1)}_t \right\r_{\si,1/2} = \left\l \mc{L}_{+0}^H \Pi_0 X_t,  Y^{(0)}_t \right\r_{\si,1/2} = 0\,,\q \forall \ t \in [0,T]\,.
\end{equation*}
We continue to decompose the formula \eqref{auxeq:space1}:
\begin{equation} \label{auxeq:space2}
\begin{aligned}
     & \left\l (\p_t - \mc{L}^H) \Pi_0 X_t, Y^{(0)}_t - Y^{(1)}_t \right\r_{2, \lad \otimes \si} \\ = & \left\l (\p_t - \mc{L}^H)  X_t, Y^{(0)}_t - Y^{(1)}_t \right\r_{2, \lad \otimes \si} - \left\l (\p_t - \mc{L}^H) (1-\Pi_0) X_t, Y^{(0)}_t - Y^{(1)}_t \right\r_{2, \lad \otimes \si},
\end{aligned}
\end{equation}
where the second term can be further reformulated as follows:
\begin{equation} \label{auxeq:space3}
    \begin{aligned}
          \Big\l (\p_t - \mc{L}^H) (1-\Pi_0) X_t, & Y^{(0)}_t - Y^{(1)}_t \Big\r_{2, \lad \otimes \si} \\ & =  \left\l (1-\Pi_0) X_t, \p_t Y^{(1)}_t + \mc{L}_{+0}^H Y^{(0)}_t  - \mc{L}^H Y^{(1)}_t\right\r_{2, \lad \otimes \si}\,.
    \end{aligned}
\end{equation}

We next estimate $\norm{\Pi_0 X_t - \l X_t \r_{\lad \otimes \si}}_{2,\lad \otimes \si}$ with the computations in \cref{auxeq:space1,auxeq:space2,auxeq:space3}. First, thanks to the dissipativity of $\mc{L}^D$, the operator $\beta - \mc{L}^D$ is invertible on $L^2_{\lad \otimes \si}([0,T];\mc{B}(\mc{H}))$ for any $\beta > 0$, which, along with the estimate \eqref{eq:mainest}, implies 
\small
\begin{equation}\label{auxeq:mainest1}
    \begin{aligned}
         & \left|\left\l \big(\p_t - \mc{L}^H \big)   X_t, Y^{(0)}_t - Y^{(1)}_t \right\r_{2, \lad \otimes \si} \right| \\ \le & \Norm{\big(\beta - \mc{L}^D\big)^{-1/2}\big(\p_t - \mc{L}^H \big)  X_t}_{2,\lad\otimes \si} \Norm{\big(\beta - \mc{L}^D \big)^{1/2}\left(Y^{(0)}_t - Y^{(1)}_t\right)}_{2,\lad\otimes \si} 
    \\ \le &\, C_{2, T,\beta} \Norm{\big(\beta - \mc{L}^D\big)^{-1/2}\big(\p_t - \mc{L}^H\big)  X_t}_{2,\lad\otimes \si}  \Norm{\Pi_0 X_t - \l X_t \r_{\lad \otimes \si}}_{2,\lad \otimes \si}\,,
    \end{aligned}
\end{equation}
\normalsize
where $$C_{2, T,\beta} = \sqrt{2} \Norm{(\beta -  \mc{L}^D)^{1/2}}_{2 \to 2} \left(\max\left\{\frac{6}{s_H^2 T} + \frac{3}{s_H}, \frac{T}{\pi} \right\} + \max\left\{\frac{1}{s_H}, \frac{T}{\pi} \right\} \right).$$ Then, again by estimates in \cref{lem:diverest}, we have
\begin{equation} \label{auxeq:mainest2}
    \begin{aligned}
         &\big\| \p_t Y^{(1)}_t + \mc{L}_{+0}^H Y^{(0)}_t  - \mc{L}^H Y^{(1)}_t \big\|_{2,\lad \otimes \si} \\ \le \, & \big\| \p_t Y^{(1)}_t\big\|_{2,\lad \otimes \si} + \big\|\mc{L}_{+0}^H Y^{(0)}_t\big\|_{2,\lad \otimes \si} + \big\|\mc{L}^H Y^{(1)}_t\big\|_{2,\lad \otimes \si}  \\
 \le & \w{C}_{1,T} \Norm{\Pi_0 X_t - \l X_t \r_{\lad \otimes \si}}_{2,\lad \otimes \si}\,,
    \end{aligned}
\end{equation}
with 
\begin{equation*}
   \w{C}_{1,T} = \sqrt{2} + 13 + \frac{12\sqrt{5}}{s_H T} + \frac{24}{(s_H T)^2} +  \sqrt{2} \max\left\{\frac{6}{s_H^2 T} + \frac{3}{s_H}, \frac{T}{\pi} \right\} \Norm{\mc{L}^H \Pi_+}_{\maf{H} \to \maf{H}}\,,
\end{equation*}
where we also used, by similar estimates as \cref{eq:prioriest2,eq:repy0,eq:est2},
\begin{equation*}
    \big\|\mc{L}_{+0}^H Y^{(0)}_t\big\|_{2,\lad \otimes \si} \le \sqrt{2} \Norm{\Pi_0 X_t - \l X_t \r_{\lad \otimes \si}}_{2,\lad \otimes \si}\,.
\end{equation*}
The proof is completed by the following estimate, along with \cref{eq:initial,auxeq:mainest1,auxeq:mainest2},
\begin{multline*}
  \Norm{\Pi_0 X_t - \l X_t \r_{\lad \otimes \si}}_{2,\lad \otimes \si}^2 \le \big|\big\l (\p_t - \mc{L}^H)  X_t, Y^{(0)}_t - Y^{(1)}_t \big\r_{2, \lad \otimes \si} \big| \\ +  \Norm{(1-\Pi_0) X_t}_{2,\lad\otimes\si}   \big\| \p_t Y^{(1)}_t + \mc{L}_{+0}^H Y^{(0)}_t  - \mc{L}^H Y^{(1)}_t \big\|_{2,\lad \otimes \si}\,. \mbox{\qedhere}
\end{multline*}
\end{proof}

We end this section with the proof of \cref{thm:ratest}.

\begin{proof}[Proof of \cref{thm:ratest}]
Note that if $\l X_0\r_\si = 0$, then for any $t > 0$, $\l \mc{P}_t X_0 \r_{\si} = \l X_0 \r_{\si} = 0$. Without loss of generality, we assume $\l X_0\r_\si = 0$. Then, for any $t, T > 0$, we compute 
\begin{align} \label{eq:energyest}
          \frac{\rd}{\rd t} \frac{1}{T}\int_t^{t + T} \norm{X_\tau}_{2,\si}^2 \ud \tau & = \int_0^{T} \frac{\rd}{\rd t} \norm{X_{t + \tau}}_{2,\si}^2 \ud \lad(\tau) \notag \\ & = \int_0^T \l X_{t + \tau}, (\mc{L} + \mc{L}^\star) X_{t + \tau} \r_{\si,1/2} \ud \lad(\tau)  \\
      & = 2 \int_0^T \l X_{t + \tau}, \mc{L}^D X_{t + \tau} \r_{\si,1/2}  \ud \lad(\tau) = - 2 \int_0^T \mc{E}_{\mc{L}^D}(X_{t + \tau}, X_{t + \tau}) \ud \lad(\tau)\,, \notag
\end{align}
due to 
\begin{align*}
    \mc{L} + \mc{L}^\star = 2 \mc{L}^D\,.
\end{align*}
Recalling the coercivity \eqref{eq:coerciveLd}, we have, for any $s > 0$,  
\begin{align*}
    \mc{E}_{\mc{L}^D}(X_s, X_s) \ge \lad_D \norm{(1 - \Pi_0)X_s}_{2,\si}^2\,.
\end{align*}
On the other hand, by the Lindblad equation, 
\begin{equation*}
    \mc{L}^D X_s = (\p_s - \mc{L}^H) X_s\,,
\end{equation*}
which implies, for any $\beta > 0$, 
\begin{align*}
\mc{E}_{\mc{L}^D}(X_s, X_s) & = \left\l \mc{L}^D X_s,  \left(- \mc{L}^D\right)^{-1} \mc{L}^D X_s \right\r_{\si,1/2} \\
& \ge \Norm{\left(\beta - \mc{L}^D\right)^{-1/2} \mc{L}^D X_s}_{2,\si}^2 = \Norm{\left(\beta - \mc{L}^D\right)^{-1/2}\left(\p_s - \mc{L}^H\right) X_s}_{2,\si}^2\,.
\end{align*}
Here $\left(- \mc{L}^D\right)^{-1}$ in the first line is well-defined as the pseudo-inverse of $\mc{L}^D$.  
It follows from \eqref{eq:energyest} that for any $\eta \in (0,1)$, 
\begin{align*}
& \frac{\rd}{\rd t} \frac{1}{T}\int_t^{t + T} \norm{X_\tau}_{2,\si}^2 \ud \tau \\
\le & - 2 \int_0^T \eta \lad_D \Norm{(1 - \Pi_0)X_{t + \tau}}_{2,\si}^2 + (1 - \eta) \Norm{\left(\beta - \mc{L}^D\right)^{-1/2}\left(\p_\tau - \mc{L}^H\right) X_{t + \tau}}_{2,\si}^2 \ud \lad(\tau)\,.
\end{align*}
We take $\eta = \frac{C_{1,T}^2}{C_{1,T}^2 + \lad_D C_{2,T,\beta}^2}$ with constants $C_{1,T}$ and $C_{2,T,\beta}$ given in \cref{thm:tspoincare}, and find 
\begin{align} \label{eq:differenceest}
    \frac{\rd}{\rd t} \frac{1}{T}\int_t^{t + T} \norm{X_\tau}_{2,\si}^2 \ud \tau \le &   -   \frac{2 \lad_D}{C_{1,T}^2 + \lad_D C_{2,T,\beta}^2} \int_0^T C_{1,T}^2  \Norm{(1 - \Pi_0)X_{t + \tau}}_{2,\si}^2 \notag \\  & \qquad \qquad \qquad  \qquad \quad  +  C_{2,T,\beta}^2 \Norm{\left(\beta - \mc{L}^D\right)^{-1/2}\left(\p_\tau - \mc{L}^H\right) X_{t + \tau}}_{2,\si}^2 \ud \lad(\tau)\\
        \le &   -   \frac{\lad_D}{C_{1,T}^2 + \lad_D C_{2,T,\beta}^2} \int_0^T  \norm{X_{t + \tau}}_{2,\si}^2 \ud \lad(\tau)\,.\notag 
\end{align}
Then, by \cref{eq:differenceest} above, a direct application of Gr\"{o}nwall's inequality gives the exponential decay of the time average of $L^2$-distance: 
\begin{align*}
    \frac{1}{T}\int_t^{t + T} \norm{X_\tau}_{2,\si}^2 \ud \tau  \le e^{-\nu t} \frac{1}{T}\int_0^{T} \norm{X_\tau}_{2,\si}^2 \ud \tau\,,
\end{align*}
with 
\begin{align} \label{eq:estconverge}
    \nu =  \frac{\lad_D}{C_{1,T}^2 + \lad_D C_{2,T,\beta}^2}\,. 
\end{align}
The decay estimate in $L^2$-distance can be readily implied by the above time-average one: 
\begin{align} \label{eq:expest}
    \norm{X_{t + T}}_{2,\si}^2 \le \frac{1}{T}\int_t^{t + T} \norm{X_\tau}_{2,\si}^2 \ud \tau  \le e^{-\nu t} \norm{X_0}_{2,\si}^2\,,
\end{align}
since $\norm{X_t}_{2,\si}$ is decreasing in $t$ by the contractivity of QMS. The estimate \eqref{eq:expest} gives 
\begin{align*} 
    \norm{X_{t}}_{2,\si}^2 \le e^{-\nu (t - T)} \norm{X_0}_{2,\si}^2\,,\q \forall\, t \ge T\,,
\end{align*}
which in fact holds for any $t > 0$, by $\norm{X_{t}}_{2,\si} \le \norm{X_0}_{2,\si}$ and $e^{-\nu (t - T)} \ge 1$ for $t \in [0, T]$. 

Noting that the convergence rate estimate \eqref{eq:estconverge}-\eqref{eq:expest} holds for any $\beta > 0$, and it is monotonically decreasing in $\beta$, we can take $\beta \to 0$ to maximize the lower bound. 
Finally, for bounding the prefactor $C_T = e^{\nu T}$, it suffices to consider $\nu T$, which can be easily estimated as follows: 
\begin{align*}
    \nu T \le \frac{\lad_D T}{C_{1,T}^2 + \lad_D C_{2,T,\beta}^2} & \le \frac{\lad_D T \pi^2}{(\sqrt{2} + 14)^2 + 2 \norm{(\mc{L}^D)^{1/2}}_{2 \to 2}^2 \lad_D T^2} \\
    & = \mathcal{O}\left(\frac{\lad_D T}{1 + (\lad_D T)^2} \right) = \mc{O}(1)\,,
\end{align*}
by using \cref{constant1,constant2} and 
\begin{equation*}
    \norm{(\beta - \mc{L}^D)^{1/2}}_{2 \to 2} \ge  \norm{( \mc{L}^D)^{1/2}}_{2 \to 2} = \sup_{X \in \maf{H}\backslash\{0\}} \frac{\norm{(\mc{L}^D)^{1/2} X}_{2,\si}}{\norm{X}_{2,\si}} \ge \sqrt{\lad_D}\,.
\end{equation*}
The proof is complete. 
\end{proof}

\section{Applications and examples} \label{sec:example}

In this section, we apply the results in \cref{sec:spacetime} to some Davies generators with detailed balance broken by a coherent term, namely, the Lindbladian $\mc{L}$ in \eqref{eq:qmslindblad} with $\mc{L}^D$ satisfying $\si$-GNS DBC (noting that up to a coherent term, the Davies generator is exactly the class of $\si$-GNS detailed balanced Lindbladian \cites{kossakowski1977quantum}; see also \cite{ding2024efficient}*{Section 2.1}). 
For the reader's convenience, we first recall the canonical form for the GNS detailed balanced QMS \cite{alicki1976detailed}. 

\begin{lemma}\label{lem:qms_gns}
For a given full-rank state $\si \in \dhh$, a Lindbladian $\mc{L}$ satisfying $\si$-{\rm GNS DBC} is of the following form:
 \begin{align} \label{eq:structure}
        \mc{L}(X) = \sum_{j \in \mc{J}} \left(e^{-\omega_j/2}L_j^\dag[X,L_j] + e^{\omega_j/2}[L_j,X]L^\dag_j\right)\,,
    \end{align}
with $\ww_j \in \R$ and $|\mc{J}| \le N^2 -1$, where jumps $L_j \in \mc{B}(\mc{H})$ satisfy
\begin{align} \label{eq:eigmodular}
    \Delta_{\si}(L_j) =  e^{-\omega_j}L_j\,,\quad \tr(L^\dag_j L_k) = c_j\d_{j,k}\,,\quad \tr(L_j) = 0\,,
\end{align}
for some normalization constants $c_j > 0$. Moreover, for each $j$, there exists $j' \in \mc{J}$ such that 
\begin{align} \label{eq:adjoint_index}
    L_j^\dag = L_{j'}\,,\quad \ww_j = - \ww_{j'}\,.
\end{align}
\end{lemma}

With the help of \cref{lem:qms_gns}, it is easy to design some primitive hypocoercive Lindblad dynamics with broken GNS detailed balance. Specifically, we start with a non-primitive $\mc{L}^D$ of the form \eqref{eq:structure} with $\si$-GNS DBC. From \cref{lem:converg},  its kernel is given by $\ker(\mc{L}^D) = \{X\,;\ [L_j, X] = 0\ \text{for all $j \in \mc{J}$}\}$. In some cases, it is possible to efficiently characterize $\ker(\mc{L}^D)$, allowing the design of a coherent term $\mc{L}^H = i[H,\cdot]$ such that $\mc{L} = \mc{L}^H + \mc{L}^D$ is primitive with \cref{ass:1}. Consequently, the results of \cref{sec:spacetime} can be applied to quantify the convergence rate of $\exp(t \mc{L})$.

\subsection{Unital case} \label{sec:infinite}
Let us first consider the simpler infinite temperature regime, where $\si = \mi/N$ is the maximally mixed state. In this case, the KMS inner product $\l \dd,\dd \r_{\si,1/2}$ reduces to the normalized HS inner product:
\begin{align} \label{eq:productid}
    \l X , Y\r_{\frac{\mi}{N}} := \frac{1}{N} \l X , Y\r\,,
\end{align}
and $\Delta_\si$ becomes the identity operator. This implies 
$\ww_j = 0$ in \eqref{eq:structure} and that the jumps $L_j$ can be taken to be self-adjoint. Thus, the Lindbladian with $\frac{\mi}{N}$-GNS DBC satisfies $\mc{L} = \mc{L}^\dag$ and has the following general form: for some self-adjoint operators $L_j$, 
\begin{align} \label{eq:gen_symm}
    \mc{L}(X) := - \sum_{j \in \mc{J}} [L_j, [L_j, X]]\,.
\end{align}

We shall focus on the dissipative Lindbladian $\mc{L}^D$ with a single jump to illustrate the ideas. The general case \eqref{eq:gen_symm} can be discussed similarly. Let $A$ be a self-adjoint operator on $\mc{H}$ with spectral decomposition: $A = \sum_i \ka_i P_i$, where $P_i$ is the spectral projection associated with the eigenvalue $\ka_i$. We define the non-primitive Lindbladian with jump $A$: 
\begin{align} \label{eq:singlejump}
    \mc{L}^D(X) := - [A, [A, X]] = - \sum_{i,j} (\ka_i - \ka_j)^2 P_i X P_j\,,
\end{align}
which is self-adjoint on $\bh$. The eigenvalues of $\mc{L}^D$ are given by $- (\ka_i - \ka_j)^2$ with the eigenspaces $P_i  \bh P_j $. We hence see the norm and spectral gap of $\mc{L}^D$: 
\begin{align} \label{eq:ep1norm}
 \norm{\mc{L}^D}_{2 \to 2} = \max_{i,j} |\ka_i - \ka_j|^2\,,\q   \lad_D = \min_{i\neq j} |\ka_i - \ka_j|^2\,,
\end{align}
as well as the kernel $\ker(\mc{L}^D) = \{X\,;\ [A, X] = 0\}$. In what follows, for simplicity, we assume that each $\ka_j$ is of multiplicity one, and we denote $P_j  = \ket{j}\bra{j}$.  Then $\ker(\mc{L}^D)$ is characterized by the diagonal matrices on the basis $\ket{i}\bra{j}$:
\begin{equation} \label{exp:kernel1}
    \ker(\mc{L}^D) = \Big\{X\in \bh\,; X = \sum_j x_j \ket{j}\bra{j}\,,\ x_j \in \C \Big\}\,.
\end{equation}

We next characterize the coherent term $\mc{L}^H = i [H, \dd]$ such that $\ker(\mc{L}^H) \bigcap \ker(\mc{L}^D) = {\rm Span}\{\mi\}$, which by \cref{lem:kernel}, implies that $\mc{L} =  \mc{L}^H + \mc{L}^D$ is primitive. We write $H = \sum_{ij} H_{ij} \ket{i}\bra{j}$, and then by \cref{exp:kernel1}, $X \in \ker(\mc{L}^H) \bigcap \ker(\mc{L}^D)$ is equivalent to 
\begin{equation} \label{eq:ep1commu}
   (x_i - x_j) H_{ij} = 0\,.
\end{equation}
The condition $\ker(\mc{L}^H) \bigcap \ker(\mc{L}^D) = {\rm Span}\{\mi\}$ implies that the constant vector $(x_j)_j$ is the only solution to \cref{eq:ep1commu}. It follows that for any two distinct indices $i_1,i_m$, there exists a sequence of indices $(i_1,i_2,\ldots,i_m)$ such that 
\begin{align*}
    \text{$H_{i_{k}i_{k+1}} \neq 0$ for $1 \le k \le m -1$}\,,
\end{align*}
that is, $H_{ij}$ is an irreducible self-adjoint matrix. 

We proceed to characterize $H$ satisfying \cref{ass:1}: $\Pi_0 \mc{L}^H \Pi_0 = 0$, where $\Pi_0$ is the projection to $\ker(\mc{L}^D)$. By our construction, for any diagonal matrices $X = \sum_j x_j \ket{j}\bra{j}$ and $H$, 
\begin{equation*}
  \l X, [H, X] \r_{\frac{\mi}{N}} = 0\,,
\end{equation*}
i.e., $\Pi_0 \mc{L}^H \Pi_0 = 0$ always holds. 

Note that $\maf{H}_0$ in \eqref{eq:decomspace} in this example is parameterized by 
\begin{align*}
  \maf{H}_0 = \Big\{X\in \bh\,; X = \sum_j x_j \ket{j}\bra{j}\,,\ \sum_j x_j = 0 \Big\},
\end{align*}
and there holds
\begin{equation*}
    \Norm{\mc{L}^H \Pi_+}_{\maf{H} \to \maf{H}} \le \lad_{\max}(H) - \lad_{\min}(H)\,.
\end{equation*}
To apply \cref{coro:simplified}, it suffices to estimate 
\begin{align} \label{auxexp2}
    s_H^2 = \inf_{X \in \maf{H}_0 \backslash\{0\}} \frac{ \norm{\mc{L}^H \Pi_0 X}_{\frac{\mi}{N}}^2}{\norm{X}_{\frac{\mi}{N}}^2}\,,
\end{align}
where for $X \in \maf{H}_0$, $\norm{X}_{\frac{\mi}{N}}^2 = \frac{1}{N} \sum_j |x_j|^2$ and
\begin{align} \label{auxexp1}
    \norm{\mc{L}^H \Pi_0 X}_{\frac{\mi}{N}}^2 = \l [H,X], [H, X] \r_{\frac{\mi}{N}} = \frac{1}{N} \sum_{i,j} |x_j - x_i|^2 |H_{ij}|^2\,.
\end{align}
Since the values of $|H_{jj}|^2$ do not influence the sum \eqref{auxexp1}, we define the symmetric matrix $\h{H}$ by 
\begin{equation} \label{def:matrix}
    \h{H}_{ij} := |H_{ij}|^2 - \d_{ij} \sum_{l} |H_{il}|^2\,,
\end{equation}
which is non-negative off the diagonal and still irreducible. Moreover, each row of $\h{H}$ sums to zero. Thus, $\h{H}$ can be viewed as the generator of an ergodic continuous-time Markov chain over the basis $\ket{j}$ and \cref{auxexp2} gives 
\begin{align} \label{eq:shexp1}
     s_H^2 = \inf_{x_j \in \C;\, \sum_j x_j = 0} \frac{\sum_{i,j} |x_j - x_i|^2 \h{H}_{ij}}{\sum_j |x_j|^2} = \inf_{x_j \in \C;\, \sum_j x_j = 0} \frac{-2\big \l x, \h{H} x \big\r }{\l x, x\r}\,,
\end{align}
where $x = (x_1,\ldots,x_N) \in \C^N$. It follows that $s_H/\sqrt{2}$ is the square root of the spectral gap (smallest nonzero eigenvalue) of $-\h{H}$. Summarizing the above discussion and applying \cref{coro:simplified} with \cref{eq:ratewithalpha}, we have the following result. 

\begin{proposition}  \label{propex1}
    Let $\mc{L}^D$ be given as in \eqref{eq:singlejump} for a self-adjoint $A = \sum_j \ka_j \ket{j}\bra{j}$ with eigenvalue $\kappa_j$ of multiplicity one. Then, for a Hamiltonian $H = \sum_{ij} H_{ij} \ket{i}\bra{j}$, the Lindbladian $\mc{L} =  \mc{L}^H + \mc{L}^D$ is primitive if and only if $H_{ij}$ is irreducible. In this case, for $X_0 \in \bh$ and $X_t = e^{t \mc{L}} X_0$, 
\begin{equation*} 
      \norm{X_{t} - \l X_t \r_{\frac{\mi}{N}}}_{\frac{\mi}{N}}^2 \le C e^{-\nu t} \norm{X_0 - \l X_0 \r_{\frac{\mi}{N}}}_{\frac{\mi}{N}}^2\,,\q \forall\, t \ge 0\,, 
\end{equation*}
with $C = \mc{O}(1)$ and convergence rate: 
\begin{align} \label{eq:estexp1}
    \nu & = \frac{ \lad(-\h{H}) \min_{i\neq j} |\ka_i - \ka_j|^2}{ \bigl(28 \lad(-\h{H})^{1/2} + 5 \bigl(\lad_{\max}(H) - \lad_{\min}(H)\bigr) \bigr)^2 + 36  \min_{i\neq j} |\ka_i - \ka_j|^2  \max_{i,j} |\ka_i - \ka_j|^2}\,,
\end{align}
where $\lad(-\h{H})$ denotes the spectral gap of $-\h{H}$ in \eqref{def:matrix}. 
\end{proposition}

Recalling \cref{rem:upplowbb}, it is useful to find a lower bound of $\lad(-\h{H})$ to derive a more concrete estimate. For this, we identify $\h{H}$ with a weighted directed graph $G$ on vertices $j$ with oriented edges $(i,j)$ of weights $\h{H}_{ij}$ connecting $i$ and $j$ 
if $\h{H}_{ij} > 0$. 
Such a graph $G$ uniquely specifies the matrix $\h{H}$ as the graph Laplacian. 
This allows us to estimate the spectral gap of $\h{H}$ (and $s_H$) by the canonical path argument \cites{fill1991eigenvalue,diaconis1991geometric,sinclair1992improved} utilizing the geometry of the underlying graph $G$. Indeed, following \cite{sinclair1992improved}, a canonical path $\gamma_{ij}$ in $G$ is a path from $i$ to $j$, where vertices can be repeated but no edge is overloaded. We denote by $|\gamma_{ij}|$ the length of the path, which is equal to the number of edges included in $\gamma_{ij}$. 
Note that the infimum in \eqref{eq:shexp1} can be taken over real $\{x_j\}$, since $\h{H}$ is Hermitian. Then, by the proof of \cite{sinclair1992improved}*{Theorem 5}, for any $x \in \R^N$ with $\sum_j x_j = 0$ and any collection of canonical paths $\Gamma = \{\gamma_{ij}\}_{i \neq j}$, we have 
\begin{equation} \label{eq:canpatharg}
    \begin{aligned}
         2 \sum_j x_j^2 = \sum_{i \neq j} (x_i - x_j)^2 & = \sum_{i\neq j} \Big(\sum_{e \in \gamma_{ij}} (x_{e^+} - x_{e^-}) \Big)^2 \\
         & \le \sum_{i\neq j} |\gamma_{ij}| \sum_{e \in \gamma_{ij}} (x_{e^+} - x_{e^-})^2 \\
          & =  \sum_{e} (x_{e^+} - x_{e^-})^2 \h{H}_{e^+e^-} \sum_{e \in \gamma_{ij}} \big(\h{H}_{e^+e^-}\big)^{-1} |\gamma_{ij}|\,,
    \end{aligned}
\end{equation}
where $e^{\pm}$ denote the end points of an edge $e = (e^+, e^-)$. We define the quantity
\begin{equation*}
    K_{\h{H},\Gamma} = \max_{e = (e^+, e^-)} \big(\h{H}_{e^+e^-}\big)^{-1} \sum_{e \in \gamma_{ij}} |\gamma_{ij}|\,,
\end{equation*}
depending on $\h{H}$ and the choice of paths $\Gamma$. It follows from \cref{eq:shexp1,eq:canpatharg} that 
\begin{equation} \label{eq:lowershh}
    s_H = \sqrt{2}\lad(- \h{H}) \ge \sqrt{2/ K_{\h{H},\Gamma}}\,.
\end{equation}
Plugging it into \cref{eq:estexp1} gives a more manageable lower bound of the convergence rate $\nu$.   

\begin{remark}
    It is known \cites{sinclair1992improved} that $K_{\h{H},\Gamma}$ measures the "bottleneck" of the graph $G$. If one could choose canonical paths that do not pass any single edge too many times, the quantity $K_{\h{H},\Gamma}$ would be small and we could find a tighter lower for $s_H$ and $\nu$. See \cites{sinclair1992improved,diaconis1991geometric} for examples. 

    We should point out that using $K_{\h{H},\Gamma}$ is just one of many possible approaches to bound the spectral gap of $\h{H}$ and $s_H$. For example, by \cites{diaconis1991geometric}, one can also estimate $s_H$ in terms of
    \begin{equation*}
        \tau_{\h{H},\Gamma}: = \max_e \sum_{e \in \gamma_{ij}} |\gamma_{ij}|_{\h{H}}\q \text{with}\q |\gamma_{ij}|_{\h{H}}: = \sum_{e \in \gamma_{ij}}\big(\h{H}_{e^+e^-}\big)^{-1}\,.
    \end{equation*}
    It is clear that if $\h{H}_{ij}$ is constant (which is the case when $\h{H}$ corresponds to a random walk on a graph), then $\tau_{\h{H},\Gamma} = K_{\h{H},\Gamma}$. However, in general, these two quantities are not comparable. Other techniques for estimating the spectral gap of a classical Markov chain,  
    including Cheeger's inequality, coupling, and stopping times (see e.g., the monograph \cite{levin2017markov}) can also be adapted to our case. 
\end{remark}

Though the model considered above is simple, the underlying analysis can be applied to other practical scenarios, such as some quantum noise models. Motivated by the examples in \cite{fang2024mixing}, we consider an $n$-qubit system with dephasing noise in $Z$ basis:
\begin{equation*}
    \mc{L}^D(X) = \gamma \sum_i (Z_i X Z_i - X)\,,
\end{equation*}
whose eigenvalues are given by $- 2 \gamma k$ with $k = 0, 1, \ldots, n$ that can be determined by analyzing its action on the Pauli basis. It follows that $\mc{L}^D$ has 
the spectral gap $\lad_D = 2 \gamma$ and the kernel
\begin{align*}
\ker(\mc{L}^D) = {\rm Span}\Big\{\prod_i Z_i^{b_i}\,;\ b_i \in \{0,1\}\Big\} = {\rm Span}\{\ket{j}\bra{j}\,;\ \text{$\ket{j}$ is an $n$-bit string}\}\,,
\end{align*}
and the norm estimate holds:
\begin{equation*}
    \norm{\mc{L}^D}_{2 \to 2} = 2 \gamma n\,.
\end{equation*}
Also, by above discussion, $\Pi_0 \mc{L}^H \Pi_0 = 0$ holds for any Hamiltonian $H$, that is, \cref{ass:1} holds, where $\Pi_0$ is the projection to $\ker(\mc{L}^D|_{\maf{H}})$. 

We consider two examples of $H$ as in \cite{fang2024mixing}. 
\begin{example}
    The first example is the quantum walk on a graph. Let $G = (V, E)$ be a $d$-regular connected graph with vertices given by $n$-bit strings and $H$ be the adjacency matrix of $G$. By \eqref{def:matrix} and \eqref{eq:shexp1}, we have 
\begin{equation*}
     s_H^2 = \inf_{x_j \in \C;\, \sum_j x_j = 0} \frac{-2\big \l x, \h{H} x \big\r }{\l x, x\r} = \inf_{x_j \in \C;\, \sum_j x_j = 0} \frac{-2\big \l x, (H - d) x \big\r }{\l x, x\r}\,,
\end{equation*}
where 
\begin{equation*} 
    \h{H}_{ij} := |H_{ij}|^2 - \d_{ij} \sum_{l} |H_{il}|^2 = H_{ij} - d \d_{ij}\,.
\end{equation*}
It is clear that $d - H$ is the graph Laplacian of $G$, and we denote its spectral gap by $\Delta$, which gives $s_H = \sqrt{2 \Delta}$. Also, note that the largest eigenvalue of $H$ is exactly $d$ since $G$ is $d$-regular. We have $ \Norm{\mc{L}^H \Pi_+}_{\maf{H} \to \maf{H}} \le d$ and we conclude the convergence rate estimate by \cref{coro:simplified}: 
\begin{equation} \label{eq:exp1walk}
    \nu \ge \frac{2\Delta \gamma} {(28 \sqrt{\Delta} + 5 d)^2 + 144 \gamma^2 n} = \Theta\left(\frac{\Delta \gamma}{\Delta + d^2 + \gamma^2 n}\right)\,,
\end{equation}
which, as $n \to \infty$, gives an $\mc{O}(n)$ improvement than the rate estimate $\mc{O}(\gamma \Delta/(d + \gamma n)^2)$ derived in \cites{fang2024mixing} under the assumption of $\gamma, \Delta \ll 1$ (our estimate \eqref{eq:exp1walk} does not rely on small $\gamma, \Delta$).  
\end{example}

\begin{example}
The second example is the transverse field Ising model ($h > 0$):
\begin{align*}
    H = \sum_{i = 1}^{n - 1} Z_i Z_{i + 1}  + h \sum_{i = 1}^n X_i\,,
\end{align*}
for which we clearly have 
\begin{equation*}
    \Norm{\mc{L}^H \Pi_+}_{\maf{H} \to \maf{H}} \le \lad_{\max}(H) \le n - 1 + h n\,.
\end{equation*}
To estimate the exponential convergence rate of $\exp(\mc{L}^H + \mc{L}^D)$, it suffices to show that $\mi/2^n$ is the unique invariant state and estimate $s_H$. For unique invariant state, letting $A = \prod_i Z_i^{b_i} \in \ker(\mc{L}^D)$, we consider 
\begin{equation} \label{auxeq:exp1}
    [H, A] = H A - A H = h\sum_i (X_i A - A X_i)\,,
\end{equation}
and find that if $A$ has $Z_i$ at site $i$, then $[X_i, A] = 2 X_i A$, meaning that $[H, A] = 0$ if and only if $b_i = 0$ for all $i$, so that $A \propto \mi$. We next estimate $s_H$ by definition. For $A \in \ker(\mc{L}^D)$, there holds 
\begin{equation} \label{auxeqinf2}
    s_H^2 = \inf_{\vec{b}\in \{0,1\}^n\backslash \{0^n, 1^n\}} \norm{\mc{L}^H_{+0} A}_{\frac{\mi}{2^n}}^2 = 4 h^2\,,
\end{equation}
since $\norm{\prod_i Z_i^{b_i}}_{\frac{\mi}{2^n}} = 1$, where, by \eqref{auxeq:exp1},
\begin{align*}
    \norm{\mc{L}^H_{+0} A}_{\frac{\mi}{2^n}}^2 = \frac{1}{2^n} \l \mc{L}^H A, \mc{L}^H A \r = \frac{1}{2^n} \norm{[H, A]}^2 =\frac{h^2}{2^n} \sum_i \norm{[X_i, A]}^2\,,
\end{align*}
and the infimum \eqref{auxeqinf2} is attained at $A = \prod_i Z_i^{b_i}$ with $b_i = 1$ for only one $i$. Combining the above estimates with \cref{coro:simplified} gives the exponential convergence rate: 
\begin{equation*}
    \nu = \frac{8 \gamma h^2}{(56 h + 5\sqrt{2}(n - 1 + h n ))^2 + 288 \gamma^2 n} = \Theta\left(\frac{\gamma h^2}{(h^2 + 1) n^2 + \gamma^2 n  }\right)\,,
\end{equation*}
which, as $n \to \infty$, is of the same order as $\mc{O}(\gamma h^2/(n^2 \gamma^2 + n^2)$ in \cite{fang2024mixing} derived under $\gamma, h \ll 1$.    
\end{example}

\subsection{Non-unital case} 
We next consider general $\si \in \dhh$ as the invariant state. 
We shall focus on a 
quantum analog of the weighted graph Laplacian with detailed balance condition broken by a coherent term, inspired by the models in \cite{gao2022complete}*{Section 5} and \cite{li2020graph}. 

Let $G = (V, E)$ be a disconnected undirected simple graph\footnote{A graph without loops or multiple edges is called a simple graph.} with vertices $V := \{1,2,\ldots, N\}$ and edge set $E$. Suppose that $G$ has no isolated vertex (i.e., each node $r \in V$ is connected to at least one other node $s \neq r$). For convenience, we identify the node $r$ with the basis vector $\ket{r}$ of $\C^N$ with $1$ at $r$th position and $0$ elsewhere. Without loss of generality, by reordering the vertices, we assume that the connected components of $G = (V, E)$ are given by
\begin{equation} \label{eq:components}
    \{G_i = (V_i, E_i)\}_{i = 1}^m\,, \q 1 \le m \le [N/2]\,,
\end{equation}
with $V_i = (j_{i - 1} + 1, \ldots ,j_{i})$ satisfying $N = \sum_{i = 1}^m |V_i|$. 

Associated with each edge $(r,s) \in E$, we define the jump $e_{rs} = \ket{r}\bra{s}$ and the Lindbladian:
\begin{align*}
   \mc{L}_{e_{rs}}(X) := 2 e_{rs}^\dag X e_{rs} - e_{rs}^\dag e_{rs} X - X e_{rs}^\dag e_{rs}\,.
\end{align*}
Then, we let $\si := \sum_{r = 1}^N \mu_r \ket{r}\bra{r}$ be a full-rank state with 
\begin{align} \label{eq:eigel}
    \mu_1 \ge \mu_2 \ge \cdots \ge \mu_N > 0\,,
\end{align}
and define $e^{\beta_{rs}} = \mu_r/\mu_s$ for $ 1 \le r,s \le N$ and 
\begin{equation*}
\mc{L}_{rs} (X) = e^{\beta_{rs}/2} \mc{L}_{e_{rs}} (X) +  e^{-\beta_{rs}/2} \mc{L}_{e_{sr}} (X)\,.
\end{equation*}
Given a symmetric weight function $w(r,s) = w(s,r) > 0$ over edges $(r,s) \in E$, we define the $\si$-GNS detailed balanced Lindbladian by\footnote{In what follows, the sum over edges $E$ counts both $(r,s) \in E$ and $(s,r) \in E$.}
\begin{align} \label{eq:ldexp2}
    \mc{L}^D = \sum_{(r,s) \in E} w(r,s) \mc{L}_{rs}\,,
\end{align}
noting that each jump $e_{rs}$ is an eigenvector of the modular operator: $ \Delta_{\si} e_{rs} = e^{\beta_{rs}} e_{rs}$, and its adjoint $e_{sr}$ is also a jump. 
It follows that $\si$ is an invariant state of $\mc{L}^D$. 

Now, a direct computation gives for $j \in V$,  
\begin{align}  \label{eq:basismap1}
    \mc{L}^D(e_{jj}) = \sum_{(r,s) \in E} w(r,s) \mc{L}_{rs}(e_{jj}) = 4 \sum_{(r,j) \in E} w(r,j) \left(e^{-\beta_{rj}/2} e_{rr} -  e^{\beta_{rj}/2} e_{jj}  \right),
\end{align}
and for $(j,k) \in E$, 
\begin{align} \label{eq:basismap2}
    \mc{L}^D(e_{jk}) = \kappa_{jk} e_{jk}\q \text{with}\q \kappa_{jk} = - 2 \left(\sum_{(r,j) \in E} w(r, j) e^{\beta_{rj}/2} + \sum_{(r,k) \in E}  w(r, k) e^{\beta_{rk}/2} \right).
\end{align}
We define $\mc{B}_{\rm diag}(\mc{H}): = \{X \in \mc{B}(\mc{H})\,;\ X = \sum_j x_j \ket{j}\bra{j}\}$ by the subspace of diagonal matrices with respect to the basis $\ket{j}$ from the invariant state $\si$, and $\mc{E}_{\rm diag}: \mc{B}(\mc{H}) \to \mc{B}_{\rm diag}(\mc{H})$ by the projection to the diagonal part of an operator $X \in \bh$. From \eqref{eq:basismap1}, we see that the QMS $\mc{P}_t = e^{t \mc{L}^D}$ maps $\mc{B}_{\rm diag}(\mc{H})$ to $\mc{B}_{\rm diag}(\mc{H})$ and its restriction $\mc{P}_t \mc{E}_{\rm diag}$ is a classical reversible random walk generated by a weighted graph Laplacian: for $(i,j) \in E$, 
\begin{align} \label{eq:classgener}
    \mc{L}_{\rm cl}(i,j) = \begin{dcases}
        4  w(i,j) e^{-\beta_{ij}/2}\,,  & (i,j) \in E\,, \\
       - 4 \sum_{\text{fixed $i$},\, (i,r) \in E} w(i,r)e^{-\beta_{ir}/2}\,,   & i = j \in V\,, \\
        0\,, & \text{otherwise}\,, \\
    \end{dcases} 
\end{align}
while \cref{eq:basismap2} shows that $e_{jk}$ is an eigenvector of $\mc{L}^D$. It follows that 
\begin{equation} \label{eq:eigsLd}
    \lad_D = \min\big\{\lad_{\min}(-\mc{L}_{\rm cl})\,, \min_{i \neq j} - \kappa_{ij} \big\}\,,\q \norm{\mc{L}^D}_{2 \to 2} = \max\big\{\lad_{\max}(-\mc{L}_{\rm cl})\,, \max_{i \neq j} - \kappa_{ij} \big\}\,,
\end{equation}
where $\lad_{\min/\max}(-\mc{L}_{\rm cl})$ denotes the min/max eigenvalue of $-\mc{L}_{\rm cl}$. 

By \cref{lem:converg}, we have 
\begin{equation} \label{exp2:kernel}
    \ker(\mc{L}^D) = \{X\in \bh\,; [X, e_{rs}] = [X, e_{sr}] = 0 \ \text{for $(r,s) \in E$} \}\,.
\end{equation}
We will see that $\dim \ker(\mc{L}^D) > 1$, which is a consequence of the disconnectedness of the graph $G$. Recalling the connected components \eqref{eq:components}, for any two nodes $r, s$ in the same component $V_i$, there exists a chain of edges $(r, r_1), (r_1,r_2), \ldots, (r_k, v)$ from $E_i$ such that 
\begin{equation*}
    e_{rv} =  e_{r r_1} e_{r_1r_2}\ldots e_{r_k v}\,,
\end{equation*}
which yields that the $\C$-algebra of $\{e_{rs}, e_{sr}\,; \ (r,s) \in E\}$ is given by the block diagonal matrices with respect to the partition $\{V_i\}_{i = 1}^m$ \eqref{eq:components}. Then, by \cref{exp2:kernel}, $X \in  \ker(\mc{L}^D)$ restricted on each $V_i \times V_i$ is a multiple of identity. It means that 
\begin{equation} \label{eq:kernelrep}
    \ker(\mc{L}^D) = \Big\{X\in \bh\,; X = \sum_{i = 1}^m c_i \sum_{r \in V_i} \ket{r}\bra{r}\,,\ c_i \in \C \Big\}\,,
\end{equation}
and hence the QMS is non-primitive (i.e., $\dim \ker(\mc{L}^D) > 1$).

We proceed to characterize the Hamiltonian $H$ such that $[H, \si] = 0$ and 
$\mc{L} = \mc{L}^H + \mc{L}^D$ is primitive with $\si$ being invariant. 

\begin{lemma} \label{lem:exp21}
    Under the above setting, it holds that 
    \begin{itemize}
        \item  If $\si$ has $N$ distinct eigenvalues, i.e., $\mu_1 > \cdots > \mu_N$ in \eqref{eq:eigel}, then for any Hamiltonian $H$ satisfying $[H, \si] = 0$, the Lindbladian $\mc{L}$ is non-primitive. 
        \item In the case of $\si$ having degenerate eigenvalues, assume that $\w{\mu}_i$ are its distinct eigenvalues with eigenvectors $\ket{r} \in \w{V}_i \subset V$, where $\{\w{V}_i\}$ is a disjoint partition of $V$. Then, a Hamiltonian $H$ satisfies $[H, \si] = 0$ and that $\mc{L}$ is primitive if and only if for any $1 \le i,j \le m$, there exists a chain $(i_0 = i, i_1), (i_1,i_2), \ldots (i_{k-1},i_k = j)$ and $\{\w{V}_{j_l}\}_{l = 1}^k$ such that 
        $H_{ab} \neq 0$ for some 
\begin{equation} \label{eq:matrixprob1}
    (a,b) \in \Big(V_{i_{l-1}} \times V_{i_{l}}\Big) \bigcap \left(\w{V}_{j_l} \times \w{V}_{j_l} \right) \neq \emptyset\,.
\end{equation}
    \end{itemize}    
\end{lemma}

\begin{proof}
For the first statement, if all eigenvalues of $\si$ are distinct, then $H$ satisfying $[H, \si] = 0$ has to be a diagonal matrix with respect to $\ket{r}$, and any $X \in \ker(\mc{L}^D)$ satisfies $[H, X] = 0$. It follows from \cref{lem:kernel} that $\ker(\mc{L}) = \ker(\mc{L}^D)$ and thus $\mc{L}$ is non-primitive. For the second statement, if $\si$ has degenerate eigenvalues, any $H$ satisfying $[H, \si] = 0$ is a block diagonal matrix for the partition $\w{V}_i$. Then $\mc{L}$ is primitive if and only if $X \in \ker(\mc{L}^D)$ such that $[H, X] = 0$ is a multiple of identity. We consider $H = \sum_{rs} H_{rs} \ket{r}\bra{s}$ with $\{H_{rs}\}$ being block diagonal and $X \in \ker(\mc{L}^D)$ of the form 
\begin{equation}\label{eq:formker}
X = \sum_r x_r \ket{r}\bra{r}\q \text{with}\q  (x_1,\ldots,x_N) = (\vec{c}_1, \ldots, \vec{c}_m)\,,\q \vec{c}_i := (c_i, \ldots, c_i) \in \C^{|V_i|}\,.
\end{equation}
Similarly to \cref{eq:ep1commu}, we need to find $H$ such that $(x_r - x_s) H_{rs} = 0$ can imply $x_1 = \ldots  = x_N$,  more precisely, $c_1 = \ldots = c_m$. It follows that for any distinct $1 \le i, j \le m$, there exists a sequence $(i_0 = i,i_1,\ldots,i_k = j)$ such that for each $(i_{l - 1},i_l)$, we can find an edge 
\begin{equation*}
(a,b) \in \Big(V_{i_{l-1}} \times V_{i_{l}}\Big) \bigcap \left(\w{V}_{j_l} \times \w{V}_{j_l} \right),
\end{equation*}
for some $\w{V}_{j_l}$ with $H_{ab} \neq 0$. The proof is complete. 
\end{proof}

We now assume that $\si$ has degenerate eigenvalues such that $(V_{i_{l-1}} \times V_{i_{l}}) \bigcap (\w{V}_{j_l} \times \w{V}_{j_l}) \neq \emptyset$, and fix a Hamiltonian $H = \sum_{r,s = 1}^N H_{rs} \ket{r}\bra{s}$ such that $[H, \si] = 0$ and $\mc{L}$ is primitive (the existence of $H$ is guaranteed by \cref{lem:exp21}). Then $[H, \si^k] = 0$ holds for any $k \in \R$, equivalently, \begin{equation} \label{eq:hsigcom}
    H_{rs} \mu_s^k = \mu_r^k H_{rs}\,.
\end{equation}
One can easily compute for $X = \sum_r x_r \ket{r}\bra{r}\in \ker(\mc{L}^D)$, 
\begin{equation*}
    \si^{1/2}(H X - X H) \si^{1/2} = \sum_{rs} (\mu_r \mu_s)^{1/2} (x_s - x_r) H_{rs} \ket{r}\bra{s}\,,
\end{equation*}
which implies $\l X, [H, X] \r_{\si,1/2}  = 0$, i.e., $\Pi_0 \mc{L}^H \Pi_0 = 0$. Similarly, we have 
\begin{equation*}
    \Norm{\mc{L}^H \Pi_+}_{\maf{H} \to \maf{H}} \le \lad_{\max}(H) - \lad_{\min}(H)\,.
\end{equation*}
We next estimate the singular value gap of $\mc{L}^H \Pi_0$:
\begin{align} \label{eq:shinf}
    s_H^2 = \inf_{X \in \maf{H}_0 \backslash \{0\}} \frac{\norm{\mc{L}^H \Pi_0 X}_{2,\si}^2}{\norm{X}_{2,\si}^2}\,.
\end{align}
With the notion in \eqref{eq:formker}, the subspace $\maf{H}_0$ in \eqref{eq:shinf} is given by those $X = \sum_r x_r \ket{r}\bra{r}$ with 
\begin{equation*}
    \sum_{r = 1}^N \mu_r x_r = \sum_{i = 1}^m \sum_{r \in V_i} \mu_r c_i = \sum_{i = 1}^m \h{\mu}_i c_i = 0\,,
\end{equation*}
where $\{\h{\mu}_i := \sum_{r \in V_i} \mu_r\}$ is a probability distribution over indices $\{i\}$. A direct computation gives 
\begin{equation} \label{eqexp21}
    \norm{X}_{2,\si}^2 = \sum_r \mu_r x_r^2 = \sum_i \h{\mu}_i c_i^2  = \l c, c\r_{2,\h{\mu}}\,, 
\end{equation}
and by \cref{eq:hsigcom}, 
\begin{equation} \label{eq:symme}
    (\mu_r\mu_s)^{1/2} |H_{rs}|^2 = \mu_r |H_{rs}|^2 =  \mu_s |H_{rs}|^2\,,
\end{equation}
and thus
\begin{align*}
    \norm{\mc{L}^H \Pi_0 X}_{2,\si}^2 & = \big\l [H, X], \si^{1/2} [H, X] \si^{1/2}\big\r = \sum_{r \neq s} \mu_r |x_s - x_r|^2 |H_{rs}|^2 \\
    & = \sum_{i \neq j} |c_j - c_i|^2 \h{\mu}_i \h{H}_{ij}\,,
\end{align*}
where 
\begin{equation} \label{eq:matrix1}
   \h{H}_{ij} : =
  \begin{dcases}
      \frac{1}{\h{\mu}_i} \sum_{r \in V_i} \sum_{s \in V_j}   \mu_r |H_{rs}|^2 \,, & i \neq j\,, \\
       - \sum_{j\ \text{with $j \neq i$}} \h{H}_{ij}  \,, & i = j\,.
  \end{dcases}
\end{equation}
By construction, $\h{H}_{ij}$ has non-negative entries for $i \neq j$ with each row summing to zero and satisfies the detailed balance with respect to the probability $\h{\mu}$ by \eqref{eq:symme}:
\begin{equation*} 
    \h{\mu}_i \h{H}_{ij} = \sum_{r \in V_i} \sum_{s \in V_j}   \mu_r |H_{rs}|^2 =  \sum_{s \in V_j}  \sum_{r \in V_i} \mu_s |H_{sr}|^2 = \h{\mu}_j \h{H}_{ji}\,.
\end{equation*}
It follows \cite{norris1998markov} that we can construct a reversible continuous-time Markov chain with $\h{H}$ being the generator and $\h{\mu}$ being an invariant measure. In addition, by \cite{norris1998markov}*{Theorem 3.2.1}, this Markov chain is irreducible if and only if $\h{H}$ is an irreducible matrix, that is, for any $1 \le i,j \le m$, there exists a chain $i_0 = i, i_1, \ldots, i_{k-1},i_k = j$ such that $\h{H}_{i_{l-1}i_l} \neq 0$, which is guaranteed by \cref{eq:matrixprob1,eq:matrix1}. Thus, recalling \eqref{eq:shinf}, we can write 
\begin{align}  \label{eq:shexp2ch}
    s_H^2 = \inf_{c\in \R^m\backslash\{0\}\,, \sum_i \h{\mu}_i c_i = 0} \frac{\sum_{i \neq j} |c_j - c_i|^2 \h{\mu}_i \h{H}_{ij}}{\l c, c\r_{2,\h{\mu}}} = \inf_{c\in \R^m\backslash\{0\}\,, \sum_i \h{\mu}_i c_i = 0} \frac{- 2\l c, \h{H} c \r_{2,\h{\mu}}}{\l c, c\r_{2,\h{\mu}}}\,,
\end{align}
which can be similarly estimated by the canonical path method as in \cref{sec:infinite}. We briefly repeat the argument below for completeness. For a collection of canonical paths $\Gamma = \{\gamma_{ij}\}_{1\le i \neq j \le m}$, we define 
\begin{equation*}
    K_{\h{\mu}, \h{H}, \Gamma} := \max_e \sum_{e \in \gamma_{ij}} \big(\h{\mu}_{e^+} \h{H}_{e^+e^-}\big)^{-1} \h{\mu}_i \h{\mu}_j |\gamma_{ij}|\,.
\end{equation*}
For any $c \in \R^m$ with $\sum_i \h{\mu}_i c_i = 0$ and a collection of paths $\Gamma = \{\gamma_{ij}\}_{i \neq j}$, we have 
\begin{equation*} 
    \begin{aligned}
         2 \sum_i \h{\mu}_i c_i^2  = \sum_{i \neq j} \h{\mu}_i \h{\mu}_j (c_i - c_j)^2 & = \sum_{i\neq j}  \h{\mu}_i \h{\mu}_j \Big(\sum_{e = (e^+, e^-) \in \gamma_{ij}} (c_{e^+} - c_{e^-}) \Big)^2 \\
         & \le \sum_{i\neq j} \h{\mu}_i \h{\mu}_j |\gamma_{ij}| \sum_{e = (e^+, e^-) \in \gamma_{ij}} (c_{e^+} - c_{e^-})^2 \\
          & =  \sum_{e} (c_{e^+} - c_{e^-})^2 \h{\mu}_{e^+} \h{H}_{e^+e^-}  \sum_{e \in \gamma_{ij}} \big(\h{\mu}_{e^+} \h{H}_{e^+e^-}\big)^{-1} \h{\mu}_i \h{\mu}_j |\gamma_{ij}|\,,
    \end{aligned}
\end{equation*}
and then
\begin{align}  \label{eq:shexp2ch2}
  s_H \ge \sqrt{2/ K_{\h{\mu}, \h{H},\Gamma}}\,.
\end{align}

We summarize the above discussion with \cref{coro:simplified} into the following proposition.

\begin{proposition} \label{propex2}
Let $\mc{L} = \mc{L}^H + \mc{L}^D$ be the primitive Lindbladian constructed as above with a unique invariant state $\si$. Then, the exponential decay holds for $X_t = e^{t \mc{L}} X_0$ with $X_0 \in \bh$, 
\begin{equation*} 
      \norm{X_{t} - \l X_t \r_{\si}}_{2,\si}^2 \le C e^{-\nu t} \norm{X_{0} - \l X_0 \r_{\si}}_{2,\si}^2\,,\q \forall\, t \ge 0\,, 
\end{equation*}
with $C = \mc{O}(1)$ and convergence rate: 
\begin{align*}
      \nu = \frac{\lad_D s_H^2}{(28 s_H + 5 \sqrt{2} \bigl(\lad_{\max}(H) - \lad_{\min}(H)\bigr))^2 + 72  \lad_D \big\|\mc{L}^D\big\|_{2 \to 2}}\,,
\end{align*}
where $\lad_D$ and $\norm{\mc{L}^D}_{2 \to 2}$ are characterized by \eqref{eq:eigsLd}, and $s_H$ is given by \eqref{eq:shexp2ch} with lower bound \eqref{eq:shexp2ch2}. 
\end{proposition}

We see that the above discussion crucially relies on the structure of $\mc{L}^D$ in \eqref{eq:basismap1}-\eqref{eq:basismap2}. We next give some examples of $\mc{L}^D$, for which more concrete estimates could be obtained.

\begin{example}
Recalling the decomposition \eqref{eq:components},
suppose that the edges in $E_i$ only connect successive vertices: $E_i = \{(k,k+1)\,;\ j_{i-1} + 1 \le k \le j_i - 1 \}$. We consider the simplest uniform weight $w(r,s) = 1$ for $(r,s) \in E$ and let $\si \propto \sum_{r = 1}^N e^{-\beta r} \ket{r}\bra{r}$ ($\beta > 0$). Then, for this example, $\kappa_{jk}$ in \eqref{eq:basismap2} satisfies
\begin{equation*}
   - \kappa_{jk} \in \left\{8 \cosh\left(\frac{\beta}{2}\right), 4 \cosh\left(\frac{\beta}{2}\right), 4 e^{\frac{\beta}{2}} +  2 e^{-\frac{\beta}{2}}, 2 e^{\frac{\beta}{2}} +  4 e^{-\frac{\beta}{2}}  \right\}\,,
\end{equation*}
 which gives $4 \cosh(\beta/2) \le - \kappa_{jk} \le 8 \cosh(\beta/2)$. Moreover, the generator \eqref{eq:classgener} on the diagonals is a block diagonal matrix for the partition $\{V_i\}$ with each block of the form:
 \[
\mc{L}_{{\rm cl},V_i} =  4 \begin{bmatrix}
  - e^{-\frac{\beta}{2}} & e^{-\frac{\beta}{2}} & 0   & \cdots & 0 \\
    e^{\frac{\beta}{2}} & - e^{-\frac{\beta}{2}}- e^{\frac{\beta}{2}} & e^{-\frac{\beta}{2}} & \cdots & 0 \\
    0   & e^{\frac{\beta}{2}} & - e^{-\frac{\beta}{2}}- e^{\frac{\beta}{2}} & \cdots & 0 \\
    \vdots & \vdots & \vdots & \ddots & \vdots \\
    0   & 0   & 0   & \cdots & - e^{\frac{\beta}{2}}
\end{bmatrix}_{|V_i| \times |V_i|}\,,
\]
which is nothing else than the generator of the classical birth-death process with constant birth and death rates, and can be viewed as a perturbed tridiagonal Toeplitz matrix: 
$$
\mc{L}_{{\rm cl},V_i} = {\rm tridiag}\left(e^{-\frac{\beta}{2}}, - e^{-\frac{\beta}{2}}- e^{\frac{\beta}{2}}, e^{\frac{\beta}{2}}\right) + e^{\frac{\beta}{2}}\ket{1}\bra{1} + e^{-\frac{\beta}{2}}\ket{|V_i|}\bra{|V_i|}\,.
$$ 
Its eigenvalues can be explicitly calculated as follows \cite{yueh2008explicit}: $\mu_1 = 0$ and for $2 \le k \le |V_i|$: 
\begin{equation*}
    \mu_k = - 8 \cosh\left(\frac{\beta}{2}\right)  + 8 \cos \left(\frac{\pi}{|V_i|}(k - 1)\right)
\end{equation*}
Thanks to \eqref{eq:eigsLd}, we have 
\begin{equation*}
    \lad_D = 4 \min\left\{ 2 \cosh\left(\frac{\beta}{2}\right)  - 2 \cos \left(\frac{\pi}{\max_i |V_i|}\right),  \cosh\left(\frac{\beta}{2}\right) \right\},
\end{equation*}
and 
\begin{equation*}
    \norm{\mc{L}^D}_{2 \to 2} = 8 \cosh\left(\frac{\beta}{2}\right)  - 8 \cos \left(\pi\frac{\max_i |V_i| - 1}{\max_i |V_i|}\right).
\end{equation*}
This allows us to give a more concrete convergence rate estimate. We remark that the spectral gap lower bounds for a general GNS-detailed balanced Lindbladian can be found in \cite{temme2013lower}.

\end{example}

\begin{example}
In this example, we show how to construct $\mc{L}^D$ of structures \eqref{eq:basismap1}-\eqref{eq:basismap2} by recent advances in quantum Gibbs samplers \cites{ding2024efficient,chen2024randomized,ramkumar2024mixing}. We first recall the efficient KMS detailed balanced quantum Gibbs samplers constructed in \cites{ding2024efficient}. Let $H_{sys} = \sum_{i} \lad_i \ket{i}\bra{i} \in \C^{N \times N}$ be a system Hamiltonian with eigenvalues $\lad_i$ counting multiplicity, and $\si_\beta \propto \exp(- \beta H_{sys})$ be the associated Gibbs state. For given coupling operator $A \in \bh$ and smooth filtering function $f(\nu)$ with compact support and symmetry:  
\begin{equation} \label{eqf:symm}
     \overline{f(\nu)e^{\beta \nu/4}}= f(-\nu)e^{-\beta \nu/4}\,, 
\end{equation}
see \cite{ding2024efficient}*{Assumption 13}. Define the jump associated with $A$ and $f$:
\begin{equation*}
    L_{A,f} := \sum_{\lad_i,\lad_j} f(\lad_i - \lad_j) A_{ij} \ket{i}\bra{j}\,,\q A_{ij}: = \bra{i}A\ket{j}\,, 
\end{equation*}
and the coherent term:
\begin{align*}
G_{A,f} = - \frac{i}{2} \sum_{\lad_i,\lad_j} \tanh \left(\frac{-\beta(\lad_i - \lad_j)}{4}\right) \bra{i} L_{A,f}^\dag L_{A,f} \ket{j} \ket{i}\bra{j}\,,
\end{align*}
where $L_{A,f}^\dag L_{A,f}$ can be computed as
\begin{equation*}
    L_{A,f}^\dag L_{A,f} = \sum_{\lad_j, \lad_m} \sum_{\lad_i}  \overline{f(\lad_i - \lad_j)} f(\lad_i - \lad_m)  \overline{A_{ij}} A_{i m} \ket{j}\bra{m}\,.
\end{equation*}
We then introduce the corresponding Lindbladian $\mc{L}_{A,f} X := i[G_{A,f},X]+ L_{A,f}^\dag X L_{A,f} - \frac{1}{2}\{L^\dagger_{A,f} L_{A,f}, X\}$. 
Following \cites{chen2024randomized,ramkumar2024mixing}, we shall show that for Haar random unitary $A$, 
\begin{equation*}
    \mc{L} = \E_{A \sim {\rm Haar}}[\mc{L}_{A,f}]
\end{equation*}
has the structure \eqref{eq:basismap1}-\eqref{eq:basismap2}. The expectation of $A$ could be taken over any unitary $1$-design since the following computation only requires the first-moment formula: 
\begin{equation}\label{eq:moment}
\E_{A \sim {\rm Haar}}[\bra{l} A \ket{i}\bra{k} A^\dag \ket{m}] = \E_{A \sim {\rm Haar}}[A_{li}\overline{A_{mk}}] = \frac{1}{N}\d_{ik}\d_{lm}\,.    
\end{equation}
We now compute
\begin{equation*}
    L_{A,f}^\dag X L_{A,f} = \sum_{\lad_j,\lad_m} \sum_{\lad_i,\lad_l} \overline{f(\lad_i - \lad_j)} f(\lad_l - \lad_m) \overline{A_{ij}} A_{lm} \bra{i} X    \ket{l} \ket{j}\bra{m}\,,
\end{equation*}
and by using \eqref{eq:moment} and letting $X_{ij} = \bra{i}X\ket{j}$, 
\begin{align*}
    \E_{A \sim {\rm Haar}}[ L_{A,f}^\dag X L_{A,f}] = \frac{1}{N} \sum_{\lad_j} \sum_{\lad_i} |f(\lad_i - \lad_j)|^2   X_{ii} \ket{j}\bra{j}\,.
\end{align*}
Similarly, we can compute 
\begin{equation*}
    \E_{A \sim {\rm Haar}} [L_{A,f}^\dag L_{A,f}] = \frac{1}{N}\sum_{\lad_j} \sum_{\lad_i} |f(\lad_i - \lad_j)|^2   \ket{j}\bra{j}\,,
\end{equation*}
and hence $ \E_{A \sim {\rm Haar}} [G_{A,f}] = 0$ by $\tanh(0) = 0$. It follows that 
\begin{align} \label{eqqq1}
    \mc{L}(\ket{i}\bra{i}) =  \frac{1}{N} \sum_{\lad_j}  |f(\lad_i - \lad_j)|^2  \ket{j}\bra{j} - |f(\lad_j - \lad_i)|^2  \ket{i}\bra{i}\,,
\end{align}
and for $k \neq l$, 
\begin{align} \label{eqqq2}
    \mc{L}(\ket{k}\bra{l}) = - \frac{1}{2 N} \sum_{\lad_i} \left(|f(\lad_i - \lad_k)|^2 + |f(\lad_i - \lad_l)|^2 \right)\ket{k}\bra{l}\,,
\end{align}
see also \cite{chen2024randomized}*{Theorem 12}. In the case of real $f$, from \eqref{eqf:symm} we can write $f(\nu) = q(\nu) e^{-\beta \nu/4}$ with $q(\nu) = q(-\nu)$. Then \cref{eqqq1,eqqq2} can be rewritten as 
\begin{align*}
     &\mc{L}(\ket{i}\bra{i}) =  \frac{1}{N} \sum_{j} w(i,j) \left(e^{\beta_{ij}/2}  \ket{j}\bra{j} - e^{\beta_{ji}/2}   \ket{i}\bra{i}\right)\,,\\
     & \mc{L}(\ket{k}\bra{l}) = - \frac{1}{2 N} \sum_{i} \left(w(i,k) e^{\beta_{ik}/2} + w(i,l) e^{\beta_{il}/2} \right)\ket{k}\bra{l}\,,
\end{align*}
with symmetric $w(i,j) = |q(\lad_i - \lad_j)|^2$ and $e^{\beta_{ij}} = e^{-\beta (\lad_i - \lad_j)}$, which is of the desired form \eqref{eq:basismap1}-\eqref{eq:basismap2}, up to some constant. This connection allows us to consider the convergence of irreversible QMS for Gibbs states of random Hamiltonians as \cites{chen2024randomized,ramkumar2024mixing}, which we leave for future investigation.
\end{example}

We finally comment the case where $\si$ has all distinct eigenvalues $\mu_1 > \cdots > \mu_N >0$. In this case, by \cref{lem:exp21}, it is impossible to find a Hamiltonian $H$ such that $\mc{L}^H + \mc{L}^D$ is primitive. However, we shall see that the necessary degeneracy of eigenvalues of $\si$ can be easily enforced by lifting $\mc{L}^D$ to an extended product space, and the desired Hamiltonian could be constructed similarly.  

To illustrate the ideas, let us consider a connected graph $G$ instead, implying that the associated Lindbladian $\mc{L}^D$ in \eqref{eq:ldexp2} is primitive. We introduce the extended invariant state: 
\begin{equation*}
    \w{\si}  = \frac{1}{2} \si \otimes \mi_2\,,\q \mi_2 = \mm 1 & 0 \\ 0 & 1 \nn.
\end{equation*} 
Let $A$ be a self-adjoint operator on $\C^2$ with nonzero eigenvalues $a \neq b \in \R$:
\begin{equation*}
    A = a \ket{0}\bra{0} + b\ket{1}\bra{1}\,.
\end{equation*}
We define the lifted $\mc{L}^D$ by 
\begin{align*}
    \w{\mc{L}}^D := \sum_{(r,s) \in E} w(r,s) \w{\mc{L}}_{rs}\,,
\end{align*}
with 
\begin{equation*}
\w{\mc{L}}_{rs} (X) := e^{\beta_{rs}/2} \mc{L}_{e_{rs} \otimes A} (X) +  e^{-\beta_{rs}/2} \mc{L}_{e_{sr} \otimes A} (X)\,.
\end{equation*}
From the construction, there holds $\Delta_{\w{\si} }(e_{rs} \otimes A) =  e^{\beta_{rs}} e_{rs} \otimes A$, and $\w{\mc{L}}^D$ satisfies $\w{\si} $-GNS DBC. Similarly, the kernel of $\w{\mc{L}}^D$ is characterized by 
\begin{align*}
    \ker(\w{\mc{L}}^D) & = \{X\in \mc{B}(\mc{H} \otimes \C^2)\,;\  [X, e_{rs} \otimes A] = [X, e_{sr} \otimes A] = 0 \ \text{for $(r,s) \in E$} \} \\
    & = \Big\{X\in \mc{B}(\mc{H} \otimes \C^2)\,;\  X = \sum_{r \in V} \sum_{i = 0,1} x_{ri} \ket{ri}\bra{ri}\,,\ x_{ri} \equiv c_i \in \C\ \, \text{for $r \in V$}  \Big\}\,,
\end{align*}
that is, $X \in \ker(\w{\mc{L}}^D)$ if and only if under the basis $\{\ket{ri}\bra{sj}\}_{r,s \in V\,, i,j = 0,1}$, $$X = \mi_N \otimes \mm c_0 & 0 \\ 0 & c_1 \nn\,,$$ where $\mi_N$ denotes the $N \times N$ identity matrix. On the other hand, the condition $[H, \w{\si} ] = 0$ gives that $H$ is of the block diagonal form $H = {\rm diag}(H^1, \ldots, H^N)$ under the basis $\{\ket{ri}\bra{sj}\}$, where $H^i$ are self-adjoint $2 \times 2$ matrices. It readily follows that for $X \in \ker(\w{\mc{L}}^D)$, $\mc{L}^H(X) = 0$ is equivalent to that for each matrix $H^i$, there holds $(c_0 - c_1) H^i_{01} = 0$, where $H_{01}^i$ is the $(0,1)$-entry of $H^i$. Then, the convergence rate of the Lindblad dynamics $e^{t \mc{L}}$ can be estimated by similar arguments as those for \cref{propex1,propex2}.  The detailed discussion for the convergence of the \emph{lifted Lindbladian} is beyond the scope of this work and left to future studies.

\subsection*{Acknowledgment} This work is supported in part by 
National Key R$\&$D Program of China Grant No. 2024YFA1016000 (B.L.) and National Science Foundation via award NSF DMS-2309378 (B.L. and J.L.). We appreciate helpful discussions with Hongrui Chen, Zhiyan Ding, Di Fang, Lin Lin, Gabriel Stoltz, Yu Tong, and Lexing Ying.

\subsection*{Data Availability} Data sharing is not applicable to this article as no datasets were generated or analyzed
during the study.

\subsection*{Conflict of interest}

Authors have no conflict of interest to declare.


\begin{bibdiv}
\begin{biblist}

\bib{achleitner2017multi}{article}{
      author={Achleitner, Franz},
      author={Arnold, Anton},
      author={Carlen, Eric~A},
       title={On multi-dimensional hypocoercive bgk models},
        date={2017},
     journal={arXiv preprint arXiv:1711.07360},
}

\bib{albritton2019variational}{article}{
      author={Albritton, Dallas},
      author={Armstrong, Scott},
      author={Mourrat, J-C},
      author={Novack, Matthew},
       title={Variational methods for the kinetic fokker-planck equation},
        date={2019},
     journal={arXiv preprint arXiv:1902.04037},
}

\bib{achleitner2015large}{article}{
      author={Achleitner, Franz},
      author={Arnold, Anton},
      author={St{\"u}rzer, Dominik},
       title={Large-time behavior in non-symmetric fokker-planck equations},
        date={2015},
     journal={Rivista di Matematica della Università di Parma},
      number={1},
       pages={1\ndash 68},
}

\bib{arnold2014sharp}{article}{
      author={Arnold, Anton},
      author={Erb, Jan},
       title={Sharp entropy decay for hypocoercive and non-symmetric fokker-planck equations with linear drift},
        date={2014},
     journal={arXiv preprint arXiv:1409.5425},
}

\bib{arnold2020sharp}{article}{
      author={Arnold, Anton},
      author={Jin, Shi},
      author={W{\"o}hrer, Tobias},
       title={Sharp decay estimates in local sensitivity analysis for evolution equations with uncertainties: from odes to linear kinetic equations},
        date={2020},
     journal={Journal of Differential Equations},
      volume={268},
      number={3},
       pages={1156\ndash 1204},
}

\bib{alicki1976detailed}{article}{
      author={Alicki, Robert},
       title={On the detailed balance condition for non-hamiltonian systems},
        date={1976},
     journal={Reports on Mathematical Physics},
      volume={10},
      number={2},
       pages={249\ndash 258},
}

\bib{adamczak2022modified}{article}{
      author={Adamczak, Radoslaw},
      author={Polaczyk, Bartlomiej},
      author={Strzelecki, Michal},
       title={Modified log-sobolev inequalities, beckner inequalities and moment estimates},
        date={2022},
     journal={Journal of Functional Analysis},
      volume={282},
      number={7},
       pages={109349},
}

\bib{arnold2020propagator}{article}{
      author={Arnold, Anton},
      author={Schmeiser, Christian},
      author={Signorello, Beatrice},
       title={Propagator norm and sharp decay estimates for fokker-planck equations with linear drift},
        date={2022},
     journal={Communications in Mathematical Sciences},
      volume={20},
      number={4},
       pages={1047–1080},
}

\bib{bardet2017estimating}{article}{
      author={Bardet, Ivan},
       title={Estimating the decoherence time using non-commutative functional inequalities},
        date={2017},
     journal={arXiv preprint arXiv:1710.01039},
}

\bib{bardet2021modified}{article}{
      author={Bardet, Ivan},
      author={Capel, Angela},
      author={Lucia, Angelo},
      author={P{\'e}rez-Garc{\'\i}a, David},
      author={Rouz{\'e}, Cambyse},
       title={On the modified logarithmic sobolev inequality for the heat-bath dynamics for 1d systems},
        date={2021},
     journal={Journal of Mathematical Physics},
      volume={62},
      number={6},
       pages={061901},
}

\bib{bernard2022hypocoercivity}{article}{
      author={Bernard, {\'E}tienne},
      author={Fathi, Max},
      author={Levitt, Antoine},
      author={Stoltz, Gabriel},
       title={Hypocoercivity with schur complements},
        date={2022},
     journal={Annales Henri Lebesgue},
      volume={5},
       pages={523\ndash 557},
}

\bib{brannan2022complete}{article}{
      author={Brannan, Michael},
      author={Gao, Li},
      author={Junge, Marius},
       title={Complete logarithmic sobolev inequalities via ricci curvature bounded below},
        date={2022},
     journal={Advances in Mathematics},
      volume={394},
       pages={108129},
}

\bib{blanchard2003decoherence}{article}{
      author={Blanchard, Ph},
      author={Olkiewicz, Robert},
       title={Decoherence induced transition from quantum to classical dynamics},
        date={2003},
     journal={Reviews in Mathematical Physics},
      volume={15},
      number={03},
       pages={217\ndash 243},
}

\bib{breuer2002theory}{book}{
      author={Breuer, Heinz-Peter},
      author={Petruccione, Francesco},
       title={The theory of open quantum systems},
   publisher={Oxford University Press on Demand},
        date={2002},
}

\bib{brigati2023construct}{article}{
      author={Brigati, Giovanni},
      author={Stoltz, Gabriel},
       title={How to construct decay rates for kinetic fokker--planck equations?},
        date={2023},
     journal={arXiv preprint arXiv:2302.14506},
}

\bib{barthel2022superoperator}{article}{
      author={Barthel, Thomas},
      author={Zhang, Yikang},
       title={Superoperator structures and no-go theorems for dissipative quantum phase transitions},
        date={2022},
     journal={Physical Review A},
      volume={105},
      number={5},
       pages={052224},
}

\bib{chatterjee2023spectral}{article}{
      author={Chatterjee, Sourav},
       title={Spectral gap of nonreversible markov chains},
        date={2023},
     journal={arXiv preprint arXiv:2310.10876},
}

\bib{chen2023quantum}{article}{
      author={Chen, Chi-Fang},
      author={Kastoryano, MJ},
      author={Brandao, FGSL},
      author={Gily{\'e}n, A},
       title={Quantum thermal state preparation},
        date={2023},
     journal={arXiv preprint arXiv:2303.18224},
      volume={10},
}

\bib{chen2023efficient}{article}{
      author={Chen, Chi-Fang},
      author={Kastoryano, Michael~J},
      author={Gily{\'e}n, Andr{\'a}s},
       title={An efficient and exact noncommutative quantum gibbs sampler},
        date={2023},
     journal={arXiv preprint arXiv:2311.09207},
}

\bib{chen2024randomized}{article}{
      author={Chen, Hongrui},
      author={Li, Bowen},
      author={Lu, Jianfeng},
      author={Ying, Lexing},
       title={A randomized method for simulating lindblad equations and thermal state preparation},
        date={2024},
     journal={arXiv preprint arXiv:2407.06594},
}

\bib{cao2023explicit}{article}{
      author={Cao, Yu},
      author={Lu, Jianfeng},
      author={Wang, Lihan},
       title={On explicit $l^2$-convergence rate estimate for underdamped langevin dynamics},
        date={2023},
     journal={Archive for Rational Mechanics and Analysis},
      volume={247},
      number={5},
       pages={90},
}

\bib{li2023quantum}{article}{
      author={Chen, Zherui},
      author={Lu, Yuchen},
      author={Wang, Hao},
      author={Liu, Yizhou},
      author={Li, Tongyang},
       title={Quantum langevin dynamics for optimization},
        date={2023},
     journal={arXiv preprint arXiv:2311.15587},
}

\bib{carlen2014analog}{article}{
      author={Carlen, Eric~A},
      author={Maas, Jan},
       title={An analog of the 2-wasserstein metric in non-commutative probability under which the fermionic fokker--planck equation is gradient flow for the entropy},
        date={2014},
     journal={Communications in mathematical physics},
      volume={331},
      number={3},
       pages={887\ndash 926},
}

\bib{carlen2017gradient}{article}{
      author={Carlen, Eric~A},
      author={Maas, Jan},
       title={Gradient flow and entropy inequalities for quantum markov semigroups with detailed balance},
        date={2017},
     journal={Journal of Functional Analysis},
      volume={273},
      number={5},
       pages={1810\ndash 1869},
}

\bib{capel2020modified}{article}{
      author={Capel, {\'A}ngela},
      author={Rouz{\'e}, Cambyse},
      author={Fran{\c{c}}a, Daniel~Stilck},
       title={The modified logarithmic sobolev inequality for quantum spin systems: classical and commuting nearest neighbour interactions},
        date={2020},
     journal={arXiv preprint arXiv:2009.11817},
}

\bib{cuevas2019quantum}{thesis}{
      author={Cuevas, Angela~Capel},
       title={Quantum logarithmic sobolev inequalities for quantum many-body systems: An approach via quasi-factorization of the relative entropy},
        type={Ph.D. Thesis},
        date={2019},
}

\bib{ding2023single}{article}{
      author={Ding, Zhiyan},
      author={Chen, Chi-Fang},
      author={Lin, Lin},
       title={Single-ancilla ground state preparation via lindbladians},
        date={2023},
     journal={arXiv preprint arXiv:2308.15676},
}

\bib{derezinski2006fermi}{incollection}{
      author={Derezi{\'n}ski, Jan},
      author={Fr{\"u}boes, Rafal},
       title={Fermi golden rule and open quantum systems},
        date={2006},
   booktitle={Open quantum systems iii: Recent developments},
   publisher={Springer},
       pages={67\ndash 116},
}

\bib{ding2024efficient}{article}{
      author={Ding, Zhiyan},
      author={Li, Bowen},
      author={Lin, Lin},
       title={Efficient quantum gibbs samplers with kubo--martin--schwinger detailed balance condition},
        date={2024},
     journal={arXiv preprint arXiv:2404.05998},
}

\bib{dupuis2012infinite}{article}{
      author={Dupuis, Paul},
      author={Liu, Yufei},
      author={Plattner, Nuria},
      author={Doll, Jimmie~D},
       title={On the infinite swapping limit for parallel tempering},
        date={2012},
     journal={Multiscale Modeling \& Simulation},
      volume={10},
      number={3},
       pages={986\ndash 1022},
}

\bib{dolbeault2009hypocoercivity}{article}{
      author={Dolbeault, Jean},
      author={Mouhot, Cl{\'e}ment},
      author={Schmeiser, Christian},
       title={Hypocoercivity for kinetic equations with linear relaxation terms},
        date={2009},
     journal={Comptes Rendus. Math{\'e}matique},
      volume={347},
      number={9-10},
       pages={511\ndash 516},
}

\bib{dolbeault2015hypocoercivity}{article}{
      author={Dolbeault, Jean},
      author={Mouhot, Cl{\'e}ment},
      author={Schmeiser, Christian},
       title={Hypocoercivity for linear kinetic equations conserving mass},
        date={2015},
     journal={Transactions of the American Mathematical Society},
      volume={367},
      number={6},
       pages={3807\ndash 3828},
}

\bib{datta2020relating}{article}{
      author={Datta, Nilanjana},
      author={Rouz{\'e}, Cambyse},
       title={Relating relative entropy, optimal transport and fisher information: a quantum hwi inequality},
        date={2020},
     journal={Annales Henri Poincar{\'e}},
      volume={21},
      number={7},
       pages={2115\ndash 2150},
}

\bib{diaconis1991geometric}{article}{
      author={Diaconis, Persi},
      author={Stroock, Daniel},
       title={Geometric bounds for eigenvalues of markov chains},
        date={1991},
     journal={The annals of applied probability},
       pages={36\ndash 61},
}

\bib{desvillettes2001trend}{article}{
      author={Desvillettes, Laurent},
      author={Villani, C{\'e}dric},
       title={On the trend to global equilibrium in spatially inhomogeneous entropy-dissipating systems: The linear fokker-planck equation},
        date={2001},
     journal={Communications on Pure and Applied Mathematics: A Journal Issued by the Courant Institute of Mathematical Sciences},
      volume={54},
      number={1},
       pages={1\ndash 42},
}

\bib{eckmann2000non}{article}{
      author={Eckmann, J-P},
      author={Hairer, Martin},
       title={Non-equilibrium statistical mechanics of strongly anharmonic chains of oscillators},
        date={2000},
     journal={Communications in Mathematical Physics},
      volume={212},
       pages={105\ndash 164},
}

\bib{eberle2024non}{article}{
      author={Eberle, Andreas},
      author={L{\"o}rler, Francis},
       title={Non-reversible lifts of reversible diffusion processes and relaxation times},
        date={2024},
     journal={arXiv preprint arXiv:2402.05041},
}

\bib{engel2000one}{book}{
      author={Engel, Klaus-Jochen},
      author={Nagel, Rainer},
      author={Brendle, Simon},
       title={One-parameter semigroups for linear evolution equations},
   publisher={Springer},
        date={2000},
      volume={194},
}

\bib{evans2022partial}{book}{
      author={Evans, Lawrence~C},
       title={Partial differential equations},
   publisher={American Mathematical Society},
        date={2022},
      volume={19},
}

\bib{franke2010behavior}{article}{
      author={Franke, Brice},
      author={Hwang, C-R},
      author={Pai, H-M},
      author={Sheu, S-J},
       title={The behavior of the spectral gap under growing drift},
        date={2010},
     journal={Transactions of the American Mathematical Society},
      volume={362},
      number={3},
       pages={1325\ndash 1350},
}

\bib{fill1991eigenvalue}{article}{
      author={Fill, James~Allen},
       title={Eigenvalue bounds on convergence to stationarity for nonreversible markov chains, with an application to the exclusion process},
        date={1991},
     journal={The annals of applied probability},
       pages={62\ndash 87},
}

\bib{fang2024mixing}{article}{
      author={Fang, Di},
      author={Lu, Jianfeng},
      author={Tong, Yu},
       title={Mixing time of open quantum systems via hypocoercivity},
        date={2024},
     journal={arXiv preprint arXiv:2404.11503},
}

\bib{ferre2020large}{article}{
      author={Ferr{\'e}, Gr{\'e}goire},
      author={Stoltz, Gabriel},
       title={{Large deviations of empirical measures of diffusions in weighted topologies}},
        date={2020},
     journal={Electronic Journal of Probability},
      volume={25},
      number={none},
       pages={1 \ndash  52},
         url={https://doi.org/10.1214/20-EJP514},
}

\bib{fagnola2007generators}{article}{
      author={Fagnola, Franco},
      author={Umanita, Veronica},
       title={Generators of detailed balance quantum markov semigroups},
        date={2007},
     journal={Infinite Dimensional Analysis, Quantum Probability and Related Topics},
      volume={10},
      number={03},
       pages={335\ndash 363},
}

\bib{fagnola2008detailed}{article}{
      author={Fagnola, Franco},
      author={Umanit{\`a}, Veronica},
       title={Detailed balance, time reversal, and generators of quantum markov semigroups},
        date={2008},
     journal={Mathematical Notes},
      volume={84},
       pages={108\ndash 115},
}

\bib{fagnola2010generators}{article}{
      author={Fagnola, Franco},
      author={Umanit{\`a}, Veronica},
       title={Generators of kms symmetric markov semigroups on symmetry and quantum detailed balance},
        date={2010},
     journal={Communications in Mathematical Physics},
      volume={298},
      number={2},
       pages={523\ndash 547},
}

\bib{frigerio1982long}{article}{
      author={Frigerio, Alberto},
      author={Verri, Maurizio},
       title={Long-time asymptotic properties of dynamical semigroups on $w^*$-algebras},
        date={1982},
     journal={Mathematische Zeitschrift},
      volume={180},
      number={3},
       pages={275\ndash 286},
}

\bib{guo2024designing}{article}{
      author={Guo, Jinkang},
      author={Hart, Oliver},
      author={Chen, Chi-Fang},
      author={Friedman, Aaron~J},
      author={Lucas, Andrew},
       title={Designing open quantum systems with known steady states: Davies generators and beyond},
        date={2024},
     journal={arXiv preprint arXiv:2404.14538},
}

\bib{gao2020fisher}{article}{
      author={Gao, Li},
      author={Junge, Marius},
      author={LaRacuente, Nicholas},
       title={Fisher information and logarithmic sobolev inequality for matrix-valued functions},
organization={Springer},
        date={2020},
     journal={Annales Henri Poincar{\'e}},
      volume={21},
      number={11},
       pages={3409\ndash 3478},
}

\bib{gao2022complete}{article}{
      author={Gao, Li},
      author={Junge, Marius},
      author={LaRacuente, Nicholas},
      author={Li, Haojian},
       title={Complete order and relative entropy decay rates},
        date={2022},
     journal={arXiv preprint arXiv:2209.11684},
}

\bib{GoriniKossakowskiSudarshan1976}{article}{
      author={Gorini, Vittorio},
      author={Kossakowski, Andrzej},
      author={Sudarshan, Ennackal Chandy~George},
       title={Completely positive dynamical semigroups of $n$-level systems},
        date={1976},
     journal={J. Math. Phys.},
      volume={17},
       pages={821\ndash 825},
}

\bib{gao2021complete}{article}{
      author={Gao, Li},
      author={Rouz{\'e}, Cambyse},
       title={Complete entropic inequalities for quantum markov chains},
        date={2022},
     journal={Archive for Rational Mechanics and Analysis},
       pages={1\ndash 56},
}

\bib{grothaus2016hilbert}{article}{
      author={Grothaus, Martin},
      author={Stilgenbauer, Patrik},
       title={Hilbert space hypocoercivity for the langevin dynamics revisited},
        date={2016},
     journal={arXiv preprint arXiv:1608.07889},
}

\bib{guionnet2003lectures}{article}{
      author={Guionnet, Alice},
      author={Zegarlinski, Boguslaw},
       title={Lectures on logarithmic sobolev inequalities},
        date={2003},
     journal={S{\'e}minaire de probabilit{\'e}s XXXVI},
       pages={1\ndash 134},
}

\bib{galkowski2024classical}{article}{
      author={Galkowski, Jeffrey},
      author={Zworski, Maciej},
       title={Classical-quantum correspondence in lindblad evolution},
        date={2024},
     journal={arXiv preprint arXiv:2403.09345},
}

\bib{hairer2009hot}{article}{
      author={Hairer, Martin},
       title={How hot can a heat bath get?},
        date={2009},
     journal={Communications in Mathematical Physics},
      volume={292},
       pages={131\ndash 177},
}

\bib{herau2006hypocoercivity}{article}{
      author={H{\'e}rau, Fr{\'e}d{\'e}ric},
       title={Hypocoercivity and exponential time decay for the linear inhomogeneous relaxation boltzmann equation},
        date={2006},
     journal={Asymptotic Analysis},
      volume={46},
      number={3-4},
       pages={349\ndash 359},
}

\bib{hwang2005accelerating}{article}{
      author={Hwang, Chii-Ruey},
      author={Hwang-Ma, Shu-Yin},
      author={Sheu, Shuenn-Jyi},
       title={{Accelerating diffusions}},
        date={2005},
     journal={The Annals of Applied Probability},
      volume={15},
      number={2},
       pages={1433 \ndash  1444},
         url={https://doi.org/10.1214/105051605000000025},
}

\bib{herau2004isotropic}{article}{
      author={H{\'e}rau, Fr{\'e}d{\'e}ric},
      author={Nier, Francis},
       title={Isotropic hypoellipticity and trend to equilibrium for the fokker-planck equation with a high-degree potential},
        date={2004},
     journal={Archive for Rational Mechanics and Analysis},
      volume={171},
       pages={151\ndash 218},
}

\bib{hormander1967hypoelliptic}{article}{
      author={H{\"o}rmander, Lars},
       title={Hypoelliptic second order differential equations},
        date={1967},
     journal={Acta Mathematica},
      volume={119},
      number={1},
       pages={147\ndash 171},
}

\bib{kochanowski2024rapid}{article}{
      author={Kochanowski, Jan},
      author={Alhambra, Alvaro~M},
      author={Capel, Angela},
      author={Rouz{\'e}, Cambyse},
       title={Rapid thermalization of dissipative many-body dynamics of commuting hamiltonians},
        date={2024},
     journal={arXiv preprint arXiv:2404.16780},
}

\bib{kato2013perturbation}{book}{
      author={Kato, Tosio},
       title={Perturbation theory for linear operators},
   publisher={Springer Science \& Business Media},
        date={2013},
      volume={132},
}

\bib{kastoryano2016quantum}{article}{
      author={Kastoryano, Michael~J},
      author={Brandao, Fernando~GSL},
       title={Quantum gibbs samplers: The commuting case},
        date={2016},
     journal={Communications in Mathematical Physics},
      volume={344},
       pages={915\ndash 957},
}

\bib{kossakowski1977quantum}{article}{
      author={Kossakowski, Andrzej},
      author={Frigerio, Alberto},
      author={Gorini, Vittorio},
      author={Verri, Maurizio},
       title={Quantum detailed balance and kms condition},
        date={1977},
     journal={Communications in Mathematical Physics},
      volume={57},
      number={2},
       pages={97\ndash 110},
}

\bib{kastoryano2011dissipative}{article}{
      author={Kastoryano, Michael~James},
      author={Reiter, Florentin},
      author={S{\o}rensen, Anders~S{\o}ndberg},
       title={Dissipative preparation of entanglement in optical cavities},
        date={2011},
     journal={Physical review letters},
      volume={106},
      number={9},
       pages={090502},
}

\bib{kastoryano2013quantum}{article}{
      author={Kastoryano, Michael~J},
      author={Temme, Kristan},
       title={Quantum logarithmic sobolev inequalities and rapid mixing},
        date={2013},
     journal={Journal of Mathematical Physics},
      volume={54},
      number={5},
       pages={052202},
}

\bib{laracuente2022self}{article}{
      author={LaRacuente, Nicholas},
       title={Self-restricting noise and exponential decay in quantum dynamics},
        date={2022},
     journal={arXiv preprint arXiv:2203.03745},
}

\bib{Lindblad1976}{article}{
      author={Lindblad, Goran},
       title={On the generators of quantum dynamical semigroups},
        date={1976},
     journal={Commun. Math. Phys.},
      volume={48},
       pages={119\ndash 130},
}

\bib{li2020graph}{article}{
      author={Li, Haojian},
      author={Junge, Marius},
      author={LaRacuente, Nicholas},
       title={Graph h{\"o}rmander systems},
        date={2020},
     journal={arXiv preprint arXiv:2006.14578},
}

\bib{li2023interpolation}{article}{
      author={Li, Bowen},
      author={Lu, Jianfeng},
       title={Interpolation between modified logarithmic sobolev and poincare inequalities for quantum markovian dynamics},
        date={2023},
     journal={Journal of Statistical Physics},
      volume={190},
      number={10},
       pages={161},
}

\bib{levin2017markov}{book}{
      author={Levin, David~A},
      author={Peres, Yuval},
       title={Markov chains and mixing times},
   publisher={American Mathematical Soc.},
        date={2017},
      volume={107},
}

\bib{lu2022explicit}{article}{
      author={Lu, Jianfeng},
      author={Wang, Lihan},
       title={On explicit $l^2$-convergence rate estimate for piecewise deterministic markov processes in mcmc algorithms},
        date={2022},
     journal={The Annals of Applied Probability},
      volume={32},
      number={2},
       pages={1333\ndash 1361},
}

\bib{mouhot2006quantitative}{article}{
      author={Mouhot, Cl{\'e}ment},
      author={Neumann, Lukas},
       title={Quantitative perturbative study of convergence to equilibrium for collisional kinetic models in the torus},
        date={2006},
     journal={Nonlinearity},
      volume={19},
      number={4},
       pages={969},
}

\bib{majewski1996quantum}{article}{
      author={Majewski, Adam~W},
      author={Zegarlinski, Boguslaw},
       title={On quantum stochastic dynamics and noncommutative $l_p$ spaces},
        date={1996},
     journal={Letters in Mathematical Physics},
      volume={36},
      number={4},
       pages={337\ndash 349},
}

\bib{neal2004improving}{article}{
      author={Neal, Radford~M},
       title={Improving asymptotic variance of mcmc estimators: Non-reversible chains are better},
        date={2004},
     journal={arXiv preprint math},
        ISSN={0407281/},
}

\bib{norris1998markov}{book}{
      author={Norris, James~R},
       title={Markov chains},
   publisher={Cambridge university press},
        date={1998},
      number={2},
}

\bib{olkiewicz1999hypercontractivity}{article}{
      author={Olkiewicz, Robert},
      author={Zegarlinski, Boguslaw},
       title={Hypercontractivity in noncommutative $l_p$ spaces},
        date={1999},
     journal={Journal of functional analysis},
      volume={161},
      number={1},
       pages={246\ndash 285},
}

\bib{peskun1973optimum}{article}{
      author={Peskun, Peter~H},
       title={Optimum monte-carlo sampling using markov chains},
        date={1973},
     journal={Biometrika},
      volume={60},
      number={3},
       pages={607\ndash 612},
}

\bib{rey2015irreversible}{article}{
      author={Rey-Bellet, Luc},
      author={Spiliopoulos, Konstantinos},
       title={Irreversible langevin samplers and variance reduction: a large deviations approach},
        date={2015},
     journal={Nonlinearity},
      volume={28},
      number={7},
       pages={2081},
}

\bib{rey2015variance}{article}{
      author={Rey-Bellet, Luc},
      author={Spiliopoulos, Konstantinos},
       title={{Variance reduction for irreversible Langevin samplers and diffusion on graphs}},
        date={2015},
     journal={Electronic Communications in Probability},
      volume={20},
      number={none},
       pages={1 \ndash  16},
         url={https://doi.org/10.1214/ECP.v20-3855},
}

\bib{ramkumar2024mixing}{article}{
      author={Ramkumar, Akshar},
      author={Soleimanifar, Mehdi},
       title={Mixing time of quantum gibbs sampling for random sparse hamiltonians},
        date={2024},
     journal={arXiv preprint arXiv:2411.04454},
}

\bib{sinclair1992improved}{article}{
      author={Sinclair, Alistair},
       title={Improved bounds for mixing rates of markov chains and multicommodity flow},
        date={1992},
     journal={Combinatorics, probability and Computing},
      volume={1},
      number={4},
       pages={351\ndash 370},
}

\bib{temme2013lower}{article}{
      author={Temme, Kristan},
       title={Lower bounds to the spectral gap of davies generators},
        date={2013},
     journal={Journal of Mathematical Physics},
      volume={54},
      number={12},
}

\bib{temme2010chi}{article}{
      author={Temme, Kristan},
      author={Kastoryano, Michael~James},
      author={Ruskai, Mary~Beth},
      author={Wolf, Michael~Marc},
      author={Verstraete, Frank},
       title={The $\chi$ 2-divergence and mixing times of quantum markov processes},
        date={2010},
     journal={Journal of Mathematical Physics},
      volume={51},
      number={12},
       pages={122201},
}

\bib{stoltz2010free}{book}{
      author={Tony, Leli\'{e}vre},
      author={Rousset, Mathias},
      author={Stoltz, Gabriel},
       title={Free energy computations: A mathematical perspective},
   publisher={World Scientific},
        date={2010},
}

\bib{villani2009hypocoercivity}{book}{
      author={Villani, C{\'e}dric},
       title={Hypocoercivity},
   publisher={American Mathematical Society},
        date={2009},
      volume={202},
      number={950},
}

\bib{vernooij2023derivations}{article}{
      author={Vernooij, Matthijs},
      author={Wirth, Melchior},
       title={Derivations and kms-symmetric quantum markov semigroups},
        date={2023},
     journal={Communications in Mathematical Physics},
      volume={403},
      number={1},
       pages={381\ndash 416},
}

\bib{verstraete2009quantum}{article}{
      author={Verstraete, Frank},
      author={Wolf, Michael~M},
      author={Ignacio~Cirac, J},
       title={Quantum computation and quantum-state engineering driven by dissipation},
        date={2009},
     journal={Nature physics},
      volume={5},
      number={9},
       pages={633\ndash 636},
}

\bib{wolf5quantum}{article}{
      author={Wolf, Michael~M},
       title={Quantum channels \& operations: guided tour. lecture notes},
        date={2012},
}

\bib{wirth2021complete}{article}{
      author={Wirth, Melchior},
      author={Zhang, Haonan},
       title={Complete gradient estimates of quantum markov semigroups},
        date={2021},
     journal={Communications in Mathematical Physics},
      volume={387},
      number={2},
       pages={761\ndash 791},
}

\bib{wirth2021curvature}{article}{
      author={Wirth, Melchior},
      author={Zhang, Haonan},
       title={Curvature-dimension conditions for symmetric quantum markov semigroups},
        date={2022},
     journal={Annales Henri Poincar{\'e}},
       pages={1\ndash 34},
}

\bib{yueh2008explicit}{article}{
      author={Yueh, Wen-Chyuan},
      author={Cheng, Sui~Sun},
       title={Explicit eigenvalues and inverses of tridiagonal toeplitz matrices with four perturbed corners},
        date={2008},
     journal={the ANZIAM Journal},
      volume={49},
      number={3},
       pages={361\ndash 387},
}

\bib{zhang2024driven}{article}{
      author={Zhang, Yikang},
      author={Barthel, Thomas},
       title={Driven-dissipative bose-einstein condensation and the upper critical dimension},
        date={2024},
     journal={Physical Review A},
      volume={109},
      number={2},
       pages={L021301},
}

\end{biblist}
\end{bibdiv}
\end{document}